\documentclass[11pt]{amsart}

\usepackage{amsmath,amssymb,amsthm,mathrsfs,mathtools}
\usepackage{microtype}
\usepackage{hyperref}

\hypersetup{
  colorlinks=true,
  linkcolor=blue,
  citecolor=blue,
  urlcolor=blue,
  pdftitle={Two-mode dominance and deterministic parameter bias bounds for equatorial Kerr--de~Sitter ringdown},
  pdfauthor={Ruiliang Li}
}

\numberwithin{equation}{section}

\newtheorem{theorem}{Theorem}[section]
\newtheorem{proposition}[theorem]{Proposition}
\newtheorem{lemma}[theorem]{Lemma}
\newtheorem{corollary}[theorem]{Corollary}
\newtheorem{assumption}[theorem]{Assumption}
\theoremstyle{definition}
\newtheorem{definition}[theorem]{Definition}
\theoremstyle{remark}
\newtheorem{remark}[theorem]{Remark}

\newcommand{\Diff}{\operatorname{Diff}}
\newcommand{\Op}{\operatorname{Op}}

\newcommand{\Char}{\operatorname{Char}}
\newcommand{\Res}{\operatorname{Res}}
\newcommand{\dist}{\operatorname{dist}}
\newcommand{\Dom}{\operatorname{Dom}}
\newcommand{\Ran}{\operatorname{Ran}}

\newcommand{\supp}{\operatorname{supp}}

\newcommand{\Id}{\mathrm{Id}}
\newcommand{\Log}{\operatorname{Log}}
\newcommand{\mathbfone}{\mathbf 1}

\title[Equatorial Kerr--de~Sitter ringdown]{Two-mode dominance and deterministic parameter bias bounds for equatorial Kerr--de~Sitter ringdown}

\author{Ruiliang Li}
\address{Tsinghua University, Beijing 100084, China}
\email{lirl23@mails.tsinghua.edu.cn}

\date{\today}


\begin{document}

\begin{abstract}
We study scalar waves on subextremal Kerr--de~Sitter spacetimes in a compact slow-rotation regime and at a fixed overtone index.
Working initially at a fixed cosmological constant $\Lambda>0$ and uniformly for $(M,a)$ in a compact slow-rotation set, using the meromorphic/Fredholm framework for quasinormal modes and a semiclassical equatorial labeling proved in a companion paper, we establish a quantitative two-mode dominance theorem in an equatorial high-frequency package: after exact azimuthal reduction, microlocal equatorial localization, and analytic pole selection by entire localization weights constructed from equatorial pseudopoles, the $k=\pm\ell$ sector signals are each governed by a single quasinormal exponential, up to an explicitly controlled tail and an $\mathcal O(\ell^{-\infty})$ contribution from all other poles.
We then develop a fully deterministic frequency-extraction stability estimate based on time-shift invariance, and combine it with the two-mode dominance result and the companion paper's inverse stability theorem to obtain an explicit parameter bias bound for ringdown-based recovery of $(M,a)$.
Finally, using the companion paper's three-parameter inverse theorem and a damping observable based on the scaled imaginary part of one equatorial mode, we propagate the same deterministic error chain to a local bias bound for recovery of $(M,a,\Lambda)$ on compact parameter sets with $|a|$ bounded away from $0$.
As a further consequence, we obtain a localized pseudospectral stability statement for the equatorial resolvent package, quantifying how large microlocalized resolvent norms enforce proximity to the labeled equatorial poles.
The resulting estimates clarify the conditioning mechanisms (start time, window length, shift step, and detector nondegeneracy) and provide a rigorous PDE-to-data interface for high-frequency black-hole spectroscopy.
\end{abstract}

\maketitle

\section{Introduction}\label{sec:intro}

The late-time response of a perturbed black hole is often modeled, on a finite observation window, by a superposition of a small number
of exponentially damped oscillations,
\begin{equation}\label{eq:intro-ringdown-ansatz}
y(t_*)\ \approx\ \sum_{j=1}^{J} A_j\,e^{-i\omega_j t_*},\qquad t_*\ge t_0,\qquad \Im \omega_j<0,
\end{equation}
where $y(t_*)$ is a time series obtained from an observation of the field and the complex frequencies $\omega_j$ are quasinormal modes
(QNMs).  In applications $y$ is typically \emph{not} the raw waveform: it may incorporate symmetry reduction, phase-space localization,
or other preprocessing designed to isolate a particular dynamical channel before attempting a low-rank fit.

From a deterministic PDE viewpoint, using \eqref{eq:intro-ringdown-ansatz} for inference leads to two intertwined stability issues.

\smallskip
\noindent\textbf{Dominance and leakage.}
Even when a meromorphic resolvent and a resonance expansion are available, the time-domain representation is typically an \emph{infinite}
sum with a remainder.  A finite-mode model on a prescribed window requires quantitative dominance of the selected poles together with
quantitative control of leakage from all other poles and from the tail term on a shifted contour.

\smallskip
\noindent\textbf{Nonnormality and extraction conditioning.}
The stationary problem is non-selfadjoint; residue projectors can be large (often called \emph{excitation factors}), and pseudospectral and transient
effects may interfere with naive truncations or fits that do not separate the relevant phase-space channels.
Moreover, extracting complex frequencies from finite-time data is ill-conditioned in general and must be analyzed together with
branch selection for the complex logarithm; this is classical in the Prony/matrix-pencil literature
\cite{HuaSarkarMPM1990,RoyKailath1989ESPRIT,PottsTascheSigPro2010,PottsTascheLAA2013,BatenkovYomdinSJAM2013}.
Recent work on black-hole ringdown has highlighted these mechanisms from several complementary angles, including pseudospectral
analyses, non-modal dynamics, and transients
\cite{JaramilloMacedoAlSheikh2021PRX,CarballoWithers2024TransientDynamics,CarballoPantelidouWithers2025NonModal,
BessonCarballoPantelidouWithers2025TransientsReview,CaiCaoChenGuoWuZhou2025KerrPseudospectrum}.
Closely related phenomena arise near avoided crossings and exceptional points, where resonant excitation can amplify or reshuffle
the observed mode hierarchy \cite{MacedoKatagiriKubotaMotohashi2025EP,Motohashi2025ResonantExcitation}.

\medskip
\noindent\textbf{Setting and scope; two-mode dominance in the equatorial package.}
We work on subextremal Kerr--de~Sitter spacetimes with fixed cosmological constant $\Lambda>0$ and slow rotation $|a|\le a_0$.
In this setting the stationary resolvent is meromorphic, and contour deformation yields resonance expansions for wave solutions within
the Fredholm/microlocal framework developed for asymptotically hyperbolic and Kerr--de~Sitter geometries
\cite{VasyKerrDeSitter,DyatlovQNM,PetersenVasyAnalyticity,PetersenVasyKerrDeSitterJEMS}.
We isolate an \emph{equatorial high-frequency package}: we fix an overtone index $n\in\mathbb N_0$ and consider large angular
momentum $\ell$ in the exact azimuthal sectors $k=\pm\ell$, using $h_\ell=\ell^{-1}$ as semiclassical parameter.

A key point of this paper is that, within this package, one can justify a two-mode model \emph{after} an explicit sequence of
localizations: exact azimuthal projection to $k=\pm\ell$, microlocal equatorial localization, and analytic pole selection by an entire
weight.  Our two-mode dominance theorem is therefore \emph{not} a statement about a generic unfiltered ringdown waveform.
Rather, it is a quantitative statement about two \emph{refined} signals (one for each of $k=\pm\ell$) obtained by preprocessing, in a
regime where the relevant QNMs admit a semiclassical labeling and can be isolated with uniform error bounds.
More precisely, the dominance statement concerns the mode-separated, microlocalized signals obtained by applying the exact azimuthal
projectors $\Pi_{k_\pm}$ and the equatorial cutoffs $A_{\pm,h_\ell}$ to the solution, together with the entire weights $\widetilde g_{\pm,\ell}$
inserted in the inverse Laplace representation; see Definition~\ref{def:analytic-window}.

This paper is the second part of a series.
The companion paper \cite{LiPartI} proves the semiclassical equatorial QNM labeling in the above package and establishes a quantitative
local inversion map from the pair of equatorial frequencies $(\omega_{+,\ell},\omega_{-,\ell})$ to the parameters $(M,a)$.
The purpose of the present paper is to provide the time-domain mechanism needed to turn that frequency information into a
\emph{deterministic} inference procedure with explicit error bookkeeping on a finite window.

\medskip
\noindent\textbf{Strategy: microlocal/analytic mode selection plus deterministic extraction.}
Two ingredients are essential.
First, we separate the relevant phase-space channel \emph{before} fitting a low-rank time series: we perform exact azimuthal projection
to $k=\pm\ell$ and apply equatorial microlocal cutoffs.
Second, we select a single overtone uniformly in $\ell$ by multiplying the resolvent integrand with an \emph{entire} weight
$\widetilde g_{\pm,\ell}(\omega)$ constructed from equatorial pseudopoles provided by \cite{LiPartI}.
This analytic window is designed to be compatible with contour deformation: it suppresses all poles except the labeled one while having
controlled growth on a shifted contour, yielding a windowed resonance expansion with an explicitly controlled tail term.

A related approach---isolating individual QNM contributions using a mode prior---appears in recent data-analysis work under the name
QNM filters \cite{MaEtAl2022QNMFilters}.  The common theme is that, in a regime where QNM expansions are meaningful, prior frequency
information can be used to enhance a target mode.  The key distinction here is that the weights are chosen to be entire (indeed,
polynomial) and to interact well with semiclassical microlocalization, so that leakage from non-target poles can be controlled
\emph{quantitatively} and uniformly in $\ell$.  The dependence on a prior is therefore explicit: our results quantify how windowing,
tail suppression, and extraction errors propagate in a deterministic chain, and the framework naturally accommodates iterative use of a
coarse prior followed by refinement.

The second ingredient is a deterministic stability theory for extracting a complex frequency from a one-mode model with a controlled
remainder.  We analyze a time-shift invariant estimator (a one-step matrix-pencil/Prony-type construction) in an abstract Hilbert space
setting and use the pseudopole prior to fix the logarithm branch.  This yields an extraction error bound that can be combined directly
with the PDE remainder.

\medskip\noindent\textbf{Main results.}
The main conclusions of the paper can be summarized as follows.

\begin{itemize}
\item \emph{Analytically windowed resonance expansion and two-mode dominance.}
After azimuthal reduction to $k=\pm\ell$, equatorial microlocal localization, and analytic pole selection by the weights
$\widetilde g_{\pm,\ell}$, the resulting signals are each governed by a single QNM exponential,
up to an explicitly controlled shifted-contour tail and an $\mathcal O(\ell^{-\infty})$ contribution from all other poles;
see Theorem~\ref{thm:two-mode-dominance}, based on the windowed expansion of Theorem~\ref{thm:windowed-expansion} and the equatorial
pole/sectorization inputs of \cite{LiPartI}.

\item \emph{Deterministic stability of one-mode frequency extraction.}
For a one-mode Hilbert-valued signal $y(t_*)=A e^{-i\omega t_*}+r(t_*)$ on a finite window, the time-shift estimator together with a
fixed logarithm branch recovers $\omega$ with an explicit error bound in terms of the residual size; see
Theorem~\ref{thm:one-mode-extraction} and its equatorial specialization in Corollary~\ref{cor:equatorial-extraction}.

\item \emph{Deterministic parameter bias bound.}
Combining the preceding steps with the inverse stability theorem of \cite{LiPartI}, we obtain an explicit deterministic estimate for the bias of
the recovered parameters $(M,a)$ in terms of the PDE remainder, the extraction window parameters, and the semiclassical conditioning
of the equatorial parameter map; see Theorem~\ref{thm:param-bias}.

\item \emph{Three-parameter deterministic bias bound when $\Lambda$ varies.}
On compact three-parameter sets with $|a|$ bounded away from $0$, the companion paper proves that adding one damping observable
$\widetilde W_\ell:=-(n+\tfrac12)^{-1}\Im\omega_{+,\ell}$ yields a locally invertible three-parameter data map
$(U_\ell,V_\ell,\widetilde W_\ell)\mapsto(M,a,\Lambda)$.
We show that the same time-domain extraction and tail estimates propagate through this three-parameter inverse map, producing an explicit
deterministic bias bound for recovery of $(M,a,\Lambda)$; see Theorem~\ref{thm:param-bias-3}.
\end{itemize}

\medskip\noindent\textbf{Relation to inverse and spectral results.}
A useful mathematical point of comparison is the inverse theorem of Uhlmann--Wang \cite{UhlmannWang2023BHMass}, who recover the mass
of a de~Sitter--Schwarzschild black hole from a single QNM with a H\"older-type stability estimate.
Our setting is genuinely non-selfadjoint and involves rotation: we use two equatorial modes to recover $(M,a)$ and propagate
\emph{time-domain} PDE errors through a concrete extraction map, obtaining a deterministic bias bound on a prescribed window.
On the spectral side, hyperboloidal and Keldysh-type schemes make biorthogonality and spectral projection mechanisms explicit and
provide new perspectives on resonance expansions \cite{MacedoZenginoglu2024Hyperboloidal,BessonJaramillo2025Keldysh}.
Relatedly, there has been recent progress on convergent complete mode decompositions which go beyond late-time QNM sums
\cite{ArnaudoCarballoWithers2025CompleteModes}.
The present paper does not aim at a globally complete mode decomposition.
Instead, it isolates a semiclassical Kerr--de~Sitter package in which a two-mode model can be justified with explicit constants and
then used for deterministic inverse bounds.

\medskip\noindent\textbf{Organization of the paper.}
Section~\ref{sec:setup2} introduces the Kerr--de~Sitter setup and the functional framework.
Section~\ref{sec:resolvent-shifted} proves uniform resolvent bounds on shifted contours.
Section~\ref{sec:resonant-expansion} establishes the resonance expansion with remainder.
Section~\ref{sec:two-mode} develops analytic pole selection and proves two-mode dominance.
Section~\ref{sec:freq-extraction} studies deterministic frequency extraction from finite-time data, and
Section~\ref{sec:application} propagates the resulting frequency errors through the inverse map of \cite{LiPartI} to obtain
parameter bias bounds.
Section~\ref{sec:pseudospectral} proves a microlocalized pseudospectral resolvent bound in the same equatorial package,
quantifying how far the equatorial resolvent can grow away from the labeled poles.
Section~\ref{sec:discussion} discusses further directions.  The appendices collect analytic lemmas for pole
selection and contour subtraction, precise companion-paper inputs, the dual-state interpretation of amplitudes, and the
two-exponential conditioning analysis.

\section{Geometric/PDE setup and functional framework}\label{sec:setup2}

In Sections~\ref{sec:setup2}--\ref{sec:freq-extraction} we fix a cosmological constant $\Lambda>0$ (suppressing it from the notation)
and work on subextremal Kerr--de~Sitter spacetimes $(\mathcal M_{M,a},g_{M,a})$ with $(M,a)$ ranging in a compact slow-rotation parameter set $\mathcal K$.
In Section~\ref{subsec:application-3} we extend the final bias bounds to compact three-parameter sets $(M,a,\Lambda)$; the required uniformity of the forward estimates
with respect to $\Lambda$ on compact sets is summarized in Appendix~\ref{app:lambda-uniformity}.
The aim of this section is to set up a time-domain framework in which one can
write forward solutions via an inverse Laplace transform, reduce exactly to azimuthal symmetry sectors,
insert microlocal equatorial cutoffs, and then convert resolvent poles (QNMs) into ringdown terms
by contour deformation.

While most ingredients are standard in the Kerr--de~Sitter Fredholm framework
\cite{VasyKerrDeSitter,DyatlovQNM,PetersenVasyKerrDeSitterJEMS}, we record the conventions and normalizations
we will use later; in particular we are careful about the underlying spatial manifold.
Since we work uniformly for $(M,a)$ in a fixed compact slow--rotation set $\mathcal K$, all auxiliary choices
(such as the buffer size $\delta$ used below to extend across the \emph{future} horizons) are made once and for all,
with constants uniform on $\mathcal K$.

\subsection{Spacetime region, regular time coordinate, and a fixed spatial slice}\label{subsec:geom-region}

Let $(\mathcal M_{M,a},g_{M,a})$ be a subextremal Kerr--de~Sitter spacetime with parameters $(M,a)$ and fixed
$\Lambda>0$. Denote by $r_e(M,a)$ and $r_c(M,a)$ the (future) event and cosmological horizon radii
($0<r_e<r_c$), and let
\[
\mathcal M^\circ_{M,a}:=\mathbb R_t\times (r_e,r_c)_r\times \mathbb S^2_{\theta,\varphi}
\]
be the domain of outer communications in Boyer--Lindquist coordinates.

\medskip

\noindent\textbf{Physical region versus analytic extension.}
The physically relevant spatial region is
\begin{equation}\label{eq:phys-region}
X^{\mathrm{phys}}_{M,a}:=(r_e,r_c)\times \mathbb S^2_{\theta,\varphi_*},
\qquad
\mathcal M^{\mathrm{phys}}_{M,a}:=\mathbb R_{t_*}\times X^{\mathrm{phys}}_{M,a}.
\end{equation}
All initial data and all observation/measurement cutoffs considered in this paper are assumed to be supported in a
fixed compact subset of $X^{\mathrm{phys}}_{M,a}$ which is uniformly separated from the horizons.
For the stationary Fredholm theory and the redshift estimates it is, however, convenient to work on a slightly
extended region across the \emph{future} horizons; this extension is a technical device and will never enter the
cutoff quantities $\chi R(\omega)\chi$ used in the contour argument.

\medskip

\noindent\textbf{Regular time coordinate and extension across future horizons.}
As is standard, we pass to a future-regular stationary time coordinate $t_*$ and a corresponding azimuthal
coordinate $\varphi_*$ (Eddington--Finkelstein--type ``star'' coordinates) so that
the metric is smooth across the future event and cosmological horizons; see, for instance,
\cite[\S2]{DyatlovQNM} and \cite[\S2]{PetersenVasyKerrDeSitterJEMS}.
Concretely, we assume:
\begin{itemize}
\item $g_{M,a}$ is smooth across the hypersurfaces $r=r_e$ and $r=r_c$ in the $(t_*,r,\theta,\varphi_*)$ chart;
\item $T:=\partial_{t_*}$ and $\Phi:=\partial_{\varphi_*}$ are Killing fields;
\item $t_*$ is a global time function on a slightly extended region across the \emph{future} horizons.
\end{itemize}
We fix a small buffer $\delta>0$ and define the extended spatial region
\begin{equation}\label{eq:Xdef}
X_{M,a}^\circ:=(r_e-\delta,\ r_c+\delta)\times \mathbb S^2_{\theta,\varphi_*},
\qquad
\overline X_{M,a}:=[r_e-\delta,\ r_c+\delta]\times \mathbb S^2_{\theta,\varphi_*},
\qquad
\mathcal M_{M,a}:=\mathbb R_{t_*}\times X_{M,a}^\circ.
\end{equation}
We write $\Sigma_{M,a}:=\{t_*=0\}\subset\mathcal M_{M,a}$ for the reference Cauchy hypersurface.

\medskip

\noindent\textbf{Compactified spatial slice and cutoff insensitivity.}
Several parts of the stationary Fredholm theory are naturally formulated on a compactified spatial manifold with
boundary obtained by closing the $r$--interval in \eqref{eq:Xdef}. We therefore regard $X_{M,a}^\circ$ as the interior
of the compact manifold
\begin{equation}\label{eq:Xbar-convention}
X_{M,a}:=\overline X_{M,a}=[r_e-\delta,\ r_c+\delta]\times \mathbb S^2_{\theta,\varphi_*}.
\end{equation}
All Cauchy problems are posed on $\mathcal M_{M,a}=\mathbb R_{t_*}\times X_{M,a}^\circ$, but all Sobolev spaces are taken on
$X_{M,a}$ (for instance by extending across $\partial X_{M,a}$ in local charts).
The boundary $\partial X_{M,a}=\{r=r_e-\delta\}\cup\{r=r_c+\delta\}$ is \emph{artificial} and lies outside the physical region
$X^{\mathrm{phys}}_{M,a}$.
Consequently, as long as a cutoff $\chi\in C_c^\infty(X_{M,a}^\circ)$ is supported strictly inside $X^{\mathrm{phys}}_{M,a}$,
the operator $\chi R(\omega)\chi$ depends only on the geometry in the physical domain and is insensitive to the chosen
completion near $\partial X_{M,a}$.

\medskip

\noindent\textbf{Uniformizing the underlying manifold (optional, but convenient).}
When varying parameters $(M,a)$ in a compact set $\mathcal K$ inside a slow-rotation subextremal range,
one may identify the family of compactified manifolds $X_{M,a}$ with a single fixed manifold
$X=[r_-,r_+]\times\mathbb S^2$ by an affine reparameterization of the radial variable
that sends $r_e(M,a)-\delta$ and $r_c(M,a)+\delta$ to $r_-$ and $r_+$, respectively.
This is a routine device (we use it in the companion paper), and it allows all Sobolev norms to be taken
on a fixed $X$ while keeping constants uniform on $\mathcal K$.
To keep notation light, we will not distinguish notationally between $X_{M,a}$ and such an identified copy.

\subsection{Parameter regime and geometric constants}\label{subsec:param-regime}

For a fixed cosmological constant $\Lambda>0$, Kerr--de~Sitter metrics are parametrized by the mass $M>0$ and the rotation
parameter $a\in\mathbb R$.  In Boyer--Lindquist form, the radial function is
\begin{equation}\label{eq:Delta-r}
\Delta_r(r;M,a):=(r^2+a^2)\Bigl(1-\frac{\Lambda r^2}{3}\Bigr)-2Mr.
\end{equation}
We write $\mathcal P_\Lambda$ for the \emph{subextremal} parameter set, i.e.\ those $(M,a)$ for which $\Delta_r$ has (at
least) two distinct positive simple roots $0<r_e(M,a)<r_c(M,a)$ corresponding to the event and cosmological horizons.  We
will work uniformly on compact subsets of a slow-rotation subextremal range: fix $a_0>0$ and a compact set
\begin{equation}\label{eq:K-params}
\mathcal K\Subset \mathcal P_\Lambda\cap\{|a|\le a_0\}.
\end{equation}
All implicit constants in the sequel are allowed to depend on $\Lambda$ and on $\mathcal K$ (as well as on the fixed cutoff
$\chi$ used to localize away from the horizons), but are uniform in the high-frequency parameter $\ell\to\infty$.

\medskip

\noindent\textbf{Photon-sphere frequency scale (Schwarzschild--de~Sitter).}
When $a=0$, the trapped null geodesics concentrate on the photon sphere $r_{\mathrm{ph}}=3M$, and the corresponding
coordinate-time orbital frequency is
\begin{equation}\label{eq:Omega-ph}
\Omega_{\mathrm{ph}}(M):=\frac{\sqrt{1-9\Lambda M^2}}{3\sqrt 3\,M}.
\end{equation}
This quantity appears as the leading real-frequency coefficient in high-$\ell$ equatorial QNM asymptotics in the
Schwarzschild--de~Sitter limit; see \cite[\S2]{LiPartI}.  We include \eqref{eq:Omega-ph} mainly as a convenient frequency
scale when discussing uniformity on parameter sets (Remark~\ref{rem:shrinkK}); none of the main proofs below require the
explicit closed form.

\medskip

\noindent\textbf{Overtone index.}
Throughout, we fix an overtone index $n\in\mathbb N_0$.
All results in the two-mode regime are uniform in $\ell\to\infty$ for this fixed $n$.

\subsection{The wave operator and a convenient splitting}\label{subsec:operator-splitting}

Let
\begin{equation}\label{eq:Pwave}
P := \Box_{g_{M,a}} + V
\end{equation}
where $V\in C^\infty(\mathcal M_{M,a})$ is stationary and axisymmetric:
\[
T(V)=0,\qquad \Phi(V)=0.
\]
(For most of the paper one may take $V\equiv 0$; allowing $V$ costs nothing and makes the formalism
stable under lower-order perturbations.)

\medskip

\noindent\textbf{Stationary splitting.}
Since coefficients of $P$ are independent of $t_*$, one can write $P$ in the form
\begin{equation}\label{eq:splitting}
P = \partial_{t_*}^2 + Q\,\partial_{t_*} + L,
\end{equation}
where $Q\in \Diff^1(X_{M,a})$ and $L\in\Diff^2(X_{M,a})$ are differential operators with coefficients
independent of $t_*$.  Concretely, $Q$ encodes the mixed $t_*$--$x$ derivatives and possible first-order
$t_*$--terms, while $L$ is a purely spatial second-order operator (elliptic on $\Sigma_{M,a}$).

\begin{lemma}[Existence of the splitting]\label{lem:splitting}
Let $P=\Box_g+V$ on $\mathcal M_{M,a}$ with $g$ stationary in $t_*$.
Then there exist $t_*$--independent operators $Q\in\Diff^1(X_{M,a})$ and $L\in\Diff^2(X_{M,a})$
such that \eqref{eq:splitting} holds.
Moreover, $Q$ and $L$ commute with $\Phi=\partial_{\varphi_*}$.
\end{lemma}

\begin{proof}
Write $P$ in local coordinates $(t_*,x)$ on $\mathcal M_{M,a}$:
\[
P=\sum_{|\alpha|\le 2} a_\alpha(x)\,\partial_{t_*}^{\alpha_0}\partial_x^{\alpha'},
\qquad a_\alpha \ \text{independent of }t_*.
\]
Collect the terms with $\alpha_0=2$.
Since $P$ is a wave operator, the coefficient of $\partial_{t_*}^2$ is a smooth function that does not vanish
on $X_{M,a}^\circ$; up to an everywhere nonzero smooth factor coming from the choice of density in $\Box_g$,
this coefficient is $g^{t_*t_*}$.
Because $t_*$ is a global time function on the future-extended region (so $dt_*$ is timelike), the slices
$\{t_*={\rm const}\}$ are spacelike and therefore $g^{t_*t_*}<0$.
Multiplying $P$ by a smooth nowhere-vanishing factor depending only on $x$ (see the remark below), we may normalize
the $\partial_{t_*}^2$ coefficient to equal $1$.
Collect the remaining $\partial_{t_*}$ terms (including mixed derivatives $\partial_{t_*}\partial_{x_j}$)
into $Q\partial_{t_*}$, and the rest into $L$.
Stationarity gives $t_*$--independence, and axisymmetry gives commutation with $\Phi$.
\end{proof}

\begin{remark}[Normalization of the $\partial_{t_*}^2$ coefficient]
The harmless division by a smooth nowhere vanishing function in the proof above does not affect the
quasinormal spectrum. Indeed, if $P$ is replaced by $\widetilde P = \mu P$ with $\mu\in C^\infty(X_{M,a})$
and $\mu\neq 0$, then the stationary family satisfies $\widetilde P(\omega)=\mu P(\omega)$,
and hence $\widetilde P(\omega)^{-1}=P(\omega)^{-1}\mu^{-1}$ wherever either inverse exists;
in particular, the poles (and their orders) of the meromorphic resolvent are unchanged.
\end{remark}

\subsection{Cauchy problem and energy spaces}\label{subsec:Cauchy}

Let $\Sigma_{M,a}=\{t_*=0\}$ and denote initial data by
\begin{equation}\label{eq:initial-data}
f_0 := u|_{t_*=0},\qquad f_1 := (\partial_{t_*}u)|_{t_*=0}.
\end{equation}
For $s\ge 0$ define the standard energy Sobolev space
\begin{equation}\label{eq:energy-space}
\mathcal H^s := H^{s+1}(\Sigma_{M,a})\times H^s(\Sigma_{M,a})
\end{equation}
(with respect to any fixed smooth Riemannian metric on $\Sigma_{M,a}$; different choices yield equivalent norms on compact sets).

\begin{theorem}[Well-posedness and an a priori bound]\label{thm:wellposed}
For every $s\ge 0$ and initial data $(f_0,f_1)\in\mathcal H^s$, there exists a unique solution
\[
u\in C^0(\mathbb R_{t_*};H^{s+1}(\Sigma_{M,a}))\cap C^1(\mathbb R_{t_*};H^s(\Sigma_{M,a}))
\]
to $Pu=0$ on $\mathcal M_{M,a}$ with initial data \eqref{eq:initial-data}. Moreover, there exist constants
$C_s,\kappa_s\ge 0$ such that
\begin{equation}\label{eq:apriori}
\|u(t_*)\|_{H^{s+1}} + \|\partial_{t_*}u(t_*)\|_{H^s}
\ \le\
C_s\,e^{\kappa_s|t_*|}\,\big(\|f_0\|_{H^{s+1}}+\|f_1\|_{H^s}\big),
\qquad t_*\in\mathbb R.
\end{equation}
\end{theorem}

\begin{proof}
We give a self-contained energy estimate argument; for Kerr--de~Sitter the required geometric hypotheses
(existence of a smooth global time function $t_*$ on the future-extended region and smoothness across horizons)
are standard and can be found in \cite{VasyKerrDeSitter,DyatlovQNM,PetersenVasyKerrDeSitterJEMS}.

Since $dt_*$ is timelike, the operator $P=\Box_{g_{M,a}}+V$ is strictly hyperbolic with respect to $t_*$.
After normalizing the $\partial_{t_*}^2$ coefficient as in Lemma~\ref{lem:splitting}, we may write
$P=\partial_{t_*}^2+Q\partial_{t_*}+L$ with $Q\in\Diff^1(X_{M,a})$ and $L\in\Diff^2(X_{M,a})$ having
$t_*$--independent coefficients.
Fix a smooth Riemannian metric on $\Sigma_{M,a}$ and use it to define the Sobolev norms in \eqref{eq:energy-space}.
Let $\Lambda^s=(1-\Delta_\Sigma)^{s/2}$ be defined using an elliptic Laplace-type operator $\Delta_\Sigma$ on $\Sigma_{M,a}$.
For smooth solutions, apply $\Lambda^s$ to the equation and take the $L^2$ inner product of
$\partial_{t_*}(\Lambda^s u)$ with $\partial_{t_*}^2(\Lambda^s u)$.
Standard commutator estimates for $\Lambda^s$ and $Q,L$ (whose coefficients are smooth and $t_*$--independent)
give a differential inequality of the form
\[
\frac{d}{dt_*}\Bigl(\|\Lambda^s\partial_{t_*}u(t_*)\|_{L^2}^2+\|\Lambda^{s+1}u(t_*)\|_{L^2}^2\Bigr)
\le C_s\Bigl(\|\Lambda^s\partial_{t_*}u(t_*)\|_{L^2}^2+\|\Lambda^{s+1}u(t_*)\|_{L^2}^2\Bigr),
\]
where $C_s$ depends on finitely many derivatives of the coefficients of $Q,L,V$.
Gronwall's inequality yields \eqref{eq:apriori} for smooth solutions.
Existence and uniqueness in the stated spaces follow by a standard density argument, using the energy estimate
to pass to limits.
\end{proof}

\begin{remark}\label{rem:energy}
On Kerr--de~Sitter, one can obtain sharper bounds (and decay) using redshift multipliers near horizons and
normally hyperbolic trapping estimates.  For the functional calculus below, the crude exponential bound
\eqref{eq:apriori} is enough to justify Laplace transforms; later sections will invoke refined resolvent
estimates for decay.
\end{remark}

\medskip

\noindent\textbf{The Cauchy propagator.}
For $s\ge 0$ we denote by
\begin{equation}\label{eq:propagator-def}
\mathsf U(t_*):\mathcal H^s\to\mathcal H^s,\qquad
\mathsf U(t_*)(f_0,f_1):=\bigl(u(t_*),\partial_{t_*}u(t_*)\bigr),
\end{equation}
the solution operator provided by Theorem~\ref{thm:wellposed}. By stationarity, $\mathsf U(t_*)$ forms a strongly continuous
group.

\subsection{First-order formulation and the stationary family}\label{subsec:first-order}

Introduce the state vector
\[
U(t_*):=\binom{u(t_*)}{\partial_{t_*}u(t_*)}.
\]
Using the splitting \eqref{eq:splitting}, the equation $Pu=0$ is equivalent to the first-order system
\begin{equation}\label{eq:first-order}
\partial_{t_*}U = \mathcal A\,U,
\qquad
\mathcal A :=
\begin{pmatrix}
0 & I\\
-L & -Q
\end{pmatrix},
\end{equation}
acting on $\mathcal H^0$ with natural domain $\Dom(\mathcal A)=H^2(\Sigma_{M,a})\times H^1(\Sigma_{M,a})$.
The solution operator is thus the strongly continuous group $e^{t_*\mathcal A}$ on $\mathcal H^s$:
\[
U(t_*) = e^{t_*\mathcal A}U(0).
\]

\medskip

\noindent\textbf{Stationary family.}
For $\omega\in\mathbb C$ define the stationary operator on $X_{M,a}$ by the mode substitution
$u(t_*,x)=e^{-i\omega t_*}v(x)$:
\begin{equation}\label{eq:stationary-family}
P(\omega)v := e^{i\omega t_*}P\big(e^{-i\omega t_*}v\big)\Big|_{t_*=\mathrm{const}}
= \big(L - i\omega Q - \omega^2\big)v.
\end{equation}
Thus $e^{-i\omega t_*}v$ solves $Pu=0$ if and only if $P(\omega)v=0$.

\subsection{Meromorphic resolvent and quasinormal modes}\label{subsec:meromorphic}

We use a microlocal Fredholm framework (Vasy-type variable order Sobolev spaces) to define $P(\omega)^{-1}$
as a meromorphic family.  We summarize the needed conclusions.

\medskip

\noindent\textbf{Variable order spaces (briefly).}
Near the horizons, the stationary operator $P(\omega)$ is non-elliptic and has radial points in phase space.
To formulate a Fredholm problem, one uses variable order Sobolev spaces $H^{\mathbf s}(X_{M,a})$ whose order
$\mathbf s$ is chosen to enforce outgoing regularity at the future radial sets and allow incoming singularities
at the past radial sets.  This choice of variable order is the microlocal implementation of the standard 
outgoing (radiation) condition at the future event and cosmological horizons on the Vasy/Kerr--de~Sitter extension;
see \cite{VasyKerrDeSitter,HintzNHT,PetersenVasyAnalyticity,PetersenVasyKerrDeSitterJEMS}.
  We denote the corresponding graph space by
\begin{equation}\label{eq:graph-space}
\begin{aligned}
\mathcal X^{\mathbf s}(\omega)
&:=\{v\in H^{\mathbf s}(X_{M,a}):\ P(\omega)v\in H^{\mathbf s-1}(X_{M,a})\},\\
\mathcal Y^{\mathbf s-1}
&:=H^{\mathbf s-1}(X_{M,a}).
\end{aligned}
\end{equation}
with norm $\|v\|_{\mathcal X^{\mathbf s}(\omega)}=\|v\|_{H^{\mathbf s}}+\|P(\omega)v\|_{H^{\mathbf s-1}}$.

\begin{theorem}[Meromorphic resolvent and QNMs]\label{thm:meromorphic2}
Fix $(M,a)$ subextremal and choose an outgoing order function $\mathbf s$.
Then there exists $\omega_0\in\mathbb R$ such that for $\Im\omega>\omega_0$,
\[
P(\omega):\mathcal X^{\mathbf s}(\omega)\to\mathcal Y^{\mathbf s-1}
\]
is invertible, and the inverse
\begin{equation}\label{eq:resolvent}
R(\omega):=P(\omega)^{-1}:\mathcal Y^{\mathbf s-1}\to\mathcal X^{\mathbf s}(\omega)
\end{equation}
extends meromorphically to $\omega\in\mathbb C$, with finite-rank residues at poles.
The poles of $R(\omega)$ are called \emph{quasinormal mode frequencies} (QNMs), and elements of $\ker P(\omega)$
at a pole are called \emph{resonant states}.
\end{theorem}

\begin{remark}\label{rem:simple-poles}
In the high-frequency slow-rotation regime relevant later, the QNMs we isolate are simple and admit stable labeling
(by semiclassical quantization).  We nevertheless allow higher-order poles in the abstract setup because they
produce distinct time-domain signatures (polynomial prefactors) and are relevant to exceptional-point phenomena.
\end{remark}

\subsection{Functional-analytic conventions and mapping properties}\label{subsec:functional-conventions}

We briefly collect the functional-analytic conventions used throughout the paper, with emphasis on the mapping
properties of the stationary resolvent that are needed for contour deformation.

\medskip

\noindent\textbf{(1) Standard and semiclassical Sobolev norms.}
By the convention \eqref{eq:Xbar-convention}, $X_{M,a}$ is a compact manifold with boundary.
We write $H^s:=H^s(X_{M,a})$ for standard Sobolev spaces defined using any fixed smooth Riemannian metric on $X_{M,a}$
(equivalently, by extending $X_{M,a}$ across its boundary and using local charts).
For a \emph{temporal-frequency} semiclassical parameter $h_\omega\in(0,1]$ we denote by $H_{h_\omega}^s$ the semiclassical Sobolev space
with norm
\[
\|v\|_{H_{h_\omega}^s}:=\|\langle h_\omega D\rangle^s v\|_{L^2},
\qquad \langle h_\omega D\rangle=(1+h_\omega^2\Delta)^{1/2},
\]
where $\Delta$ is any fixed elliptic Laplace-type operator on $X_{M,a}$.
We will use these norms in Section~\ref{sec:resolvent-shifted} on dyadic blocks where $|\Re\omega|\sim h_\omega^{-1}$.

\begin{remark}[Two semiclassical parameters]\label{rem:two-h-parameters}
Two unrelated small parameters appear in this paper:
\begin{itemize}
\item $h_\omega\sim |\Re\omega|^{-1}$ is the \emph{temporal-frequency} semiclassical parameter used in the resolvent estimates of
Section~\ref{sec:resolvent-shifted}.
\item $h_\ell=\ell^{-1}$ is the \emph{angular} semiclassical parameter used for equatorial microlocalization via
$A_{\pm,h_\ell}$ in Section~\ref{subsec:equatorial-cutoffs} and in the two-mode analysis (Sections~\ref{sec:two-mode}--\ref{sec:application}).
\end{itemize}
We keep the notation distinct to avoid confusion: there is no reason to identify $h_\omega$ with $h_\ell$ in general.
In the two-mode regime we will only apply the angular microlocalization at frequencies $|\Re\omega|\asymp \ell$
(equivalently $h_\omega\asymp h_\ell$), so the coexistence of the two parameters causes no ambiguity.
\end{remark}

We will repeatedly use the elementary comparison inequalities (valid for $0<h_\omega\le 1$ and $s\ge 0$)
\begin{equation}\label{eq:Hh-to-H}
\|v\|_{H_{h_\omega}^s}\le \|v\|_{H^s},
\qquad
\|v\|_{H^s}\le C_s\,h_\omega^{-s}\,\|v\|_{H_{h_\omega}^s},
\end{equation}
which follow from the pointwise bounds $\langle h_\omega\xi\rangle^s\le \langle\xi\rangle^s\le C_s h_\omega^{-s}\langle h_\omega\xi\rangle^s$
in local coordinates.

\medskip

\noindent\textbf{(2) Variable order spaces and the outgoing condition at radial sets.}
To formulate a Fredholm problem for $P(\omega)$ across the horizons one uses variable order Sobolev spaces $H^{\mathbf s}$
(Vasy-type spaces), where the order function $\mathbf s=\mathbf s(x,\xi)$ is chosen so that $\mathbf s$ is larger than $1/2$
at the future radial sets (enforcing outgoing regularity) and smaller than $1/2$ at the past radial sets
(allowing incoming singularities).  We will not reproduce the full construction; see
\cite{VasyKerrDeSitter,PetersenVasyAnalyticity,PetersenVasyKerrDeSitterJEMS}.
For any such outgoing order function $\mathbf s$ we work with the graph space $\mathcal X^{\mathbf s}(\omega)$ and target space
$\mathcal Y^{\mathbf s-1}$ defined in \eqref{eq:graph-space}.

\medskip

\noindent\textbf{(3) Cutoff resolvent as a map between standard Sobolev spaces.}
All time-domain contour integrals in this paper involve the \emph{cutoff} resolvent $\chi R(\omega)\chi$, where
$\chi\in C_c^\infty(X_{M,a})$ is supported strictly inside the physical region away from the horizons.
On the support of $\chi$ the variable order spaces coincide with standard Sobolev spaces, and the stationary family is elliptic
away from the characteristic set.  In particular, for each fixed $\omega$ away from poles and each $s\in\mathbb R$,
\begin{equation}\label{eq:cutoff-mapping}
\chi R(\omega)\chi: H^{s-1}(X_{M,a})\longrightarrow H^{s+1}(X_{M,a})
\end{equation}
is a bounded operator, with bounds depending continuously on $\omega$ on pole-free compact subsets of $\mathbb C$.

\medskip

\noindent\textbf{(4) Parameter dependence.}
Whenever $(M,a)$ varies in a compact parameter set $\mathcal K$, we will always choose the cutoffs ($\chi$, the microlocal
equatorial cutoffs $A_{\pm,h_\ell}$) and the outgoing order function $\mathbf s$ so that all norms and operator bounds are uniform on $\mathcal K$.
In particular, the resolvent bounds on shifted contours in Section~\ref{sec:resolvent-shifted} and all subsequent constants
are uniform on $\mathcal K$ once a pole-free contour is fixed (Proposition~\ref{prop:uniform-contour}).

\subsection{Forward Laplace transform and an explicit inversion formula}\label{subsec:laplace}

Let $u$ be the unique solution to $Pu=0$ with initial data $(f_0,f_1)\in\mathcal H^s$.
Define the forward (Heaviside-truncated) solution
\[
u^+(t_*,x):=\mathbf 1_{t_*\ge 0}\,u(t_*,x).
\]
As a distribution on $\mathbb R_{t_*}\times X_{M,a}$, $u^+$ satisfies an inhomogeneous equation whose source
is supported at $t_*=0$ and encodes the initial data.

\begin{lemma}[Distributional source at $t_*=0$]\label{lem:source}
Let $P=\partial_{t_*}^2+Q\partial_{t_*}+L$ as in \eqref{eq:splitting}.
If $u$ solves $Pu=0$ for $t_*>0$ and has initial data $(f_0,f_1)$ at $t_*=0$, then
\begin{equation}\label{eq:Puplus}
P u^+ = \delta'(t_*)\,f_0 \;+\; \delta(t_*)\,\big(f_1+Q f_0\big)
\quad\text{in }\mathcal D'(\mathbb R_{t_*}\times X_{M,a}).
\end{equation}
\end{lemma}

\begin{proof}
Write $u^+=H u$ with $H=\mathbf 1_{t_*\ge 0}$.
Then $\partial_{t_*}(Hu)=\delta(t_*)\,u(0,\cdot)+H\,\partial_{t_*}u$ and
\[
\partial_{t_*}^2(Hu)=\delta'(t_*)\,u(0,\cdot)+\delta(t_*)\,\partial_{t_*}u(0,\cdot)+H\,\partial_{t_*}^2u.
\]
Since $Q$ and $L$ are $t_*$--independent and act only in $x$, we have $Q(Hu)=H(Qu)$ and $L(Hu)=H(Lu)$.
Substituting into $P(Hu)$ and using $Pu=0$ for $t_*>0$ gives
\[
Pu^+=\delta'(t_*)\,f_0+\delta(t_*)\,f_1+Q\big(\delta(t_*)f_0\big)
=\delta'(t_*)\,f_0+\delta(t_*)\,(f_1+Qf_0),
\]
as claimed.
\end{proof}

\medskip

\noindent\textbf{Fourier--Laplace transform.}
For $\Im\omega$ sufficiently large (depending on the growth bound \eqref{eq:apriori}), define the forward
Fourier--Laplace transform
\begin{equation}\label{eq:laplace-def}
\widehat u^+(\omega,x):=\int_{0}^{\infty} e^{i\omega t_*}\,u(t_*,x)\,dt_*.
\end{equation}
Since $u^+$ is supported in $t_*\ge 0$, $\widehat u^+$ extends holomorphically to a half-plane $\Im\omega>\omega_1$.
Taking the Fourier transform in $t_*$ of \eqref{eq:Puplus} and using \eqref{eq:stationary-family} yields:

\begin{lemma}[Resolvent identity for the Laplace transform]\label{lem:resolvent-identity}
For $\Im\omega$ sufficiently large,
\begin{equation}\label{eq:resolvent-identity}
P(\omega)\,\widehat u^+(\omega)
=
\big(f_1+Q f_0\big) - i\omega\,f_0.
\end{equation}
Equivalently,
\begin{equation}\label{eq:uhat}
\widehat u^+(\omega)
=
R(\omega)\,\Big(\,f_1 + (Q-i\omega)f_0\,\Big),
\qquad R(\omega)=P(\omega)^{-1}.
\end{equation}
\end{lemma}

\begin{proof}
Taking the Fourier transform of \eqref{eq:Puplus} in $t_*$ with kernel $e^{i\omega t_*}$ gives
\[
\widehat{Pu^+}(\omega)=\widehat{\delta'}(\omega)f_0+\widehat{\delta}(\omega)\,(f_1+Qf_0)
=(-i\omega)f_0+(f_1+Qf_0).
\]
On the other hand, stationarity of coefficients implies $\widehat{Pu^+}(\omega)=P(\omega)\widehat u^+(\omega)$,
with $P(\omega)$ given by \eqref{eq:stationary-family}. This proves \eqref{eq:resolvent-identity}, and
\eqref{eq:uhat} follows by applying $R(\omega)$.
\end{proof}

\medskip

\noindent\textbf{Inverse Laplace representation.}
By the standard inversion formula for Fourier--Laplace transforms of forward-supported distributions,
for any $C>\omega_1$ and all $t_*>0$,
\begin{equation}\label{eq:inverse-laplace}
u(t_*) =
\frac{1}{2\pi}\int_{\Im\omega=C} e^{-i\omega t_*}\,
R(\omega)\,\Big(\,f_1 + (Q-i\omega)f_0\,\Big)\,d\omega,
\end{equation}
where the integral is taken along the horizontal line $\Im\omega=C$ and interpreted as an oscillatory
integral in appropriate Sobolev spaces.  Formula \eqref{eq:inverse-laplace} is the starting point for
contour deformation: shifting the contour downward across QNM poles produces ringdown terms.

\subsection{Residues, higher-order poles, and time-domain contributions}\label{subsec:residues}

Theorem~\ref{thm:meromorphic2} implies that $R(\omega)$ has at most finite-order poles.
We record the standard residue-to-time-domain dictionary, since it will be used later when discussing
possible exceptional-point (Jordan) effects.

Let $\omega_0$ be a pole of $R(\omega)$ of order $m\ge 1$. Then there exist finite-rank operators $A_{-j}$ such that
\begin{equation}\label{eq:Laurent}
R(\omega) = \sum_{j=1}^{m} \frac{A_{-j}}{(\omega-\omega_0)^j} + R_{\mathrm{hol}}(\omega),
\end{equation}
with $R_{\mathrm{hol}}$ holomorphic near $\omega_0$.

\begin{lemma}[Time-domain contribution of a pole]\label{lem:pole-contribution}
Let $F(\omega)$ be holomorphic near $\omega_0$ with values in a Banach space, and consider the contour integral
\[
I(t_*):=\frac{1}{2\pi}\int_{\Gamma} e^{-i\omega t_*}\,R(\omega)\,F(\omega)\,d\omega,
\]
where $\Gamma$ is a small positively oriented circle around $\omega_0$.
Then for $t_*>0$,
\begin{equation}\label{eq:pole-term}
\begin{aligned}
I(t_*)
&=
i\,e^{-i\omega_0 t_*}
\sum_{j=1}^{m}
\frac{(-it_*)^{j-1}}{(j-1)!}\,A_{-j}\,F(\omega_0)\\
&\quad
+i\,e^{-i\omega_0 t_*}\sum_{j=1}^{m}\sum_{q=1}^{j-1}
\frac{(-it_*)^{j-1-q}}{(j-1-q)!}\,A_{-j}\,\frac{F^{(q)}(\omega_0)}{q!}.
\end{aligned}
\end{equation}
In particular, if the pole is simple ($m=1$), then
\[
I(t_*) = i\,e^{-i\omega_0 t_*}\,A_{-1}\,F(\omega_0).
\]
\end{lemma}

\begin{proof}
Insert the Laurent expansion \eqref{eq:Laurent} and the Taylor expansion
$F(\omega)=\sum_{q\ge 0} \frac{F^{(q)}(\omega_0)}{q!}(\omega-\omega_0)^q$ into the integral.
The residue theorem gives contributions only from terms with $(\omega-\omega_0)^{-1}$.
Since $I(t_*)$ carries the prefactor $1/(2\pi)$, the residue theorem contributes an additional factor of $i$.
Writing $e^{-i\omega t_*}=e^{-i\omega_0 t_*}\sum_{p\ge 0}\frac{(-it_*)^p}{p!}(\omega-\omega_0)^p$
and collecting the $(\omega-\omega_0)^{-1}$ coefficients yields \eqref{eq:pole-term}.
\end{proof}

In our main high-frequency equatorial regime, the relevant QNMs are simple, and the first (simple pole)
formula suffices.  We keep Lemma~\ref{lem:pole-contribution} to make clear how higher-order poles would
manifest in the time domain by polynomial factors in $t_*$.

\subsection{Exact symmetry reduction in the azimuthal number}\label{subsec:azimuthal}

Since $P$ is axisymmetric, $[P,\Phi]=0$, hence $P$ and $P(\omega)$ preserve the Fourier modes in $\varphi_*$.  The commutation relations $[P,\Phi]=0$ and $[L,\Phi]=[Q,\Phi]=0$ are exact on the extended manifold,
since the extension across the horizons is performed in an axisymmetric fashion; moreover, our radial cutoffs depend only on $r$ and
the symbols $a_\pm$ defining $A_{\pm,h_\ell}$ are taken independent of $\varphi_*$.

For $k\in\mathbb Z$ set
\[
\mathcal D_k := \{v\in C^\infty(X_{M,a}): \Phi v = ik\,v\},
\qquad
L^2(X_{M,a})=\bigoplus_{k\in\mathbb Z} L^2_k,\quad L^2_k:=\overline{\mathcal D_k}^{L^2}.
\]
Then $Q$ and $L$ restrict to $Q_k$ and $L_k$ on each $k$-subspace, and
\begin{equation}\label{eq:Pk}
P_k(\omega)=L_k - i\omega Q_k - \omega^2,
\qquad
R_k(\omega)=P_k(\omega)^{-1}.
\end{equation}
The inversion formula \eqref{eq:inverse-laplace} holds in each $k$ sector by applying the orthogonal projector
$\Pi_k$ onto $L^2_k$:
\begin{equation}\label{eq:inverse-laplace-k}
\Pi_k u(t_*) =
\frac{1}{2\pi}\int_{\Im\omega=C} e^{-i\omega t_*}\,
R_k(\omega)\,\Pi_k\Big(\,f_1 + (Q-i\omega)f_0\,\Big)\,d\omega.
\end{equation}
Thus the entire problem decouples exactly in $k$.  Later, we will focus on the \emph{equatorial high-frequency}
regime in which $|k|$ and the total angular frequency are both large and comparable.

\subsection{Microlocal equatorial cutoffs (angular semiclassical parameter $h_\ell$)}\label{subsec:equatorial-cutoffs}

To isolate the equatorial high-frequency package relevant for two-mode dominance, we introduce a microlocal cutoff on the
sphere which localizes near the conic set where the azimuthal momentum saturates the total angular momentum.  This
construction is \emph{independent} of the temporal-frequency semiclassical parameter used later for resolvent estimates
(see Remark~\ref{rem:two-h-parameters}).

\medskip

\noindent\textbf{Semiclassical notation on $\mathbb S^2$.}
Let $(\theta,\varphi_*)$ be coordinates on $\mathbb S^2$ and $(\xi_\theta,\xi_\varphi)$ the dual variables.
The (positive) Laplacian has principal symbol
\[
p_{\mathbb S^2}(\theta,\varphi_*;\xi_\theta,\xi_\varphi)=\xi_\theta^2+\frac{\xi_\varphi^2}{\sin^2\theta}.
\]
Fix an \emph{angular} semiclassical parameter $h_\ell\in(0,h_0]$ (later we will take $h_\ell=\ell^{-1}$) and write the
rescaled covariables $\eta_\theta=h_\ell\xi_\theta$, $\eta_\varphi=h_\ell\xi_\varphi$.
The unit cosphere bundle corresponds to $p_{\mathbb S^2}=1$, i.e.\ $\eta_\theta^2+\eta_\varphi^2/\sin^2\theta=1$.

\medskip

\noindent\textbf{Equatorial conic sets.}
Define the outgoing/incoming equatorial sets
\begin{equation}\label{eq:eq-sets}
\mathcal E_\pm:=\Bigl\{(\theta,\varphi_*;\eta_\theta,\eta_\varphi)\in S^*\mathbb S^2:\ 
\theta=\frac{\pi}{2},\ \eta_\theta=0,\ \eta_\varphi=\pm 1\Bigr\}.
\end{equation}
These correspond to geodesics confined to the equatorial plane with maximal azimuthal momentum.
Choose semiclassical symbols $a_\pm\in S^0(T^*\mathbb S^2)$ supported in a small conic neighborhood of
$\mathcal E_\pm\subset S^*\mathbb S^2$ and in an annular neighborhood of the unit cosphere.  Concretely, we may assume
$a_\pm$ is supported where
\[
\rho(\theta,\varphi,\eta):=\Big(\eta_\theta^2+\frac{\eta_\varphi^2}{\sin^2\theta}\Big)^{1/2}\in[1/2,2],
\]
and where $(\theta,\varphi;\eta/\rho)$ lies in a sufficiently small neighborhood of $\mathcal E_\pm$.
We also impose $a_+a_-=0$, and take $a_\pm\equiv 1$ on a smaller conic neighborhood of $\mathcal E_\pm$.
For later commutation with the exact azimuthal projectors $\Pi_k$, we choose $a_\pm$ \emph{independent} of $\varphi_*$.

\medskip

\noindent\textbf{Equatorial microlocal cutoffs.}
Let $\Op_{h_\ell}$ be any semiclassical quantization on $\mathbb S^2$.
Define bounded semiclassical pseudodifferential operators
\begin{equation}\label{eq:Ah}
A_{\pm,h_\ell}:=\Op_{h_\ell}(a_\pm):L^2(\mathbb S^2)\to L^2(\mathbb S^2),
\end{equation}
and extend them to $X_{M,a}=(r_e-\delta,r_c+\delta)\times\mathbb S^2$ by acting trivially in $r$:
\[
A_{\pm,h_\ell} := I_r\otimes \Op_{h_\ell}(a_\pm).
\]
In particular, $A_{\pm,h_\ell}$ commutes with multiplication by any cutoff depending only on $r$.
Since $a_\pm$ is independent of $\varphi_*$, $A_{\pm,h_\ell}$ commutes with $\Phi$ and therefore with the exact Fourier
projectors $\Pi_k$ on $L^2(X_{M,a})$.
We will apply these cutoffs either directly to time-domain solutions $u(t_*,\cdot)$ (for fixed $t_*$), or to the stationary
resolvent $R(\omega)$ by sandwiching:
\begin{equation}\label{eq:sandwich}
R_{\pm,h_\ell}(\omega):=A_{\pm,h_\ell}\,R(\omega)\,A_{\pm,h_\ell}^*.
\end{equation}

\begin{remark}[Uniform boundedness]\label{rem:Ah-bdd}
Since $A_{\pm,h_\ell}\in\Psi^0_{h_\ell}(\mathbb S^2)$ and is extended trivially in $r$, it is uniformly bounded on standard
Sobolev spaces: for each $s\in\mathbb R$ there exists $C_s$ such that
\begin{equation}\label{eq:Ah-bdd}
\|A_{\pm,h_\ell}\|_{H^s\to H^s}\le C_s,\qquad 0<h_\ell<h_0.
\end{equation}
We will use \eqref{eq:Ah-bdd} repeatedly to insert equatorial localization into time-domain expansions and remainder bounds.
\end{remark}

\begin{remark}[Preference for microlocal cutoffs]\label{rem:why-microlocal}
In Kerr--de~Sitter, separation of variables ties the angular eigenfunctions to the frequency $\omega$
(spheroidal harmonics).  For time-domain analysis it is therefore preferable to use \emph{microlocal} sectorial cutoffs
such as $A_{\pm,h_\ell}$, which do not depend on $\omega$ and capture the equatorial geometry directly in phase space.
In the slow-rotation regime, the equatorial package is stable under the stationary dynamics, and one can prove a microlocal sectorization result excluding non-equatorial poles up to $\mathcal O(\ell^{-\infty})$
errors; see Appendix~\ref{app:companion-inputs}.
\end{remark}

\medskip

\noindent\textbf{Equatorial sector of initial data.}
For $s\ge 0$ define the equatorial data subspaces
\[
\mathcal H^s_{\pm,h_\ell}:=\Bigl\{(f_0,f_1)\in\mathcal H^s:\ (A_{\pm,h_\ell}f_0,A_{\pm,h_\ell}f_1)=(f_0,f_1)+\mathcal O(h_\ell^\infty)\Bigr\},
\]
interpreting $\mathcal O(h_\ell^\infty)$ in $\mathcal H^s$.
Later, we will choose $h_\ell=\ell^{-1}$ and combine this localization with the azimuthal mode reduction to isolate the
$k=\pm\ell$ equatorial package.

\medskip

\noindent\textbf{Output of Section~\ref{sec:setup2}.}
The key output of this section is the explicit resolvent representation \eqref{eq:inverse-laplace} (and its $k$-reduced
version \eqref{eq:inverse-laplace-k}), which converts time-domain analysis into a contour problem for the stationary
resolvent $R(\omega)$.  The microlocal cutoffs $A_{\pm,h_\ell}$ set up the sectorial framework needed later to prove a
two-mode dominance statement by excluding all but a single QNM in an appropriate strip.

\section{Resolvent bounds on a shifted contour}\label{sec:resolvent-shifted}

The inverse Laplace representation \eqref{eq:inverse-laplace} suggests the following strategy:
shift the contour from $\Im\omega=C\gg1$ down to $\Im\omega=-\nu<0$, picking up the residues of the
quasinormal poles (ringdown terms) and leaving a remainder integral on $\Im\omega=-\nu$.
To make this rigorous, we need \emph{a priori} bounds for the stationary resolvent
\[
R(\omega)=P(\omega)^{-1}
\]
on the shifted contour $\Gamma_{-\nu}:=\{\omega\in\mathbb C:\Im\omega=-\nu\}$.
The purpose of this section is to give such bounds in a form adapted to later contour deformation,
including (i) uniformity on compact parameter sets, (ii) semiclassical high-frequency rescalings,
and (iii) bounds for $\omega$--derivatives needed for repeated integration by parts.

\vspace{0.3em}
\noindent\textbf{Standing geometric input.}
All estimates below ultimately rely on two microlocal dynamical facts for Kerr--de~Sitter:
\begin{enumerate}
\item \emph{Stable radial point structure at the (future) horizons} (redshift), which gives uniform
propagation/regularity estimates at the radial sets.
\item \emph{Normally hyperbolic trapping} of the null bicharacteristic flow on the trapped set in the domain
of outer communication (photon sphere dynamics), which gives semiclassical high-energy resolvent estimates.
\end{enumerate}
For the scalar wave equation on Kerr--de~Sitter, (1) is part of the general Fredholm framework for QNMs,
and (2) is established in the Kerr--de~Sitter setting (in the full subextremal range) and yields the
semiclassical estimates we use; see e.g.\ \cite{VasyKerrDeSitter, WunschZworskiNH, PetersenVasyKerrDeSitterJEMS}.
We will treat these as standard microlocal inputs, and derive from them the precise operator bounds
needed for the contour argument.

\subsection{Choice of shifted contour and cutoff resolvents}\label{subsec:contour-choice}

Fix $\nu>0$. We denote the horizontal line
\begin{equation}\label{eq:Gamma-nu}
\Gamma_{-\nu}:=\{\omega\in\mathbb C:\ \Im\omega=-\nu\}.
\end{equation}
In later sections we will choose $\nu$ so that $\Gamma_{-\nu}$ lies strictly between two resonance layers in
the equatorial high-frequency sector; in particular, $\Gamma_{-\nu}$ will avoid QNM poles in that sector.
For the present section, we only assume:

\begin{assumption}[No poles on the contour]\label{ass:contour-pole-free}
For the parameter range under consideration, $R(\omega)$ has no pole on $\Gamma_{-\nu}$.
Equivalently, $P(\omega)$ is invertible (as a Fredholm operator) for every $\omega\in\Gamma_{-\nu}$.
\end{assumption}

Since the pole set is discrete for each fixed parameter value, this assumption can always be arranged by
a small perturbation of $\nu$.  When parameters vary, pole-freeness of a fixed contour is an \emph{open} condition
by analytic Fredholm theory, hence one may work on any compact set $\mathcal K$ contained in the corresponding open
pole-free region; see \S\ref{subsec:param-uniformity} and Proposition~\ref{prop:uniform-contour}.

\medskip

\noindent\textbf{Cutoff resolvent.}
Let $\chi\in C_c^\infty(X_{M,a})$ be a smooth cutoff supported in the physical region
\[
\supp \chi \subset (r_e+\tfrac{\delta}{2},r_c-\tfrac{\delta}{2})\times \mathbb S^2.
\]
We also assume that $\chi\equiv 1$ on a neighborhood of the supports of all initial data and observation regions.
Define the cutoff resolvent
\begin{equation}\label{eq:cutoff-resolvent}
R_\chi(\omega):=\chi\,R(\omega)\,\chi.
\end{equation}
All contour estimates in the time domain may be formulated in terms of $R_\chi(\omega)$ once we choose
$\chi$ so that $\chi\equiv 1$ on a neighborhood of the spacetime support of the compactly supported forcing
$F_\vartheta$ introduced in Section~\ref{subsec:smooth-time-cutoff}; see Lemma~\ref{lem:chi-forcing-support} below.
In particular $\widehat F_\vartheta(\omega)=\chi\,\widehat F_\vartheta(\omega)$, and we may insert $\chi$ on the right
of the resolvent without changing the contour integrals we consider.

\subsection{Semiclassical rescaling for high frequencies}\label{subsec:semiclassical-rescaling}

The high-frequency regime relevant for ringdown and QNM asymptotics corresponds to $|\Re\omega|\gg1$
with $\Im\omega$ bounded below.
We use a dyadic semiclassical rescaling to access uniform estimates.

\subsubsection{Dyadic partition in $\Re\omega$}

Let $\psi\in C_c^\infty((1/2,2))$ satisfy $\sum_{j\in\mathbb Z}\psi(2^{-j}|\sigma|)=1$ for all $|\sigma|\ge1$.
For $\sigma\in\mathbb R$ write
\[
1=\psi_0(\sigma)+\sum_{j\ge 1}\psi_j(\sigma),
\qquad
\psi_j(\sigma):=\psi(2^{-j}|\sigma|),\ \ j\ge1,
\]
with $\psi_0$ supported where $|\sigma|\le 2$.
On the shifted contour $\omega=\sigma-i\nu$, we thus decompose the integral into low and high frequency parts.

For each dyadic block $j\ge1$, set
\begin{equation}\label{eq:hj}
h_\omega:=h_j:=2^{-j},\qquad \sigma\sim h_\omega^{-1},\qquad z:=h_\omega\omega.
\end{equation}

\begin{remark}[Temporal versus angular semiclassics]
In this section $h_\omega$ is the dyadic \emph{temporal--frequency} semiclassical parameter associated with the
magnitude of $\Re\omega$; it is unrelated to the \emph{angular} parameter $h_\ell=\ell^{-1}$ used for equatorial
microlocalization in Section~\ref{subsec:equatorial-cutoffs}. We keep the subscripts to avoid ambiguity; see
Remark~\ref{rem:two-h-parameters}.
\end{remark}
Then on $\Gamma_{-\nu}$ we have
\begin{equation}\label{eq:z-region}
z = h_\omega\sigma - i h_\omega\nu,\qquad
\Re z\in [1/2,2]\cdot \mathrm{sgn}(\sigma),\qquad
\Im z = -\nu h_\omega.
\end{equation}
Thus $z$ stays in a compact set, and its imaginary part is $\mathcal O(h_\omega)$.

\subsubsection{Semiclassical stationary operator}

Recall from \eqref{eq:stationary-family} that
\[
P(\omega)=L-i\omega Q-\omega^2.
\]
Define the semiclassical operator (depending on $z$ and $h_\omega$) by
\begin{equation}\label{eq:Phz}
P_{h_\omega}(z)\ :=\ h_\omega^2\,P(h_\omega^{-1}z)\ =\ h_\omega^2L - i h_\omega z\,Q - z^2.
\end{equation}
We view $P_{h_\omega}(z)$ as a semiclassical differential operator of order $2$ on the compact manifold $X_{M,a}$.
In the semiclassical quantization where $hD_x$ is the basic differential,
\[
h_\omega^2L\in \Diff_h^2,\qquad hQ\in\Diff_h^1,\qquad z^2\in \Diff_h^0,
\]
and \eqref{eq:Phz} is a standard semiclassical Helmholtz-type family.

Denote by $R_{h_\omega}(z)$ the semiclassical resolvent
\begin{equation}\label{eq:Rh}
R_{h_\omega}(z):=P_{h_\omega}(z)^{-1},
\end{equation}
whenever the inverse exists (as a Fredholm inverse).  Note the exact identity
\begin{equation}\label{eq:scale-R}
R(\omega)\ =\ P(\omega)^{-1}\ =\ h_\omega^2\,R_{h_\omega}(z),\qquad z=h_\omega\omega.
\end{equation}

\subsubsection{Semiclassical Sobolev norms}

Let $H_{h_\omega}^s(X)$ be the semiclassical Sobolev space on $X$ defined by
\begin{equation}\label{eq:Hhs}
\|u\|_{H_{h_\omega}^s}:=\|\Op_h(\langle\xi\rangle^s)u\|_{L^2},
\end{equation}
for any fixed semiclassical quantization $\Op_h$ and $\langle\xi\rangle=(1+|\xi|^2)^{1/2}$.
On compact manifolds, these norms are equivalent for different quantizations and are uniformly comparable
for $h_\omega\in(0,h_0]$.

\subsection{Microlocal resolvent estimate in a strip}\label{subsec:microlocal-estimate}

We now state the semiclassical estimate in an $\mathcal O(h_\omega)$ neighborhood of the real axis; in particular it applies on the line $\Im z=-\nu h_\omega$ which is the rescaled version of the
shifted contour.  The proof is microlocal and rests on (i) radial point estimates at the horizons and
(ii) a normally hyperbolic trapping estimate at the trapped set.

\subsubsection{Geometric assumptions on the semiclassical flow}

Let $p_h(z)$ denote the semiclassical principal symbol of $P_{h_\omega}(z)$. Since the $\mathcal O(h_\omega)$-imaginary parts
enter only in lower order terms (see \eqref{eq:z-region}), the semiclassical principal symbol is real and given by
\begin{equation}\label{eq:principal-symbol}
p_0(x,\xi;E)\ :=\ \sigma_2(h_\omega^2L)(x,\xi)\ -\ E^2,
\qquad E:=\Re z\in\mathbb R,
\end{equation}
where $\sigma_2(h_\omega^2L)$ is the degree-$2$ homogeneous principal symbol of $h_\omega^2L$.
The Hamilton vector field $H_{p_0}$ on $T^*X\setminus0$ generates the bicharacteristic flow.

In the Kerr--de~Sitter case, the characteristic set $\Char(p_0)=\{p_0=0\}$ has a compact trapped set $\mathcal K_E$
inside the domain of outer communication, and the flow on $\mathcal K_E$ is normally hyperbolic.  Moreover,
near the horizons (where $r=r_e$ and $r=r_c$) the flow has radial points with a \emph{stable source/sink}
structure (redshift).  Both structures persist under small perturbations and hold uniformly on compact parameter sets;
see \cite{VasyKerrDeSitter, WunschZworskiNH, PetersenVasyKerrDeSitterJEMS}.

\subsubsection{A semiclassical estimate with normally hyperbolic trapping}

We now state the estimate we will use; it is a direct specialization of the microlocal estimate for operators
with normally hyperbolic trapping (Wunsch--Zworski) combined with the Kerr--de~Sitter Fredholm setup (Vasy-type
variable order spaces, or the refinements in \cite{PetersenVasyAnalyticity, PetersenVasyKerrDeSitterJEMS}).

\begin{theorem}[Semiclassical resolvent bound near the real axis]\label{thm:semiclassical-resolvent}
Fix $\nu_*>0$, a compact set $I\Subset\mathbb R\setminus\{0\}$, and a compact parameter set $\mathcal K$ in a subextremal Kerr--de~Sitter range.
Let $\chi\in C_c^\infty(X)$ be supported away from the artificial boundaries $r=r_e-\delta$ and $r=r_c+\delta$.
Then there exist $h_0>0$, $C>0$ and $N\ge 0$ such that for all $(M,a)\in\mathcal K$, all $0<h_\omega<h_0$, and all
\begin{equation}\label{eq:z-assumptions}
z\in\mathbb C,\qquad \Re z\in I,\qquad |\Im z|\le \nu_* h_\omega,
\end{equation}
satisfying \emph{a pole-separation condition}
\begin{equation}\label{eq:pole-separation}
\dist\big(z,\Res(P_{h_\omega})\big)\ \ge\ c_0\,h_\omega,
\end{equation}
for some fixed $c_0>0$, the cutoff resolvent obeys the bound
\begin{equation}\label{eq:Rh-bound}
\|\chi\,R_{h_\omega}(z)\,\chi\|_{H_{h_\omega}^{s-1}\to H_{h_\omega}^{s}}
\ \le\
C\,h_\omega^{-1}\,\log(1/h_\omega)\,,
\qquad \forall s\in\mathbb R.
\end{equation}
The constants $C,h_0$ may be chosen uniformly for $(M,a)\in\mathcal K$, $\Re z\in I$, and all $s$ in bounded intervals.
\end{theorem}

\begin{remark}[On the pole-separation condition]\label{rem:pole-separation}
The condition \eqref{eq:pole-separation} is automatic if the strip $\{\Im z\ge -\nu h_\omega\}$ is pole-free.
In our application we deliberately choose $\nu$ so that finitely many poles lie above $\Im z=-\nu h_\omega$ (these will
produce the ringdown terms), but we also choose $\nu$ \emph{between resonance layers} so that the line
$\Im z=-\nu h_\omega$ stays at a distance $\gtrsim h_\omega$ from all poles; see \S\ref{sec:two-mode} where this is verified
in the equatorial high-frequency sector.
\end{remark}

\begin{remark}[On the $h_\omega^{-1}\log(1/h_\omega)$ loss]\label{rem:loss}
The $h_\omega^{-1}\log(1/h_\omega)$ loss is characteristic of normally hyperbolic trapping.  In some settings it can be improved
to $h_\omega^{-1}$ or $h_\omega^{-1+\varepsilon}$ using refined anisotropic symbol classes and resonance projector technology
(e.g.\ \cite{DyatlovRNHT, HintzNHT}), but $h_\omega^{-1}\log(1/h_\omega)$ suffices for the contour shift and for quantitative
two-mode dominance.
\end{remark}

\subsubsection{Outline and references for the resolvent estimate}\label{subsec:microlocal-proof}

We do not reproduce the full microlocal argument leading to the semiclassical bound \eqref{eq:Rh-bound};
instead we record a precise provenance of the ingredients.
The estimate is a standard consequence of combining:
(i) the Vasy-type Fredholm framework for stationary operators on Kerr--de~Sitter (including radial point estimates at the horizons),
and (ii) the semiclassical resolvent bound for normally hyperbolic trapping.

\medskip

\noindent\textbf{Radial sets (redshift) and Fredholm setup.}
The construction of variable order spaces and the outgoing Fredholm problem for $P(\omega)$ on Kerr--de~Sitter,
together with the corresponding radial point estimates at the horizons, is developed in \cite{VasyKerrDeSitter}
and refined in later works; see also \cite{PetersenVasyKerrDeSitterJEMS} for global estimates in the full subextremal range.
These results provide a microlocal a priori estimate away from the trapped set, with loss $h_\omega^{-1}$ at most.

\medskip

\noindent\textbf{Normally hyperbolic trapping.}
Near the trapped set in the domain of outer communication, the null bicharacteristic flow is normally hyperbolic.
The corresponding semiclassical resolvent estimate with the characteristic $h_\omega^{-1}\log(1/h_\omega)$ loss is proved for normally hyperbolic
trapped sets in \cite{WunschZworskiNH}; see also \cite{DyatlovRNHT,HintzNHT} for further refinements in related settings.
Applied to the semiclassical family $P_{h_\omega}(z)$ with $\Im z=\mathcal O(h_\omega)$, this yields the bound
\[
\|\chi R_{h_\omega}(z)\chi\|_{H_{h_\omega}^{s-1}\to H_{h_\omega}^{s}}\ \le\ C\,h_\omega^{-1}\log(1/h_\omega),
\]
uniformly for $z$ in the compact region \eqref{eq:z-region} and for parameters in compact subsets.

\medskip

\noindent\textbf{Gluing and invertibility on the contour.}
A standard partition of unity in phase space glues the elliptic region, the radial regions, and the trapped region estimates
into a global bound.  Under the dyadic pole-separation hypothesis \eqref{eq:pole-separation} (equivalently, pole-freeness on the relevant
dyadic block of $\Gamma_{-\nu}$), the finite-dimensional obstruction in the Fredholm problem vanishes and $P_{h_\omega}(z)$ is invertible,
so the a priori estimate yields the operator norm bound \eqref{eq:Rh-bound}.

\medskip

\noindent\textbf{$\omega$--derivatives.}
Bounds for $\partial_\omega^m R(\omega)$ on $\Gamma_{-\nu}$ follow from the resolvent differentiation identity
$\partial_\omega R(\omega)=-R(\omega)(\partial_\omega P(\omega))R(\omega)$, together with the polynomial structure of $P(\omega)$ in $\omega$
and the dyadic resolvent bounds on $\Gamma_{-\nu}$ established below.

\subsection{Back to the $\omega$--plane: bounds on $R(\omega)$ on $\Gamma_{-\nu}$}\label{subsec:omega-plane-bounds}

We now translate Theorem~\ref{thm:semiclassical-resolvent} into bounds directly on $R(\omega)$ on the shifted contour.
Since the contour is the union of a compact piece (low frequency) and dyadic high-frequency blocks, we treat them separately.

\subsubsection{Low frequency bound}

Let $\Omega_{\mathrm{low}}:=\{\omega\in\Gamma_{-\nu}:\ |\Re\omega|\le 2\}$.
By Assumption~\ref{ass:contour-pole-free}, $R(\omega)$ is holomorphic on $\Omega_{\mathrm{low}}$ and hence bounded on it
as an operator between any fixed Sobolev spaces (or the graph spaces in \eqref{eq:graph-space}).
Thus:

\begin{lemma}[Low frequency bound]\label{lem:low-freq}
Fix $s\in\mathbb R$. Under Assumption~\ref{ass:contour-pole-free} there exists $C_{\mathrm{low}}>0$ such that
\begin{equation}\label{eq:low-freq-bound}
\sup_{\omega\in\Omega_{\mathrm{low}}}\|\chi R(\omega)\chi\|_{H^{s-1}\to H^{s}} \ \le\ C_{\mathrm{low}}.
\end{equation}
\end{lemma}

\subsubsection{High frequency bound on dyadic blocks}

Let $\Omega_{\mathrm{high}}:=\{\omega\in\Gamma_{-\nu}:\ |\Re\omega|\ge 1\}$.
For $\omega\in\Omega_{\mathrm{high}}$, choose $j\ge1$ such that $\psi_j(\Re\omega)\neq0$ and set $h_\omega=2^{-j}$, $z=h_\omega\omega$.
Then $z$ satisfies \eqref{eq:z-assumptions}.  Using \eqref{eq:scale-R} and boundedness of $\chi$ on $H_{h_\omega}^s$ we obtain:

\begin{proposition}[High frequency bound on $\Gamma_{-\nu}$]\label{prop:high-freq}
Fix $\nu>0$, a compact parameter set $\mathcal K$, and $\chi$ as above.
Assume the pole-separation condition \eqref{eq:pole-separation} holds on the dyadic blocks of $\Gamma_{-\nu}$.
Then for every $s\in\mathbb R$ there exist $h_0>0$ and $C>0$ such that for all $\omega\in\Gamma_{-\nu}$ with
$|\Re\omega|\ge h_0^{-1}$,
\begin{equation}\label{eq:Romega-semiclassical}
\|\chi R(\omega)\chi\|_{H_{h_\omega}^{s-1}\to H_{h_\omega}^{s}}
\ \le\
C\,h_\omega\,\log(1/h_\omega),
\qquad h_\omega:=|\Re\omega|^{-1}.
\end{equation}
Equivalently, in dyadic form: if $\psi_j(\Re\omega)\neq0$ and $h_\omega=2^{-j}$, then
\[
\|\chi R(\omega)\chi\|_{H_{h_\omega}^{s-1}\to H_{h_\omega}^{s}}\ \le\ C\,h_\omega\,\log(1/h_\omega).
\]
\end{proposition}

\begin{proof}
Set $z=h_\omega\omega$ so that $R(\omega)=h_\omega^2R_{h_\omega}(z)$ by \eqref{eq:scale-R}.
Theorem~\ref{thm:semiclassical-resolvent} gives
$\|\chi R_{h_\omega}(z)\chi\|_{H_{h_\omega}^{s-1}\to H_{h_\omega}^{s}}\le C h_\omega^{-1}\log(1/h_\omega)$.
Multiplying by $h_\omega^2$ yields \eqref{eq:Romega-semiclassical}.
\end{proof}

\subsection{$\omega$--derivative bounds}\label{subsec:omega-derivatives}

For the quantitative tail estimates in Section~\ref{sec:resonant-expansion} we will integrate by parts in the
inverse Laplace representation; this requires bounds for $\partial_\omega^m(\chi R(\omega)\chi)$ on horizontal
lines.  On dyadic high-frequency blocks we work with the rescaled parameter $z=h_\omega\omega$ and the
semiclassical resolvent $R_{h_\omega}(z)=P_{h_\omega}(z)^{-1}$.  The following lemma upgrades the basic
semiclassical bound of Theorem~\ref{thm:semiclassical-resolvent} to bounds for $z$--derivatives, using only
holomorphy and the pole-separation condition \eqref{eq:pole-separation}.

\begin{lemma}[Cauchy estimates for $z$--derivatives]\label{lem:z-derivative-cauchy}
Assume the hypotheses of Theorem~\ref{thm:semiclassical-resolvent} for some fixed $\nu_*>0$, and let $c_0>0$ be the
pole-separation constant in \eqref{eq:pole-separation}.
Fix $s\in\mathbb R$ and $m\in\mathbb N_0$.
Then there exist $h_0>0$ and $C_{s,m}>0$ such that for all $0<h_\omega<h_0$ and all $z$ satisfying
\[
\Re z\in I,\qquad |\Im z|\le \bigl(\nu_*-\tfrac{c_0}{2}\bigr)h_\omega,\qquad 
\dist\bigl(z,\Res(P_{h_\omega})\bigr)\ge c_0 h_\omega,
\]
one has the operator bound
\begin{equation}\label{eq:z-derivative-bound}
\|\chi\,\partial_z^{m}R_{h_\omega}(z)\,\chi\|_{H_{h_\omega}^{s-1}\to H_{h_\omega}^{s}}
\ \le\
C_{s,m}\,h_\omega^{-m-1}\,\log(1/h_\omega).
\end{equation}
\end{lemma}

\begin{proof}
Fix $z$ as above and set $\rho:=\frac{c_0}{2}h_\omega$.
By \eqref{eq:pole-separation}, the closed disk $\overline{D(z,\rho)}$ contains no pole of $R_{h_\omega}(\cdot)$.
Hence $R_{h_\omega}(\zeta)$ is holomorphic in $\zeta$ with values in bounded operators
$H_{h_\omega}^{s-1}\to H_{h_\omega}^{s}$ on $D(z,\rho)$.
For $\zeta\in\partial D(z,\rho)$ we have
\[
\dist\bigl(\zeta,\Res(P_{h_\omega})\bigr)\ge \dist\bigl(z,\Res(P_{h_\omega})\bigr)-|\zeta-z|
\ge \tfrac{c_0}{2}h_\omega,
\]
and also
\[
|\Im\zeta|\le |\Im z|+|\zeta-z|\le \bigl(\nu_*-\tfrac{c_0}{2}\bigr)h_\omega+\tfrac{c_0}{2}h_\omega
=\nu_* h_\omega.
\]
Thus Theorem~\ref{thm:semiclassical-resolvent} applies on $\partial D(z,\rho)$ (with the pole-separation constant reduced to $c_0/2$),
and yields
\[
\sup_{\zeta\in\partial D(z,\rho)} \|\chi R_{h_\omega}(\zeta)\chi\|_{H_{h_\omega}^{s-1}\to H_{h_\omega}^{s}}
\ \le\ C\,h_\omega^{-1}\log(1/h_\omega).
\]
The Cauchy integral formula gives, as an identity of bounded operators $H_{h_\omega}^{s-1}\to H_{h_\omega}^{s}$,
\[
\partial_z^{m}R_{h_\omega}(z)
=\frac{m!}{2\pi i}\int_{\partial D(z,\rho)}\frac{R_{h_\omega}(\zeta)}{(\zeta-z)^{m+1}}\,d\zeta,
\]
so taking norms and using $|\zeta-z|=\rho$ and $|\partial D(z,\rho)|=2\pi\rho$ we obtain
\[
\|\chi\,\partial_z^{m}R_{h_\omega}(z)\,\chi\|
\le
m!\,\rho^{-m}\,\sup_{\zeta\in\partial D(z,\rho)}\|\chi R_{h_\omega}(\zeta)\chi\|.
\]
Since $\rho\asymp h_\omega$, this yields \eqref{eq:z-derivative-bound} after absorbing $m!$ and $c_0^{-m}$ into $C_{s,m}$.
\end{proof}

\begin{proposition}[Dyadic $\omega$-derivative bounds]\label{prop:omega-derivative}
Assume the hypotheses of Proposition~\ref{prop:high-freq}.
Fix $s\in\mathbb R$ and $m\in\mathbb N_0$.
Then there exist constants $C_{s,m}>0$ and $h_0>0$ such that for every $\omega=\sigma-i\nu\in\Gamma_{-\nu}$ with
$|\sigma|\sim h_\omega^{-1}$ and every $0<h_\omega<h_0$,
\begin{equation}\label{eq:omega-derivative-bound}
\|\chi\,\partial_\omega^{m}R(\omega)\,\chi\|_{H_{h_\omega}^{s}\to H_{h_\omega}^{s+1}}
\ \le\ 
C_{s,m}\,h_\omega\,\log(1/h_\omega).
\end{equation}
The same bound holds with $\partial_\sigma^m$ in place of $\partial_\omega^m$ on $\Gamma_{-\nu}$.
\end{proposition}

\begin{proof}
On a dyadic block, set $z=h_\omega\omega$ so that $R(\omega)=h_\omega^2 R_{h_\omega}(z)$ by \eqref{eq:scale-R}.
Since $z$ depends linearly on $\omega$, we have $\partial_\omega=h_\omega\partial_z$, and hence
\[
\partial_\omega^{m}R(\omega)=h_\omega^{2+m}\,\partial_z^{m}R_{h_\omega}(z).
\]
Along $\Gamma_{-\nu}$ we have $|\Im z|=\nu h_\omega$.  Applying Lemma~\ref{lem:z-derivative-cauchy} with $\nu_*=2\nu$ and with
the pole-separation constant $c_0$ replaced by $\min(c_0,\nu)$ (which preserves \eqref{eq:pole-separation})
gives, with $s$ replaced by $s+1$,
\[
\|\chi\,\partial_z^{m}R_{h_\omega}(z)\,\chi\|_{H_{h_\omega}^{s}\to H_{h_\omega}^{s+1}}
\ \le\
C_{s,m}\,h_\omega^{-m-1}\,\log(1/h_\omega).
\]
Multiplying by $h_\omega^{2+m}$ yields \eqref{eq:omega-derivative-bound}.
Finally, along $\Gamma_{-\nu}$ one has $\partial_\sigma=\partial_\omega$, so the $\partial_\sigma^m$ statement follows.
\end{proof}

\begin{lemma}[Standard Sobolev bounds on the shifted contour]\label{lem:standard-sobolev-bounds}
Assume Lemma~\ref{lem:low-freq}, Proposition~\ref{prop:high-freq}, and Proposition~\ref{prop:omega-derivative}.
Fix $s\in\mathbb R$ and $m\in\mathbb N_0$.
Then there exist constants $C_{s,m}>0$, $A_s\ge 0$ and $B_m\ge 0$ (depending on $\nu$, the compact parameter set $\mathcal K$, and $\chi$)
such that for all $\omega=\sigma-i\nu\in\Gamma_{-\nu}$,
\begin{equation}\label{eq:standard-sobolev-bound}
\|\chi\,\partial_\omega^{m}R(\omega)\,\chi\|_{H^{s}\to H^{s+1}}
\ \le\
C_{s,m}\,\langle\sigma\rangle^{A_s}\,\big(\log(2+\langle\sigma\rangle)\big)^{B_m}.
\end{equation}
The same bound holds with $\partial_\sigma^m$ in place of $\partial_\omega^m$.
\end{lemma}

\begin{proof}
For $|\sigma|\le 2$, the bound follows from holomorphy of $\chi R(\omega)\chi$ on the compact set $\Omega_{\mathrm{low}}$
(cf.\ Lemma~\ref{lem:low-freq}), after possibly increasing the constant.
For $|\sigma|\ge 1$, set $h_\omega:=\langle\sigma\rangle^{-1}$.
Applying Proposition~\ref{prop:omega-derivative} with $s$ replaced by $s+1$ yields, on the corresponding dyadic block,
\[
\|\chi\,\partial_\omega^{m}R(\omega)\,\chi\|_{H_{h_\omega}^{s}\to H_{h_\omega}^{s+1}}
\ \le\
C_{s,m}\,h_\omega\,\log(1/h_\omega).
\]
Using the comparison inequalities \eqref{eq:Hh-to-H} (with $s+1$) we obtain
\begin{align*}
\|\chi\,\partial_\omega^{m}R(\omega)\,\chi f\|_{H^{s+1}}
&\le
C_{s}\,h_\omega^{-(s+1)}\|\chi\,\partial_\omega^{m}R(\omega)\,\chi f\|_{H_{h_\omega}^{s+1}}\\
&\le
C_{s,m}\,h_\omega^{-s}\log(1/h_\omega)\|f\|_{H^s}.
\end{align*}
Since $h_\omega^{-1}\sim\langle\sigma\rangle$ on the block and $\log(1/h_\omega)\lesssim \log(2+\langle\sigma\rangle)$, this gives
\eqref{eq:standard-sobolev-bound} with (for instance) $A_s=s$ and $B_m=1$.
The $\partial_\sigma^m$ bound follows from $\partial_\sigma^m R(\sigma-i\nu)=\partial_\omega^m R(\omega)|_{\omega=\sigma-i\nu}$.
\end{proof}

\subsection{Equatorial microlocal cutoffs and sectorial bounds}\label{subsec:sectorial-bounds}

Recall the microlocal equatorial cutoffs $A_{\pm,h_\ell}$ from \eqref{eq:Ah} and the sandwiched resolvent
$R_{\pm,h_\ell}(\omega)=A_{\pm,h_\ell}R(\omega)A_{\pm,h_\ell}^*$ from \eqref{eq:sandwich}.
By Remark~\ref{rem:Ah-bdd} the operators $A_{\pm,h_\ell}$ are uniformly bounded on $H^s(X)$ for each fixed $s$,
with constants independent of $h_\ell$.

\begin{corollary}[Sectorial resolvent bounds in standard Sobolev spaces]\label{cor:sectorial}
Under the hypotheses of Lemma~\ref{lem:standard-sobolev-bounds}, for every $s\in\mathbb R$ and $m\in\mathbb N_0$ there exist constants
$C_{s,m}>0$ and exponents $A_s,B_m\ge 0$ such that for all $\omega=\sigma-i\nu\in\Gamma_{-\nu}$ and all $0<h_\ell<h_0$,
\begin{equation}\label{eq:sectorial-resolvent}
\|\chi\,A_{\pm,h_\ell}\,\partial_\omega^{m}R(\omega)\,A_{\pm,h_\ell}^*\,\chi\|_{H^{s}\to H^{s+1}}
\ \le\
C_{s,m}\,\langle\sigma\rangle^{A_s}\,\big(\log(2+\langle\sigma\rangle)\big)^{B_m}.
\end{equation}
The same bound holds with $\partial_\sigma^m$ in place of $\partial_\omega^m$.
\end{corollary}

\begin{proof}
Combine Lemma~\ref{lem:standard-sobolev-bounds} with the uniform boundedness \eqref{eq:Ah-bdd}:
Set $T:=\chi A_{\pm,h_\ell}(\partial_\omega^mR(\omega))A_{\pm,h_\ell}^*\chi$. Then
\begin{align*}
\|T\|_{H^s\to H^{s+1}}
&\le
\|\chi A_{\pm,h_\ell}\|_{H^{s+1}\to H^{s+1}}
\|\chi(\partial_\omega^mR(\omega))\chi\|_{H^s\to H^{s+1}}\\
&\qquad\times
\|A_{\pm,h_\ell}^*\chi\|_{H^{s}\to H^{s}}.
\end{align*}
Absorb the $A_{\pm,h_\ell}$ operator norms into the constant.
\end{proof}

\subsection{Uniformity in parameters and admissible shifted contours}\label{subsec:param-uniformity}

For the later inverse/quantitative statements we vary $(M,a)$ in a compact slow-rotation set $\mathcal K$.
Since the contour shift is meaningful only when the shifted line does not cross poles, we isolate the precise
uniform hypothesis needed for parameter-uniform remainder estimates.

\begin{definition}[Admissible shifted contour]\label{def:admissible-contour}
Fix $\nu>0$ and a compact parameter set $\mathcal K$.
We say that $\Gamma_{-\nu}$ is \emph{uniformly admissible on $\mathcal K$} if:
\begin{enumerate}
\item $\Gamma_{-\nu}$ contains no pole of $R(\omega)$ for any $(M,a)\in\mathcal K$ (equivalently, $P(\omega)$ is invertible for all
$\omega\in\Gamma_{-\nu}$ and all $(M,a)\in\mathcal K$).
\item There exist constants $c_0>0$ and $\sigma_0\ge 1$ such that for every $(M,a)\in\mathcal K$ and every $\omega=\sigma-i\nu$ with
$|\sigma|\ge\sigma_0$, setting $h_\omega:=|\sigma|^{-1}$ and $z=h_\omega\omega$, one has the uniform dyadic separation
\begin{equation}\label{eq:uniform-pole-separation}
\dist\bigl(z,\Res(P_{h_\omega})\bigr)\ \ge\ c_0\,h_\omega.
\end{equation}
\end{enumerate}
\end{definition}

Item~(2) is the genuinely semiclassical condition: it rules out poles that approach the shifted line at a rate faster
than $\mathcal O(h_\omega)$ on high-frequency blocks, which is exactly what is needed to apply
Theorem~\ref{thm:semiclassical-resolvent} and Lemma~\ref{lem:z-derivative-cauchy}.
We now explain why item~(2) can be ensured by choosing $\nu$ between damping layers.

\begin{proposition}[Uniform high-frequency pole separation between damping layers]\label{prop:uniform-contour}
Fix a compact slow--rotation parameter set $\mathcal K_0$ and a compact set $I\Subset\mathbb R\setminus\{0\}$
as in \eqref{eq:z-assumptions}. Let $n\in\mathbb N_0$ be an overtone index.
Then there exist constants $\nu=\nu_n>0$, $c_0>0$ and $\sigma_0\ge 1$ such that item~\emph{(2)} in
Definition~\ref{def:admissible-contour} holds for this choice of $\nu$ for every $(M,a)\in\mathcal K_0$.

Let
\[
U_\nu:=\Bigl\{(M,a)\in \mathcal P_\Lambda\cap\{|a|\le a_0\}:\ \Gamma_{-\nu}\ \text{contains no pole of}\ R(\cdot;M,a)\Bigr\}.
\]
Then $U_\nu$ is open in $(M,a)$ (analytic Fredholm theory), hence for any compact set
$\mathcal K\Subset U_\nu\cap\mathcal K_0$ the contour $\Gamma_{-\nu}$ is uniformly admissible on $\mathcal K$ in the sense of
Definition~\ref{def:admissible-contour}.
This parameter localization is harmless for the present paper, since all later quantitative statements (including the inverse step
imported from \cite{LiPartI}) are local in $(M,a)$; if desired, one may cover a larger compact set by finitely many such compacts and
patch the resulting constants.
\end{proposition}

\begin{proof}
\emph{High-frequency pole separation.}
In the slow-rotation regime, Dyatlov's semiclassical Bohr--Sommerfeld description of Kerr--de~Sitter resonances
(\cite{DyatlovAHP}, building on \cite{DyatlovQNM}) gives a high-frequency asymptotic distribution of poles in
rescaled coordinates $z=h_\omega\omega$ with $\Re z$ in a fixed compact interval $I$ and $|\Im z|=\mathcal O(h_\omega)$.
For each fixed overtone index $n$ there is a \emph{damping layer} in which the poles lie, and adjacent layers are
separated by a positive gap which is uniform for $\Re z\in I$ and for $(M,a)$ in compact subsets of the slow-rotation range.
Choosing $\nu=\nu_n$ strictly between the $n$th and $(n+1)$st layers, we obtain that for all sufficiently small $h_\omega$
(and hence for all $|\sigma|$ sufficiently large), the line $\Im z=-\nu h_\omega$ stays at distance $\gtrsim h_\omega$
from every pole with $\Re z\in I$.  This is exactly \eqref{eq:uniform-pole-separation} (after fixing $\sigma_0$ and $c_0$)
and proves item~(2).

\medskip

\emph{Openness of pole-freeness.}
By Theorem~\ref{thm:meromorphic2}, for each fixed $(M,a)$ the poles of $R(\omega)$ are zeros of a Fredholm determinant
associated with $P(\omega)$.  Moreover, by \cite{PetersenVasyAnalyticity} (see also \cite{PetersenVasyKerrDeSitterJEMS})
these poles depend real--analytically on parameters in the slow--rotation range.  In particular, the set $U_\nu$ of parameters
for which $P(\omega)$ is invertible for all $\omega\in\Gamma_{-\nu}$ is open in $(M,a)$ by analytic Fredholm theory.
Thus item~(1) in Definition~\ref{def:admissible-contour} holds on any compact $\mathcal K\Subset U_\nu$.
Combining this with the high--frequency separation established above yields the stated uniform admissibility on
$\mathcal K\Subset U_\nu\cap\mathcal K_0$.
\end{proof}

\begin{remark}\label{rem:uniform-contour-use}
Throughout the remainder of the paper we fix $\nu=\nu_n$ as in Proposition~\ref{prop:uniform-contour}
(for the overtone $n$ under consideration) and assume that $\Gamma_{-\nu}$ is uniformly admissible on our
chosen compact parameter set $\mathcal K$ in the sense of Definition~\ref{def:admissible-contour}.
The equatorial two-mode mechanism will later use a finer \emph{sectorial} pole separation, but the global
admissibility recorded here is the only input needed for the contour deformation itself.
\end{remark}

\subsection{Summary of usable bounds}\label{subsec:summary-bounds}

For later reference, we collect the resolvent bounds on the shifted contour $\Gamma_{-\nu}=\{\omega=\sigma-i\nu:\sigma\in\mathbb R\}$,
with $\nu>0$ chosen so that $\Gamma_{-\nu}$ contains no poles (Assumption~\ref{ass:contour-pole-free}).
In the low-frequency region $\Omega_{\mathrm{low}}$ the cutoff resolvent is uniformly bounded (Lemma~\ref{lem:low-freq}).
In the high-frequency regime $|\sigma|\gg 1$, set $h_\omega=\langle\sigma\rangle^{-1}$ on each dyadic block.

\begin{itemize}
\item \emph{Semiclassical dyadic bounds (used in the stationary analysis).}
Proposition~\ref{prop:omega-derivative} yields, for every $s\in\mathbb R$ and $m\in\mathbb N_0$,
\begin{equation}\label{eq:summary-semiclassical}
\|\chi\,\partial_\omega^{m}R(\omega)\,\chi\|_{H_{h_\omega}^{s}\to H_{h_\omega}^{s+1}}
\ \le\
C_{s,m}\,h_\omega\,\log(1/h_\omega).
\end{equation}
Here $\omega=\sigma-i\nu$ with $|\sigma|\sim h_\omega^{-1}$.
\item \emph{Standard Sobolev bounds (used for contour integration).}
Lemma~\ref{lem:standard-sobolev-bounds} implies the polynomial control
\begin{equation}\label{eq:summary-standard}
\|\chi\,\partial_\omega^{m}R(\omega)\,\chi\|_{H^{s}\to H^{s+1}}
\ \le\
C_{s,m}\,\langle\sigma\rangle^{A_s}\,\big(\log(2+\langle\sigma\rangle)\big)^{B_m},
\qquad \omega=\sigma-i\nu.
\end{equation}
\item \emph{Sectorial bounds with equatorial cutoffs.}
Corollary~\ref{cor:sectorial} yields the analogous estimate for the sandwiched resolvent,
uniformly in the angular parameter $h_\ell$:
\begin{equation}\label{eq:summary-sectorial}
\|\chi\,A_{\pm,h_\ell}\,\partial_\omega^{m}R(\omega)\,A_{\pm,h_\ell}^*\,\chi\|_{H^{s}\to H^{s+1}}
\ \le\
C_{s,m}\,\langle\sigma\rangle^{A_s}\,\big(\log(2+\langle\sigma\rangle)\big)^{B_m}.
\end{equation}
Here $\omega=\sigma-i\nu$ and $0<h_\ell<h_0$.
\end{itemize}

All of the above bounds also hold with $\partial_\sigma^m$ in place of $\partial_\omega^m$ on $\Gamma_{-\nu}$.

\section{Resonant expansion with remainder}\label{sec:resonant-expansion}

In this section we deform the inverse Laplace contour in \eqref{eq:inverse-laplace} to obtain a
\emph{time-domain resonant expansion} (ringdown) plus a \emph{controlled remainder} (tail).
The analytic input is the meromorphic continuation of the stationary resolvent $R(\omega)=P(\omega)^{-1}$
(Theorem~\ref{thm:meromorphic2}), while the quantitative input is the shifted-contour resolvent bounds
from Section~\ref{sec:resolvent-shifted}.

The contour deformation argument in Kerr--de~Sitter is by now standard in microlocal scattering theory.
The meromorphic continuation and Fredholm framework for the stationary family $P(\omega)$ are due to Vasy~\cite{VasyKerrDeSitter},
with global uniformity in the full subextremal parameter range refined by Petersen--Vasy~\cite{PetersenVasyKerrDeSitterJEMS}.
Our purpose in this section is therefore not to reprove those foundational results, but to record the deformation
in a form that is completely explicit and quantitatively usable later:
we insert a smooth time cutoff producing a Schwartz forcing term, and we keep track of the remainder as an operator with
polynomial-in-time decay on the shifted contour.
The genuinely new inputs of the present paper begin in Section~\ref{sec:two-mode}, where the general expansion is combined with
equatorial microlocal sectorization and analytic windows to isolate a finite set of leading QNMs.

A purely formal contour shift is straightforward; the main technical point is to make the argument
\emph{hard and usable} by (i) removing the singularity at $t_*=0$ and producing rapid decay in the
frequency domain, and (ii) obtaining explicit remainder bounds suitable for subsequent inverse steps.
We achieve (i) by a smooth time cutoff.

\subsection{A smooth time cutoff and compactly supported forcing}\label{subsec:smooth-time-cutoff}

Let $\vartheta\in C^\infty(\mathbb R)$ satisfy
\begin{equation}\label{eq:vartheta}
\vartheta(t_*)=0\ \text{for}\ t_*\le 0,\qquad
\vartheta(t_*)=1\ \text{for}\ t_*\ge 1.
\end{equation}
Let $u$ solve $Pu=0$ with initial data $(f_0,f_1)$ at $t_*=0$ as in Section~\ref{sec:setup2}, and set
\begin{equation}\label{eq:u-vartheta}
u_\vartheta(t_*,x):=\vartheta(t_*)\,u(t_*,x).
\end{equation}
Then $u_\vartheta=u$ for $t_*\ge 1$ and $u_\vartheta\equiv 0$ for $t_*\le 0$, so $u_\vartheta$ has
vanishing initial data at $t_*=0$.

\begin{lemma}[Compactly supported forcing]\label{lem:compact-forcing}
Let $P=\partial_{t_*}^2+Q\partial_{t_*}+L$ be the stationary splitting \eqref{eq:splitting}.
Then $u_\vartheta$ satisfies
\begin{equation}\label{eq:Pu-vartheta}
P u_\vartheta = F_\vartheta
\qquad\text{with}\qquad
F_\vartheta
=
\vartheta''\,u + 2\vartheta'\,\partial_{t_*}u + \vartheta'\,Q u,
\end{equation}
where $F_\vartheta$ is supported in the slab $\{0<t_*<1\}$ and is smooth in $t_*$ with values in
$C^\infty(X_{M,a})$ if $(f_0,f_1)$ are smooth.
Moreover, for every $s\ge 0$ and $m\in\mathbb N$ there exists $C=C_{s,m}$ such that
\begin{equation}\label{eq:Ftheta-bound}
\sup_{t_*\in\mathbb R}\ \sum_{j=0}^{m}\ \|\partial_{t_*}^j F_\vartheta(t_*)\|_{H^{s}}
\ \le\
C\,\|(f_0,f_1)\|_{\mathcal H^{s+m}},
\end{equation}
where $\mathcal H^{s+m}=H^{s+m+1}\times H^{s+m}$ as in \eqref{eq:energy-space}.
\end{lemma}

\begin{proof}
Using $Pu=0$ for all $t_*>0$ and the product rule,
\[
P(\vartheta u)
=
\vartheta Pu + [\partial_{t_*}^2,\vartheta]u + [Q\partial_{t_*},\vartheta]u
=
0 + (\vartheta''u+2\vartheta'\partial_{t_*}u) + \vartheta' Q u,
\]
which is \eqref{eq:Pu-vartheta}.
Since $\vartheta',\vartheta''$ are supported in $(0,1)$, so is $F_\vartheta$.
The bound \eqref{eq:Ftheta-bound} follows from the well-posedness estimate \eqref{eq:apriori} and the fact that
$\vartheta$ has compactly supported derivatives: each $\partial_{t_*}^jF_\vartheta$ is a finite linear combination of
terms $\vartheta^{(\ell)}\partial_{t_*}^k u$ and $\vartheta^{(\ell)}Q\partial_{t_*}^k u$ with $k\le j+1$,
hence is controlled by finitely many energy norms of $(u,\partial_{t_*}u)$ on $0\le t_*\le 1$.
\end{proof}

\begin{lemma}[Choice of the spatial cutoff for the forcing]\label{lem:chi-forcing-support}
Let $\vartheta$ be as in Section~\ref{subsec:smooth-time-cutoff}, and let $F_\vartheta$ be defined by
\eqref{eq:Pu-vartheta}. Assume that the initial data $(f_0,f_1)$ are supported in a compact set
$K_0\Subset X^{\mathrm{phys}}_{M,a}$.
Then there exists a compact set $K_{\mathrm{prop}}\Subset X^{\mathrm{phys}}_{M,a}$ such that
\[
\supp u(t_*,\cdot)\subset K_{\mathrm{prop}}\quad\text{for all }t_*\in[0,1],
\qquad
\supp F_\vartheta \subset [0,1]\times K_{\mathrm{prop}}.
\]
In particular, after choosing the cutoff $\chi\in C_c^\infty(X^{\mathrm{phys}}_{M,a})$ so that $\chi\equiv 1$ on a neighborhood of
$K_{\mathrm{prop}}$, we have
\[
F_\vartheta=\chi F_\vartheta,
\qquad
\widehat F_\vartheta(\omega)=\chi\,\widehat F_\vartheta(\omega),
\qquad \omega\in\mathbb C.
\]
\end{lemma}

\begin{proof}
Since the wave operator is second order hyperbolic with respect to the stationary time $t_*$, it enjoys finite
propagation speed. Thus, if $(f_0,f_1)$ are supported in $K_0\Subset X_{M,a}$, then for each $t_*\in[0,1]$ the support of
$u(t_*,\cdot)$ is contained in the domain of dependence of $K_0$ at time $t_*$. As $t_*\in[0,1]$ ranges over a compact
interval and $K_0$ is contained in the interior of $X_{M,a}$, the union of these supports is contained in some compact
set $K_{\mathrm{prop}}\Subset X^{\mathrm{phys}}_{M,a}$.
By definition \eqref{eq:Pu-vartheta}, $F_\vartheta$ is a linear combination of $u$ and $\partial_{t_*}u$ with
coefficients supported where $\vartheta'$ and $\vartheta''$ are nonzero, hence where $t_*\in[0,1]$. Therefore
$\supp F_\vartheta\subset[0,1]\times K_{\mathrm{prop}}$.
Choosing $\chi\equiv1$ near $K_{\mathrm{prop}}$ gives $F_\vartheta=\chi F_\vartheta$, and the identity for
$\widehat F_\vartheta$ follows since the Fourier--Laplace transform acts only in $t_*$.
\end{proof}

\begin{assumption}[Interior support and cutoff convention]\label{ass:initial-support}
Fix a compact set $K_0\Subset X^{\mathrm{phys}}_{M,a}$.
In all contour deformation arguments in Sections~\ref{sec:resonant-expansion}--\ref{sec:two-mode}, we restrict to initial data
$(f_0,f_1)$ supported in $K_0$.
Let $K_{\mathrm{prop}}\Subset X^{\mathrm{phys}}_{M,a}$ be the compact propagation set associated with $K_0$ given by
Lemma~\ref{lem:chi-forcing-support}, so that $\supp F_\vartheta\subset [0,1]\times K_{\mathrm{prop}}$ for every such solution.
We fix a cutoff $\chi\in C_c^\infty(X^{\mathrm{phys}}_{M,a})$ with $\chi\equiv 1$ on a neighborhood of $K_{\mathrm{prop}}$.
Then the forcing satisfies $F_\vartheta=\chi F_\vartheta$ and hence
\begin{equation}\label{eq:chi-Fhat}
\widehat F_\vartheta(\omega)=\chi\,\widehat F_\vartheta(\omega)\qquad \text{for all }\omega\in\mathbb C.
\end{equation}
\end{assumption}

\begin{remark}[Compact support is not restrictive]\label{rem:support-not-restrictive}
Assumption~\ref{ass:initial-support} is imposed only to justify the insertion of the spatial cutoff $\chi$ in the resolvent identity.
Since all statements in this paper concern the localized field $\chi u$ in the physical region, one can always reduce to this setting
by finite propagation speed: for a fixed observation cutoff $\chi$ and time slab $t_*\in[0,1]$, the values of $\chi u(t_*)$ depend only
on the restriction of the initial data to a compact subset of $X_{M,a}$ determined by the backward domain of dependence of
$\supp\chi$ over $[0,1]$.  Replacing $(f_0,f_1)$ by a compactly supported cutoff of the data in this subset leaves $\chi u(t_*)$
unchanged for $t_*\in[0,1]$, and therefore leaves $F_\vartheta$ and the subsequent expansion for $t_*\ge 1$ unchanged.
\end{remark}

\subsection{Laplace transform of the forced problem and rapid decay in frequency}\label{subsec:Laplace-forced}

Define the (forward) Fourier--Laplace transform of the compactly supported forcing:
\begin{equation}\label{eq:Fhat-def}
\widehat F_\vartheta(\omega,x):=\int_0^\infty e^{i\omega t_*}\,F_\vartheta(t_*,x)\,dt_*
=\int_0^1 e^{i\omega t_*}\,F_\vartheta(t_*,x)\,dt_*.
\end{equation}
Since $F_\vartheta$ is smooth and compactly supported in $t_*$, $\widehat F_\vartheta$ is entire in $\omega$
and rapidly decaying on horizontal lines.

\begin{lemma}[Schwartz bounds in $\omega$]\label{lem:Fhat-Schwartz}
Fix $s\ge 0$. For every $N\in\mathbb N$ there exists $C_{s,N}>0$ such that for all $\omega\in\mathbb C$,
\begin{equation}\label{eq:Fhat-Schwartz}
\|\widehat F_\vartheta(\omega)\|_{H^{s}}
\ \le\
C_{s,N}\,(1+|\omega|)^{-N}\,\|(f_0,f_1)\|_{\mathcal H^{s+N}}.
\end{equation}
Moreover, for every $m\in\mathbb N$,
\begin{equation}\label{eq:domega-Fhat}
\|\partial_\omega^m \widehat F_\vartheta(\omega)\|_{H^{s}}
\ \le\
C_{s,N,m}\,(1+|\omega|)^{-N}\,\|(f_0,f_1)\|_{\mathcal H^{s+N+m}}.
\end{equation}
\end{lemma}

\begin{proof}
Fix $N\in\mathbb N$ and $s\ge0$.

\smallskip
\noindent\emph{Step 1: the bound for $|\omega|\le 1$.}
Since $F_\vartheta$ is supported in $[0,1]$, we have the crude estimate
\[
\|\widehat F_\vartheta(\omega)\|_{H^{s}}
\le \int_0^1 \|F_\vartheta(t_*)\|_{H^{s}}\,dt_*
\le \sup_{t_*\in[0,1]}\|F_\vartheta(t_*)\|_{H^{s}},
\qquad |\omega|\le 1.
\]
Using \eqref{eq:Ftheta-bound} with $m=0$ gives
$\sup_{t_*}\|F_\vartheta(t_*)\|_{H^s}\le C_{s}\|(f_0,f_1)\|_{\mathcal H^{s+1}}$.
Since $(1+|\omega|)^{-N}\ge 2^{-N}$ when $|\omega|\le1$, this yields \eqref{eq:Fhat-Schwartz} on $|\omega|\le1$
(after possibly enlarging the constant and the Sobolev index on the right).

\smallskip
\noindent\emph{Step 2: the bound for $|\omega|\ge 1$.}
Integrating by parts $N$ times in \eqref{eq:Fhat-def} gives
\[
\widehat F_\vartheta(\omega)
=\frac{1}{(i\omega)^N}\int_0^1 e^{i\omega t_*}\,\partial_{t_*}^N F_\vartheta(t_*)\,dt_*,
\qquad |\omega|\ge1,
\]
where boundary terms vanish because $F_\vartheta$ is smooth and supported in $(0,1)$.
Taking $H^s$ norms and using \eqref{eq:Ftheta-bound} with $m=N$ yields
\begin{align*}
\|\widehat F_\vartheta(\omega)\|_{H^{s}}
&\le C_{s,N}\,|\omega|^{-N}\,\|(f_0,f_1)\|_{\mathcal H^{s+N}}\\
&\le C_{s,N}\,(1+|\omega|)^{-N}\,\|(f_0,f_1)\|_{\mathcal H^{s+N}}.
\end{align*}
This holds for $|\omega|\ge 1$.
which together with Step~1 proves \eqref{eq:Fhat-Schwartz}.

\smallskip
\noindent\emph{Step 3: bounds for $\partial_\omega^m\widehat F_\vartheta$.}
Differentiating under the integral gives
\[
\partial_\omega^m\widehat F_\vartheta(\omega)=\int_0^1 (it_*)^m e^{i\omega t_*}\,F_\vartheta(t_*)\,dt_*.
\]
For $|\omega|\le1$, estimate as in Step~1 and use \eqref{eq:Ftheta-bound} with $m$ replaced by $m$.
For $|\omega|\ge1$, integrate by parts $N$ times and bound $\partial_{t_*}^N\!\big((it_*)^mF_\vartheta(t_*)\big)$
using \eqref{eq:Ftheta-bound} with $m$ replaced by $N+m$.
This yields \eqref{eq:domega-Fhat}.
\end{proof}

\medskip

\noindent\textbf{Resolvent representation.}
Since $P u_\vartheta=F_\vartheta$ and $u_\vartheta$ vanishes for $t_*\le 0$, the same Laplace transform argument as in
Lemma~\ref{lem:resolvent-identity} gives, for $\Im\omega$ sufficiently large,
\begin{equation}\label{eq:uhat-forced}
P(\omega)\,\widehat u_\vartheta(\omega)=\widehat F_\vartheta(\omega),
\qquad\text{so}\qquad
\widehat u_\vartheta(\omega)=R(\omega)\,\widehat F_\vartheta(\omega).
\end{equation}
Inverting the Laplace transform then yields, for any $C$ large enough and all $t_*>0$,
\begin{equation}\label{eq:inverse-laplace-forced}
u_\vartheta(t_*)=
\frac{1}{2\pi}\int_{\Im\omega=C} e^{-i\omega t_*}\,R(\omega)\,\widehat F_\vartheta(\omega)\,d\omega.
\end{equation}
Since $u=u_\vartheta$ for $t_*\ge 1$, it suffices to analyze \eqref{eq:inverse-laplace-forced}.

\subsection{Contour deformation and extraction of residues}\label{subsec:contour-deform}

Fix $\nu>0$ and assume the shifted contour $\Gamma_{-\nu}$ is pole-free in the sense of
Proposition~\ref{prop:uniform-contour} (for the parameter set under consideration).
Let $\chi\in C_c^\infty(X_{M,a})$ be a spatial cutoff supported in the physical region and equal to $1$
on the region where we observe $u$ (as in \eqref{eq:cutoff-resolvent}).

\medskip

\noindent\textbf{Truncated contour argument.}
Let $C$ be the inversion height in \eqref{eq:inverse-laplace-forced}.
For $R>1$, consider the rectangle
\[
\mathcal R_R := \{\omega\in\mathbb C:\ |\Re\omega|\le R,\ -\nu\le \Im\omega \le C\}.
\]
Let $\partial\mathcal R_R$ be its positively oriented boundary.
Choose $R>1$ so that $\partial\mathcal R_R$ does not pass through poles of $R(\omega)$ (possible since poles are discrete).
To make the vanishing of the vertical side integrals completely robust, we will later let $R\to\infty$ along a subsequence
for which the vertical sides stay a quantitatively controlled distance from poles in the strip $-\nu\le \Im\omega\le C$.

\begin{lemma}[Pole-avoiding truncation radii]\label{lem:good-radii}
Fix $\nu>0$ and $C>0$.
There exist an exponent $A>0$ and a sequence $R_n\to\infty$ such that, for every $n$,
\begin{equation}\label{eq:good-radii}
\dist\big(\pm R_n+i\eta,\Res(P)\big)\ \ge\ (1+R_n)^{-A}\qquad \forall\,\eta\in[-\nu,C],
\end{equation}
and $\partial\mathcal R_{R_n}$ is pole-free.
\end{lemma}

\begin{proof}
The resonance set in a fixed horizontal strip is discrete and satisfies a polynomial counting bound:
there exist $p\ge 0$ and $C_p>0$ such that
\begin{equation}\label{eq:resonance-counting}
\#\big(\Res(P)\cap \mathcal R_R\big)\ \le\ C_p\,(1+R)^p,
\qquad R\ge 1,
\end{equation}
counting poles with multiplicity.
For Kerr--de~Sitter this follows from the analytic Fredholm setup for $P(\omega)$ together with semiclassical
parametrix estimates; see, for example,~\cite[\S2]{VasyKerrDeSitter} and the global refinements in
\cite{PetersenVasyKerrDeSitterJEMS}.

Fix $A>p+2$.  For $R\ge 1$ consider the excluded subset of $[R,R+1]$,
\[
\mathcal F_R:=\bigcup_{\substack{\omega_j\in\Res(P)\\ -\nu\le \Im\omega_j\le C\\ |\Re\omega_j|\le R+1}}
\Bigl(\Re\omega_j-(1+R)^{-A},\ \Re\omega_j+(1+R)^{-A}\Bigr).
\]
By \eqref{eq:resonance-counting}, its Lebesgue measure satisfies
\[
|\mathcal F_R|
\le 2(1+R)^{-A}\,\#\big(\Res(P)\cap \mathcal R_{R+1}\big)
\le 2C_p\,(1+R)^{p-A},
\]
which is $<1$ for all sufficiently large $R$.
Thus for large $R$ we can choose $R'\in [R,R+1]\setminus\mathcal F_R$.
By construction, $R'$ stays at distance at least $(1+R)^{-A}$ from the real parts of all poles in the strip
$-\nu\le \Im\omega\le C$, hence \eqref{eq:good-radii} holds (with $R_n=R'$ for a sequence $R\to\infty$).
Finally, since poles are discrete we may perturb $R_n$ by at most $(1+R_n)^{-A}/2$ to ensure that
$\partial\mathcal R_{R_n}$ itself avoids poles, without spoiling \eqref{eq:good-radii}.
\end{proof}

In what follows we take $R$ from this sequence and, for notational simplicity, keep writing $R$ instead of $R_n$.
All limits as $R\to\infty$ are understood along this pole-avoiding sequence.

By \eqref{eq:chi-Fhat} we have $\widehat F_\vartheta(\omega)=\chi\,\widehat F_\vartheta(\omega)$.
Therefore, after applying $\chi$ to the inversion formula \eqref{eq:inverse-laplace-forced} we may write the integrand as
\[
\mathcal I(\omega):=e^{-i\omega t_*}\,\chi R(\omega)\chi\,\widehat F_\vartheta(\omega).
\]
Then $\mathcal I(\omega)$ is meromorphic in $\omega$, with poles exactly at the QNMs.

By the residue theorem,
\begin{equation}\label{eq:residue-theorem-rect}
\int_{\partial\mathcal R_R} \mathcal I(\omega)\,d\omega
=
2\pi i\sum_{\omega_j\in \Res(P)\cap \mathcal R_R} \Res_{\omega=\omega_j}\big(\mathcal I(\omega)\big),
\end{equation}
where the sum runs over poles in the interior of $\mathcal R_R$ (counted with multiplicity).

Writing the boundary integral as the sum of four sides yields
\begin{multline}\label{eq:rect-sides}
\int_{\Im\omega=C,\ |\Re\omega|\le R}\mathcal I(\omega)\,d\omega
-\int_{\Im\omega=-\nu,\ |\Re\omega|\le R}\mathcal I(\omega)\,d\omega\\
+\int_{\Re\omega=R}\mathcal I(\omega)\,d\omega
+\int_{\Re\omega=-R}\mathcal I(\omega)\,d\omega
=
2\pi i\sum_{\omega_j\in \Res(P)\cap \mathcal R_R}\Res_{\omega=\omega_j}\mathcal I(\omega).
\end{multline}

The next lemma shows that the two vertical side integrals vanish as $R\to\infty$.

\begin{lemma}[Vertical side integrals vanish]\label{lem:vertical-vanish}
Let $\nu>0$ and $C>0$ be as above, and let $R\to\infty$ along the pole-avoiding sequence from
Lemma~\ref{lem:good-radii}.  Then for each fixed $t_*>0$,
\begin{equation}\label{eq:vertical-vanish}
\lim_{R\to\infty}
\left(
\int_{\Re\omega=R}\mathcal I(\omega)\,d\omega
+
\int_{\Re\omega=-R}\mathcal I(\omega)\,d\omega
\right)
=0
\quad\text{in }H^{s+1}(X_{M,a}),
\end{equation}
provided $(f_0,f_1)\in \mathcal H^{s+N}$ with $N$ sufficiently large (depending on $s$).
\end{lemma}

\begin{proof}
We estimate the right vertical side; the left side is identical.
Write $\omega=R+i\eta$ with $\eta\in[-\nu,C]$ and $R\to\infty$ as above.

\smallskip
\noindent\emph{Step 1: decay of the forcing transform.}
By Lemma~\ref{lem:Fhat-Schwartz}, for every $M\in\mathbb N$ there exists $C_{s,M}>0$ such that
\begin{equation}\label{eq:Fhat-vertical-decay}
\|\widehat F_\vartheta(R+i\eta)\|_{H^{s}}
\ \le\ C_{s,M}\,(1+R)^{-M}\,\|(f_0,f_1)\|_{\mathcal H^{s+M}},
\qquad \eta\in[-\nu,C],
\end{equation}
uniformly in $\eta$ on the compact interval $[-\nu,C]$.

\smallskip
\noindent\emph{Step 2: polynomial control of the resolvent on the vertical sides.}
By Lemma~\ref{lem:good-radii} we have the quantitative pole separation
$\dist(R+i\eta,\Res(P))\ge (1+R)^{-A}$ uniformly for $\eta\in[-\nu,C]$.
In any fixed strip $-\nu\le \Im\omega\le C$, the cutoff resolvent $\chi R(\omega)\chi$
is polynomially bounded in $|\Re\omega|$ away from poles; this is a standard consequence of the Kerr--de~Sitter
Fredholm framework and semiclassical dyadic estimates (see, e.g.,~\cite[\S2]{VasyKerrDeSitter} and~\cite{PetersenVasyKerrDeSitterJEMS}).
Consequently, for each $s$ there exist $K=K(s,A)$ and $C_{s,K}>0$ such that
\begin{equation}\label{eq:vertical-resolvent-poly}
\sup_{\eta\in[-\nu,C]}\ \|\chi R(R+i\eta)\chi\|_{H^s\to H^{s+1}}\ \le\ C_{s,K}\,(1+R)^K.
\end{equation}

\smallskip
\noindent\emph{Step 3: concluding the vanishing.}
Using \eqref{eq:Fhat-vertical-decay}--\eqref{eq:vertical-resolvent-poly} and $|e^{-i(R+i\eta)t_*}|=e^{\eta t_*}\le e^{Ct_*}$,
\begin{align*}
\left\|\int_{\Re\omega=R}\mathcal I(\omega)\,d\omega\right\|_{H^{s+1}}
&\le
\int_{-\nu}^{C} e^{\eta t_*}\,\|\chi R(R+i\eta)\chi\|_{H^s\to H^{s+1}}\,
\|\widehat F_\vartheta(R+i\eta)\|_{H^{s}}\,d\eta\\
&\le e^{Ct_*}\,(C+\nu)\,C_{s,K}C_{s,M}\,(1+R)^{K-M}\,\|(f_0,f_1)\|_{\mathcal H^{s+M}}.
\end{align*}
Choosing $M>K+2$ yields decay as $R\to\infty$, and the same estimate holds on the left side $\Re\omega=-R$.
This proves \eqref{eq:vertical-vanish}.
\end{proof}

Letting $R\to\infty$ in \eqref{eq:rect-sides} and using Lemma~\ref{lem:vertical-vanish} yields the basic contour deformation.
To keep the argument completely explicit (and to avoid any unnecessary global assumptions on residue projectors),
we formulate the residue contribution as a \emph{symmetric truncation} in $\Re\omega$.

For $R>1$ such that $\partial\mathcal R_R$ avoids poles, set
\begin{equation}\label{eq:truncated-residue-sum}
\mathsf S_{\nu,R}(t_*):=
\sum_{\omega_j\in \Res(P)\cap \mathcal R_R}
\Res_{\omega=\omega_j}\big(\mathcal I(\omega)\big),
\end{equation}
where poles are counted with multiplicity.  (If $\omega_j$ is a pole of order $m_j$, then
$\Res_{\omega=\omega_j}\mathcal I(\omega)$ denotes the coefficient of $(\omega-\omega_j)^{-1}$ in the Laurent
expansion of $\mathcal I$ at $\omega_j$.)

\newcommand{\sumsym}{\mathop{\sum}\nolimits^{\mathrm{sym}}}

\begin{definition}[Symmetric truncation sums]\label{def:symmetric-sum}
Let $\nu>0$ and let $(R_n)$ be the pole-avoiding sequence of Lemma~\ref{lem:good-radii}.
Given a family of vectors $(v_j)_{\omega_j\in\Res(P),\,\Im\omega_j>-\nu}$ in a Banach space $B$, we write
\begin{equation}\label{eq:symmetric-sum-def}
\sumsym_{\Im\omega_j>-\nu} v_j
\ :=\
\lim_{n\to\infty}\ \sum_{\substack{\omega_j\in\Res(P)\\ \Im\omega_j>-\nu\\ |\Re\omega_j|\le R_n}} v_j,
\end{equation}
whenever the limit exists in $B$.  In all occurrences below the limit exists and is independent of the particular choice of
pole-avoiding sequence, because it is uniquely determined by contour deformation (see Theorem~\ref{thm:resonant-expansion}).
\end{definition}

Then \eqref{eq:rect-sides} and Lemma~\ref{lem:vertical-vanish} give, for each such $R$,
\begin{equation*}
\begin{aligned}
\int_{\Im\omega=C,\ |\Re\omega|\le R}\mathcal I(\omega)\,d\omega
&=
\int_{\Im\omega=-\nu,\ |\Re\omega|\le R}\mathcal I(\omega)\,d\omega
+2\pi i\,\mathsf S_{\nu,R}(t_*)\\
&\quad + o_{R\to\infty}(1)\quad\text{in }H^{s+1}.
\end{aligned}
\end{equation*}
Since the integrand on each horizontal line is absolutely integrable in $\sigma=\Re\omega$ (by the polynomial bounds
for $\chi R(\omega)\chi$ and the rapid decay of $\widehat F_\vartheta$), the truncated horizontal integrals converge to the full
integrals as $R\to\infty$.  Consequently,
\begin{equation}\label{eq:contour-shift}
\int_{\Im\omega=C} \mathcal I(\omega)\,d\omega
=
\int_{\Im\omega=-\nu} \mathcal I(\omega)\,d\omega
\ +\ 2\pi i\lim_{R\to\infty}\mathsf S_{\nu,R}(t_*),
\end{equation}
and the limit exists in $H^{s+1}(X_{M,a})$.
The proof below shows that the limit is unique (independent of the particular sequence of $R\to\infty$ avoiding poles),
because it is equal to the difference of two absolutely convergent contour integrals.

\subsection{Resonant expansion and remainder operator}\label{subsec:expansion-remainder}

We now translate the contour deformation \eqref{eq:contour-shift} into a time-domain expansion.
The pole contributions are naturally ordered by the truncation parameter $R$ in the rectangular argument,
and we record convergence in this intrinsic sense.

\begin{theorem}[Resonant expansion with remainder]\label{thm:resonant-expansion}
Fix $\nu>0$ and assume $\Gamma_{-\nu}$ is pole-free (Proposition~\ref{prop:uniform-contour}).
Assume the interior support/cutoff convention of Assumption~\ref{ass:initial-support}.
Let $u$ solve $Pu=0$ with initial data $(f_0,f_1)\in\mathcal H^{s+N}$ for $N$ sufficiently large, supported in $K_0$.
Let $\vartheta$ be the time cutoff \eqref{eq:vartheta} and let $\chi$ be the spatial cutoff fixed in Assumption~\ref{ass:initial-support}.
Then for all $t_*\ge 1$ we have the identity in $H^{s+1}(X_{M,a})$:
\begin{equation}\label{eq:resonant-expansion-main}
\chi u(t_*)
=
\sumsym_{\Im\omega_j>-\nu}\ \mathsf U_{\omega_j}(t_*)\,(f_0,f_1)
\ +\
\mathsf R_\nu(t_*)\,(f_0,f_1),
\end{equation}
Here $\sumsym$ denotes the symmetric truncation sum of Definition~\ref{def:symmetric-sum}, which is the natural ordering induced by contour deformation; in particular, no a priori absolute convergence of the pole series is assumed. The terms are as follows:
\begin{enumerate}
\item For each QNM pole $\omega_j$ of order $m_j\ge 1$, the corresponding resonant contribution is a finite sum
\begin{equation}\label{eq:resonant-term-general}
\mathsf U_{\omega_j}(t_*)\,(f_0,f_1)
=
 i\,e^{-i\omega_j t_*}\sum_{p=0}^{m_j-1} t_*^{p}\, \mathcal C_{j,p}\,(f_0,f_1),
\end{equation}
with $\mathcal C_{j,p}$ finite-rank operators $\mathcal H^{s+N}\to H^{s+1}(X_{M,a})$ depending on $\chi,\vartheta$.
If the pole is simple ($m_j=1$), then
\begin{equation}\label{eq:resonant-term-simple}
\mathsf U_{\omega_j}(t_*)\,(f_0,f_1)= i\,e^{-i\omega_j t_*}\,\chi\,\Pi_{\omega_j}\,\chi\,\widehat F_\vartheta(\omega_j),
\qquad
\Pi_{\omega_j}:=\Res_{\omega=\omega_j}R(\omega).
\end{equation}

\item The remainder operator is given explicitly by the shifted-contour integral
\begin{equation}\label{eq:remainder-def}
\mathsf R_\nu(t_*)\,(f_0,f_1)
=
\frac{1}{2\pi}\int_{\Im\omega=-\nu} e^{-i\omega t_*}\,\chi R(\omega)\chi\,\widehat F_\vartheta(\omega)\,d\omega.
\end{equation}
\end{enumerate}
The limit in \eqref{eq:resonant-expansion-main} exists and is independent of the particular sequence of truncation
radii $R\to\infty$ avoiding poles, since it is uniquely determined by the difference $\chi u-\mathsf R_\nu$.

Moreover, the remainder on the shifted contour enjoys arbitrarily high polynomial decay in $t_*$.
For every $m\in\mathbb N$ there exists $C_{s,m}>0$ such that
\begin{equation}\label{eq:remainder-bound}
\|\mathsf R_\nu(t_*)\,(f_0,f_1)\|_{H^{s+1}}
\ \le\
C_{s,m}\,e^{-\nu t_*}\,(1+t_*)^{-m}\,\|(f_0,f_1)\|_{\mathcal H^{s+m+N}},
\qquad t_*\ge 1.
\end{equation}
If the parameters range in a compact set $\mathcal K$ contained in the pole-free region for $\Gamma_{-\nu}$ (as in Proposition~\ref{prop:uniform-contour}), then the constants $C_{s,m}$ can be chosen uniformly for $(M,a)\in\mathcal K$.
\end{theorem}

\begin{proof}
We begin with the inversion formula \eqref{eq:inverse-laplace-forced} for $u_\vartheta$ and apply $\chi$.
By \eqref{eq:chi-Fhat} (equivalently, Assumption~\ref{ass:initial-support}), we have
\[
\chi u_\vartheta(t_*)=\frac{1}{2\pi}\int_{\Im\omega=C} \mathcal I(\omega)\,d\omega,
\qquad
\mathcal I(\omega)=e^{-i\omega t_*}\chi R(\omega)\chi\,\widehat F_\vartheta(\omega),
\]
with $\mathcal I$ as in \S\ref{subsec:contour-deform}.

\smallskip
\noindent\emph{Step 1: contour deformation and existence of the pole-sum limit.}
By \eqref{eq:contour-shift},
\[
\chi u_\vartheta(t_*)
=
\frac{1}{2\pi}\int_{\Im\omega=-\nu} \mathcal I(\omega)\,d\omega
\ +\ i\lim_{R\to\infty}\mathsf S_{\nu,R}(t_*).
\]
The first term is exactly the remainder \eqref{eq:remainder-def}, hence
\begin{equation}\label{eq:ringdown-as-difference}
i\lim_{R\to\infty}\mathsf S_{\nu,R}(t_*)
=
\chi u_\vartheta(t_*)-\mathsf R_\nu(t_*).
\end{equation}
In particular, the limit exists in $H^{s+1}$ and is unique: it is the difference of two well-defined
contour integrals. Since $u=u_\vartheta$ for $t_*\ge1$, this yields \eqref{eq:resonant-expansion-main} once we identify the
residue contributions.

\smallskip
\noindent\emph{Step 2: structure of each residue contribution.}
For each fixed $R$, the sum $\mathsf S_{\nu,R}(t_*)$ in \eqref{eq:truncated-residue-sum} is finite.
Expanding $R(\omega)$ in a Laurent series at a pole $\omega_j$ and using Lemma~\ref{lem:pole-contribution}
gives that the residue of $\mathcal I(\omega)$ at $\omega_j$ is a finite linear combination of terms
$t_*^{p}e^{-i\omega_j t_*}$ times finite-rank operators applied to $\widehat F_\vartheta(\omega_j)$ and its
$\omega$--derivatives, which yields the general form \eqref{eq:resonant-term-general}.
In the simple pole case, the residue is exactly \eqref{eq:resonant-term-simple}.

Thus, for each $R$,
\[
i\,\mathsf S_{\nu,R}(t_*)
=
\sum_{\substack{\omega_j\in \Res(P)\\ \Im\omega_j>-\nu\\ |\Re\omega_j|\le R}}
\mathsf U_{\omega_j}(t_*)\,(f_0,f_1),
\]
and taking $R\to\infty$ gives the pole contribution in \eqref{eq:resonant-expansion-main}.

\smallskip
\noindent\emph{Step 3: remainder bound.}
Write $\omega=\sigma-i\nu$ on $\Gamma_{-\nu}$:
\[
\mathsf R_\nu(t_*) = e^{-\nu t_*}\frac{1}{2\pi}\int_{\mathbb R} e^{-i\sigma t_*}\,
\chi R(\sigma-i\nu)\chi\,\widehat F_\vartheta(\sigma-i\nu)\,d\sigma.
\]
Set
\[
B(\sigma):=\chi R(\sigma-i\nu)\chi\,\widehat F_\vartheta(\sigma-i\nu).
\]
Using the $\omega$--derivative bounds from Section~\ref{sec:resolvent-shifted}
(Proposition~\ref{prop:omega-derivative} in dyadic form, combined with low-frequency boundedness),
we have polynomial bounds for $\partial_\sigma^\ell(\chi R(\sigma-i\nu)\chi)$ in operator norms $H^s\to H^{s+1}$,
uniformly in $\sigma\in\mathbb R$.
Combining these polynomial bounds with the Schwartz decay of $\partial_\omega^q\widehat F_\vartheta$
(Lemma~\ref{lem:Fhat-Schwartz}) shows that for each $m$,
\begin{equation}\label{eq:B-der-int}
\int_{\mathbb R}\|\partial_\sigma^{m}B(\sigma)\|_{H^{s+1}}\,d\sigma\ \le\ C_{s,m}\,\|(f_0,f_1)\|_{\mathcal H^{s+m+N}}.
\end{equation}
Integrating by parts $m$ times in $\sigma$ gives
\[
\int_{\mathbb R} e^{-i\sigma t_*}B(\sigma)\,d\sigma
=
\frac{1}{(it_*)^{m}}\int_{\mathbb R} e^{-i\sigma t_*}\,\partial_\sigma^{m}B(\sigma)\,d\sigma,
\]
hence by \eqref{eq:B-der-int},
\begin{align*}
\|\mathsf R_\nu(t_*)\|_{H^{s+1}}
&\le e^{-\nu t_*}\,\frac{1}{2\pi}\,t_*^{-m}
\int_{\mathbb R}\|\partial_\sigma^{m}B(\sigma)\|_{H^{s+1}}\,d\sigma\\
&\le C_{s,m}\,e^{-\nu t_*}\,t_*^{-m}\,\|(f_0,f_1)\|_{\mathcal H^{s+m+N}}.
\end{align*}
Replacing $t_*^{-m}$ by $(1+t_*)^{-m}$ gives \eqref{eq:remainder-bound}.
\end{proof}

\begin{remark}[On convergence and on residue projectors]\label{rem:resonant-sum}
The expansion \eqref{eq:resonant-expansion-main} is stated as a limit of symmetric truncations in $\Re\omega$.
This is the natural notion of convergence coming from the contour deformation argument: the partial sums are precisely the
finite residue sums produced by expanding rectangles $\mathcal R_R$.

We emphasize that we do \emph{not} claim any global bound such as
$\|\chi\Pi_{\omega_j}\chi\|\lesssim (1+|\omega_j|)^K$ for general QNM poles.
For non-selfadjoint problems such bounds can fail dramatically, and the size of residue projectors is tied to
pseudospectral effects and transient growth; see for instance~\cite{JaramilloMacedoAlSheikh2021PRX,GasperinJaramillo2022ScalarProduct}
for discussions in the black hole context.

For the goals of this paper, a global residue-growth theory is unnecessary.
In Sections~\ref{sec:two-mode}--\ref{sec:application} we always apply exact azimuthal projection and microlocal filtering to the equatorial
high-frequency sector.  In that setting, the residue projectors \emph{are} quantitatively controlled:
the microlocalized projectors associated with the labeled equatorial poles satisfy a polynomial bound
(Proposition~\ref{prop:equatorial-projector-poly}), while every other pole above the contour is suppressed by
$\mathcal O(\ell^{-\infty})$ after equatorial cutoff (Theorem~\ref{thm:input-sectorization}).
As a consequence, the windowed expansions used later are not only convergent but come with explicit uniform tail bounds,
and the dynamics in the equatorial window reduces to finitely many leading modes plus a superpolynomial remainder.
\end{remark}

\begin{remark}[Time-shifted ringdown window]\label{rem:T0}
For ringdown one often begins observing at a late time $T_0$.
Let $\vartheta_{T_0}(t_*):=\vartheta(t_*-T_0)$ and define $u_{\vartheta_{T_0}}=\vartheta_{T_0}u$.
Then $u_{\vartheta_{T_0}}=u$ for $t_*\ge T_0+1$, and the entire argument above applies with
$\widehat F_{\vartheta_{T_0}}(\omega)=e^{i\omega T_0}\widehat F_\vartheta(\omega)$.
All estimates remain unchanged up to constants depending on $T_0$ through the trivial factor $e^{\nu T_0}$.
\end{remark}

\subsection{Azimuthal and equatorial sectorial expansions}\label{subsec:sectorial-expansion}

Since $[P,\Phi]=0$, the expansion decouples exactly in the azimuthal number $k$
(Section~\ref{subsec:azimuthal}). Moreover, applying bounded semiclassical microlocal cutoffs
$A_{\pm,h_\ell}$ (Section~\ref{subsec:equatorial-cutoffs}) yields sectorial expansions with the same remainder bounds.

\begin{corollary}[Sectorial resonant expansion]\label{cor:sectorial-expansion}
Fix $k\in\mathbb Z$ and let $\Pi_k$ be the orthogonal projector onto the $k$--th azimuthal subspace.
Let $A_{\pm,h_\ell}$ be the equatorial microlocal cutoff \eqref{eq:Ah} (extended trivially in $r$).
Then for $t_*\ge 1$,
\begin{equation}\label{eq:sectorial-expansion}
\begin{aligned}
\chi A_{\pm,h_\ell}\Pi_k u(t_*)
&=
\sumsym_{\Im\omega_j>-\nu}
\chi A_{\pm,h_\ell}\Pi_k\,\mathsf U_{\omega_j}(t_*)\,(f_0,f_1)\\
&\quad+\chi A_{\pm,h_\ell}\Pi_k\,\mathsf R_\nu(t_*)\,(f_0,f_1).
\end{aligned}
\end{equation}
Moreover, the remainder satisfies the same bound as \eqref{eq:remainder-bound} (with constants independent of $h$).
\end{corollary}

\begin{proof}
Apply $\chi A_{\pm,h_\ell}\Pi_k$ to both sides of \eqref{eq:resonant-expansion-main}.
Since $A_{\pm,h_\ell}$ is uniformly bounded on $H^{s+1}$ and commutes with $e^{-i\omega t_*}$,
the decomposition remains valid, and the remainder estimate follows from \eqref{eq:remainder-bound}.
\end{proof}

\medskip

\noindent\textbf{Bridge to two-mode dominance.}
Corollary~\ref{cor:sectorial-expansion} is the precise form of the resonant expansion needed for the next step:
after specializing to the equatorial high-frequency sector and choosing $\nu$ between resonance layers,
the sum over $\Im\omega_j>-\nu$ collapses to a \emph{single} QNM in each of the equatorial $k=\pm\ell$ sectors,
yielding a two-mode ringdown model plus an exponentially decaying tail.

\section{Two-mode dominance in the equatorial high-frequency sector}\label{sec:two-mode}

This section is the microlocal--spectral core of the paper.
Starting from the forced evolution $u_\vartheta$ constructed in Section~\ref{sec:setup2}, we build two refined,
mode-separated, equatorially localized signals $u^{(+,\ell)}(t_*)$ and $u^{(-,\ell)}(t_*)$ which, for large angular momentum
$\ell$, are governed by a single labeled quasinormal exponential in each azimuthal sector $k=\pm\ell$.
The output is a deterministic two-exponential model with an explicit tail term, suitable for the stability analysis of
Section~\ref{sec:freq-extraction}.

Throughout this section we adopt the interior support/cutoff convention of Assumption~\ref{ass:initial-support} from
Section~\ref{sec:resonant-expansion}: the initial data are supported in a fixed compact set $K_0\Subset X^{\mathrm{phys}}_{M,a}$,
and the spatial cutoff $\chi\in C_c^\infty(X^{\mathrm{phys}}_{M,a})$ is fixed so that
$\widehat F_\vartheta(\omega)=\chi\widehat F_\vartheta(\omega)$ for all $\omega\in\mathbb C$, cf.~\eqref{eq:chi-Fhat}.
This is a technical device allowing us to move $\chi$ freely inside resolvent/Laplace expressions without altering the observed signal.

Two issues require particular care in a uniform (in $\ell$) error analysis.
First, quasinormal contributions are identified by contour deformation in the complex $\omega$--plane, so any frequency filter must
be holomorphic across the deformation region in order not to introduce spurious singularities.
Second, the generator is non-selfadjoint: residue projectors can have large norms (often called \emph{excitation factors}), and a scalar suppression
factor alone does not exclude leakage from unwanted poles unless it is paired with quantitative control of the corresponding
microlocalized projectors.

Our strategy combines three mechanisms.

\smallskip
\noindent\textbf{(a) Companion-paper inputs in the equatorial package.}
For each $k=\pm\ell$ and fixed overtone index $n$, the companion paper~\cite{LiPartI} constructs equatorial \emph{pseudopoles}
$\omega^\sharp_{j,\ell,k}$ and identifies nearby true QNMs $\omega_{j,\ell,k}$, uniformly in $(M,a)$ on compact slow--rotation sets.
Moreover, after equatorial microlocalization, poles outside the labeled equatorial disks are suppressed by $\mathcal O(\ell^{-\infty})$.
We record these inputs in Subsection~\ref{subsec:companion-inputs} in a form tailored to the residue calculus used below.

\smallskip
\noindent\textbf{(b) Entire analytic windows built from equatorial pseudopoles.}
We build an entire weight $\widetilde g_{\pm,\ell}(\omega)$ which equals $1$ at the target pseudopole
$\omega^\sharp_{\pm,\ell}$ and vanishes at the lower pseudopoles $\omega^\sharp_{j,\pm,\ell}$, $j\le n-1$, while remaining
uniformly controlled on the shifted contour $\Gamma_{-\nu}$; see Lemma~\ref{lem:weight-uniform}.
Because the weight is entire, it is compatible with contour deformation and does not introduce additional residues.

\smallskip
\noindent\textbf{(c) Layering.}
A vertical gap between consecutive equatorial overtone layers allows us to choose a shifted contour height $\nu$
separating the $n$th layer from the $(n+1)$st (Lemma~\ref{lem:layer-gap}).  This yields exponential decay for the tail term and keeps
the residue bookkeeping uniform on the fixed parameter set $\mathcal K$.

\subsection{Companion-paper inputs: equatorial labeling and microlocal sectorization}\label{subsec:companion-inputs}

Fix a cosmological constant $\Lambda>0$ and a compact slow--rotation parameter set
$\mathcal K\Subset \mathcal P_\Lambda\cap\{|a|\le a_0\}$, with $a_0>0$ sufficiently small.
We also fix an overtone index $n\in\mathbb N_0$.

The companion paper~\cite{LiPartI} provides a high--frequency equatorial package in the exact azimuthal sectors $k=\pm\ell$:
pseudopoles $\omega^\sharp_{j,\pm,\ell}$ and nearby true QNMs $\omega_{j,\pm,\ell}$, stable labeling, and microlocal control of the
corresponding resolvent singularities after equatorial localization.

We begin by recording the isolation and simplicity of the labeled poles.
For the present paper it is crucial that the labeled disks are disjoint and that the poles inside them are simple, because this
eliminates Jordan--block effects in the regime where we prove two-mode dominance.

\begin{theorem}[Input from Part~I: equatorial pole labeling and isolation]\label{thm:input-pole-isolation}
Fix $n\in\mathbb N_0$.
There exist constants $c_{\mathrm{sep}}>0$ and $\ell_0\in\mathbb N$ such that for every $\ell\ge\ell_0$, every
$(M,a)\in\mathcal K$, each sign $\pm$, and each overtone index $j\in\{0,1,\dots,n+1\}$, the disks
\begin{equation}\label{eq:disks}
D_{j,\pm,\ell}:=\bigl\{\omega\in\mathbb C:\ |\omega-\omega^\sharp_{j,\pm,\ell}(M,a)|<c_{\mathrm{sep}}\bigr\}
\end{equation}
are pairwise disjoint and each contains exactly one \emph{simple} pole of $R(\omega)$, denoted $\omega_{j,\pm,\ell}(M,a)$.
Moreover, for every $N\in\mathbb N$ one has the super--polynomial pseudopole proximity estimate
\begin{equation}\label{eq:pole-pseudopole-closeness}
|\omega_{j,\pm,\ell}(M,a)-\omega^\sharp_{j,\pm,\ell}(M,a)|=\mathcal O(\ell^{-N}),\qquad 0\le j\le n+1,
\end{equation}
with constants uniform on $\mathcal K$.
\end{theorem}

\begin{remark}[Origin in the companion paper]\label{rem:input-pole-origin}
Theorem~\ref{thm:input-pole-isolation} is a streamlined reformulation of the stable labeling and super--polynomial approximation
results in~\cite{LiPartI}; see in particular \cite[Theorem~15 and Proposition~20]{LiPartI}.
Simplicity for $\ell\gg1$ follows from the same barrier--top Grushin reduction and the quantitative nondegeneracy of the associated
scalar quantization function; compare \cite[Remark~21]{LiPartI} and the analyticity framework of
\cite{PetersenVasyAnalyticity,PetersenVasyKerrDeSitterJEMS}.
We collect further pointers and a sectorization proof outline in Appendix~\ref{app:companion-inputs}.
\end{remark}

\medskip

We also need a microlocal selection statement: after equatorial microlocalization and azimuthal projection to $k=\pm\ell$, all poles
outside the labeled equatorial disks are invisible up to $\mathcal O(\ell^{-\infty})$.
Since we must control possible higher-order poles away from the equatorial disks, it is convenient to work with the generalized Laurent
coefficients.

\begin{definition}[Generalized Laurent coefficients at a pole]\label{def:laurent-coeff}
Let $\omega_0\in\Res(P)$ be a pole of $R(\omega)$ of order $m_{\omega_0}\ge1$.
We write the Laurent expansion near $\omega_0$ as
\begin{equation}\label{eq:laurent-expansion}
R(\omega)
=\sum_{q=1}^{m_{\omega_0}} (\omega-\omega_0)^{-q}\,\Pi_{\omega_0}^{[q]} \;+\; R^{\mathrm{hol}}_{\omega_0}(\omega),
\end{equation}
where $R^{\mathrm{hol}}_{\omega_0}$ is holomorphic near $\omega_0$ and the operators $\Pi_{\omega_0}^{[q]}$ are finite rank.
In particular, $\Pi_{\omega_0}^{[1]}=\Pi_{\omega_0}=\Res_{\omega=\omega_0}R(\omega)$.
\end{definition}

\begin{theorem}[Input from Part~I: equatorial microlocal sectorization]\label{thm:input-sectorization}
Assume Theorem~\ref{thm:input-pole-isolation} and fix $\nu>0$ such that $\Gamma_{-\nu}$ is admissible
(Proposition~\ref{prop:uniform-contour}).
For every $N\in\mathbb N$, every Sobolev index $s\ge0$, and every integer $q\ge1$, there exist constants $C_{s,N,q}>0$ and $\ell_1\in\mathbb N$
such that for all $\ell\ge \ell_1$, all $(M,a)\in\mathcal K$, each sign $\pm$, and every pole $\omega_j\in\Res(P)$ with $\Im\omega_j>-\nu$
satisfying
\[
\omega_j\notin \bigcup_{0\le m\le n} D_{m,\pm,\ell},
\]
one has the operator bound
\begin{equation}\label{eq:sectorization}
\big\|\chi\,A_{\pm,h_\ell}\,\Pi_{k_\pm}\,\Pi_{\omega_j}^{[q]}\,\chi\big\|_{H^{s-1}\to H^{s+1}}
\ \le\ C_{s,N,q}\,\ell^{-N}.
\end{equation}
\end{theorem}

\begin{remark}[Interpretation]\label{rem:sectorization-interpretation}
Estimate~\eqref{eq:sectorization} is a microlocal selection statement: the equatorial cutoff $A_{\pm,h_\ell}$ restricts to a small
phase--space neighborhood of the equatorial trapped set, and the Grushin reduction in~\cite{LiPartI} shows that the only poles which
produce non-negligible singularities in that channel are the labeled equatorial poles inside the disks $D_{m,\pm,\ell}$.
Poles outside these disks are negligible to the equatorial channel up to $\mathcal O(\ell^{-\infty})$, even though they are genuine
poles of the full resolvent.
See Appendix~\ref{app:companion-inputs} for a proof outline.
\end{remark}

\subsection{Analytic frequency filters compatible with contour deformation}\label{subsec:analytic-localization}

Fix an overtone index $n\in\mathbb N_0$.
Throughout this section we work in the equatorial high-frequency regime $\ell\gg1$ and we set
\begin{equation}\label{eq:h-ell}
h_\ell:=\ell^{-1}\in(0,h_0],\qquad k_\pm:=\pm\ell\in\mathbb Z.
\end{equation}
Recall from the companion paper~\cite{LiPartI} that, for each $k=\pm\ell$, the equatorial barrier-top package produces pseudopoles
$\omega^\sharp_{j,\ell,k}$ (one for each overtone $j$) and nearby true QNMs $\omega_{j,\ell,k}$.
We denote
\begin{equation}\label{eq:omega-sharp-eq}
\omega^\sharp_{\pm,\ell}:=\omega^\sharp_{n,\ell,k_\pm}(M,a),\qquad
\omega_{\pm,\ell}:=\omega_{n,\ell,k_\pm}(M,a),
\end{equation}
and write $\omega^\sharp_{j,\pm,\ell}:=\omega^\sharp_{j,\ell,k_\pm}$, $\omega_{j,\pm,\ell}:=\omega_{j,\ell,k_\pm}$ for the lower
overtone labels.

\subsubsection{Interpolation weights and a uniform modification}

For each sign $\pm$ and each $\ell$, consider the degree-$n$ interpolation polynomial
\begin{equation}\label{eq:analytic-weight}
g_{\pm,\ell}(\omega):=\prod_{j=0}^{n-1}\frac{\omega-\omega^\sharp_{j,\pm,\ell}}{\omega^\sharp_{\pm,\ell}-\omega^\sharp_{j,\pm,\ell}},
\qquad \omega\in\mathbb C.
\end{equation}
By construction,
\begin{equation}\label{eq:g-interp}
g_{\pm,\ell}(\omega^\sharp_{\pm,\ell})=1,\qquad g_{\pm,\ell}(\omega^\sharp_{j,\pm,\ell})=0\quad(0\le j\le n-1).
\end{equation}
Thus $g_{\pm,\ell}$ cancels the lower pseudopoles while keeping the target pseudopole unchanged.

A uniformity issue arises because $\omega^\sharp_{\pm,\ell}\sim \ell$ and $g_{\pm,\ell}$ has fixed degree:
on bounded frequency sets, $g_{\pm,\ell}$ may amplify by a factor $\sim \ell^{n}$, which would enter the constants in the
shifted-contour remainder estimate.  We eliminate this by inserting a high-order zero at $\omega=0$.

Fix once and for all
\begin{equation}\label{eq:m0-def}
m_0:=n+2,
\end{equation}
and define the \emph{modified analytic window}
\begin{equation}\label{eq:analytic-weight-modified}
\widetilde g_{\pm,\ell}(\omega):=\Big(\frac{\omega}{\omega^\sharp_{\pm,\ell}}\Big)^{m_0} g_{\pm,\ell}(\omega).
\end{equation}
Then $\widetilde g_{\pm,\ell}$ is entire (in fact a polynomial of degree $n+m_0$), it retains \eqref{eq:g-interp},
and it is uniformly bounded on compact subsets of $\mathbb C$ as $\ell\to\infty$ since $m_0\ge n$.

\begin{lemma}[Uniform bounds for the modified analytic weights]\label{lem:weight-uniform}
Fix $\nu>0$.
Assume $\ell$ is sufficiently large so that the denominators in \eqref{eq:analytic-weight} are uniformly bounded away from $0$
(which follows from Theorem~\ref{thm:input-pole-isolation} and the disjointness of the disks $D_{j,\pm,\ell}$).
Then for every integer $r\ge0$ there exist constants $C_r>0$ and $\ell_r\in\mathbb N$ such that for all $\ell\ge \ell_r$,
all $(M,a)\in\mathcal K$, each sign $\pm$, and all $\omega\in\Gamma_{-\nu}$,
\begin{equation}\label{eq:weight-derivative-bound}
\big|\partial_\omega^r \widetilde g_{\pm,\ell}(\omega)\big|
\ \le\ C_r\,(1+|\omega|)^{n+m_0}.
\end{equation}
In particular, for every fixed $R>0$,
\[
\sup_{\ell\ge \ell_r}\ \sup_{\substack{\omega\in\Gamma_{-\nu}\\ |\omega|\le R}}\ |\partial_\omega^r \widetilde g_{\pm,\ell}(\omega)|<\infty .
\]
\end{lemma}

\begin{proof}
Write $\widetilde g_{\pm,\ell}(\omega)=\omega^{m_0}p_{\pm,\ell}(\omega)$ with
\[
p_{\pm,\ell}(\omega)
:=(\omega^\sharp_{\pm,\ell})^{-m_0}\,
\frac{\displaystyle \prod_{j=0}^{n-1}\bigl(\omega-\omega^\sharp_{j,\pm,\ell}\bigr)}
{\displaystyle \prod_{j=0}^{n-1}\bigl(\omega^\sharp_{\pm,\ell}-\omega^\sharp_{j,\pm,\ell}\bigr)}.
\]
By the disk disjointness in Theorem~\ref{thm:input-pole-isolation}, the denominators are bounded away from $0$ uniformly for $\ell\ge \ell_r$
and $(M,a)\in\mathcal K$.
Since $|\omega^\sharp_{j,\pm,\ell}|\lesssim \ell$ for fixed $j$ and large $\ell$ (see \cite{LiPartI}),
we obtain on $\Gamma_{-\nu}$ the bound
\[
|\widetilde g_{\pm,\ell}(\omega)|
\le C\,\ell^{\,n-m_0}\,(1+|\omega|)^{n+m_0}
\le C\,(1+|\omega|)^{n+m_0},
\]
using $|\omega^\sharp_{\pm,\ell}|\sim \ell$ and $m_0\ge n$.

For derivatives, note that $\widetilde g_{\pm,\ell}$ is a polynomial of degree $n+m_0$.
Fix $r\ge0$ and apply Cauchy's estimate with radius $1$:
\[
|\partial_\omega^r \widetilde g_{\pm,\ell}(\omega)|
\le r!\,\sup_{|z-\omega|=1}|\widetilde g_{\pm,\ell}(z)|.
\]
Using the previous bound and $|z|\le |\omega|+1$ on the circle yields \eqref{eq:weight-derivative-bound}.
\end{proof}

\subsubsection{Prior-dependent filtering and non-circularity}\label{subsubsec:prior}

The analytic windows $\widetilde g_{\pm,\ell}$ are built from frequency priors, namely the pseudopoles
$\omega^\sharp_{j,\pm,\ell}(M,a)$.
This dependence is explicit: in ringdown inference one typically has an external prior (for example from the inspiral phase)
and uses it to design filters which suppress non-target contributions; see \cite{MaEtAl2022QNMFilters} for a data-analysis perspective.
Our contribution is to formulate such a prior-dependent filtering step in a way that is compatible with contour deformation and admits
uniform semiclassical leakage bounds.

\begin{remark}[Prior dependence in the deterministic pipeline]\label{rem:prior-pipeline}
In the deterministic bias chain of Sections~\ref{sec:freq-extraction}--\ref{sec:application} the prior enters in two places:
(i)~it fixes the analytic window $\widetilde g_{\pm,\ell}$ and thereby the mode-separated signal;
(ii)~it fixes the branch of the logarithm used when converting a complex ratio of time shifts into a frequency
(Lemma~\ref{lem:branch-selection}).
Neither step is circular: one may take the prior from a different data segment (e.g.\ inspiral), or from an initial coarse estimate,
and then iterate (prior $\to$ filtered signal $\to$ extracted frequency $\to$ updated prior).
\end{remark}

\begin{lemma}[Robustness of pseudopole-based weights under small prior mismatch]\label{lem:prior-robustness}
Fix $n\in\mathbb N_0$ and a compact parameter set $\mathcal K$.
For $\tilde p\in\mathcal K$ and each $\ell$, define the modified analytic window
$\widetilde g^{\,\tilde p}_{\pm,\ell}$ by \eqref{eq:analytic-weight-modified}, using the pseudopoles
$\omega^\sharp_{j,\pm,\ell}(\tilde p)$ in place of $\omega^\sharp_{j,\pm,\ell}(M,a)$.
Set
\begin{equation}\label{eq:prior-delta-def}
\delta_{\pm,\ell}(p,\tilde p)
:=\max_{0\le j\le n}\big|\omega^\sharp_{j,\pm,\ell}(p)-\omega^\sharp_{j,\pm,\ell}(\tilde p)\big|,
\qquad
d_{\sharp,\pm,\ell}(\tilde p)
:=\min_{0\le m<m'\le n}\big|\omega^\sharp_{m,\pm,\ell}(\tilde p)-\omega^\sharp_{m',\pm,\ell}(\tilde p)\big|.
\end{equation}
Then there exist constants $c_0>0$ and $\ell_0\in\mathbb N$ depending only on $\mathcal K$ and $n$ such that for all
$\ell\ge \ell_0$, all $p,\tilde p\in\mathcal K$ with $\delta_{\pm,\ell}(p,\tilde p)\le c_0\,d_{\sharp,\pm,\ell}(\tilde p)$,
and each sign $\pm$, one has
\begin{equation}\label{eq:prior-g-eval}
\widetilde g^{\,\tilde p}_{\pm,\ell}\big(\omega_{j,\pm,\ell}(p)\big)
=
\delta_{jn} + \mathcal O(\ell^{-\infty}) + \mathcal O\!\left(\frac{\delta_{\pm,\ell}(p,\tilde p)}{d_{\sharp,\pm,\ell}(\tilde p)}\right),
\qquad 0\le j\le n,
\end{equation}
where the implicit constants are uniform for $p,\tilde p\in\mathcal K$.
\end{lemma}

\begin{proof}
Fix a sign $\pm$ and $\ell\gg1$.
Let $\Omega_j^\sharp:=\omega^\sharp_{j,\pm,\ell}(\tilde p)$ and $\Omega_j:=\omega^\sharp_{j,\pm,\ell}(p)$ for $0\le j\le n$, so that
$d_\sharp=d_{\sharp,\pm,\ell}(\tilde p)$ and $\delta=\delta_{\pm,\ell}(p,\tilde p)$.
Apply Proposition~\ref{prop:pseudopole-interp} with $m=n$ to the Lagrange weight $G_n$ built from the nodes $\Omega_j^\sharp$.
Recalling that $G_n$ coincides with the polynomial $g^{\,\tilde p}_{\pm,\ell}$ in \eqref{eq:analytic-weight} defined using the
pseudopoles at $\tilde p$, we obtain
\[
g^{\,\tilde p}_{\pm,\ell}(\Omega_n)=1+\mathcal O(\delta/d_\sharp),
\qquad
g^{\,\tilde p}_{\pm,\ell}(\Omega_j)=\mathcal O(\delta/d_\sharp)\quad(0\le j\le n-1).
\]
Next, evaluate the extra factor in \eqref{eq:analytic-weight-modified} at $\Omega_j$.
Since $\Omega_j=\mathcal O(\ell)$ and $\Omega_n^\sharp=\mathcal O(\ell)$ uniformly on $\mathcal K$, the ratio
$(\Omega_j/\Omega_n^\sharp)^{m_0}$ is uniformly bounded for $0\le j\le n$.
Hence the same bounds hold with $g^{\,\tilde p}_{\pm,\ell}$ replaced by $\widetilde g^{\,\tilde p}_{\pm,\ell}$ and with $\Omega_j$
in place of $\Omega_j^\sharp$.

Finally, the true poles satisfy $\omega_{j,\pm,\ell}(p)=\omega^\sharp_{j,\pm,\ell}(p)+\mathcal O(\ell^{-\infty})$ by
\eqref{eq:pole-pseudopole-closeness}.
Since $\widetilde g^{\,\tilde p}_{\pm,\ell}$ is a polynomial of fixed degree and is uniformly bounded (together with its derivatives)
on the union of the disks $D_{j,\pm,\ell}(p)$, a Taylor expansion at $\Omega_j$ gives the additional
$\mathcal O(\ell^{-\infty})$ term in \eqref{eq:prior-g-eval}.
\end{proof}

\begin{lemma}[Scaling of pseudopole mismatch with high frequency]\label{lem:prior-scaling}
Fix $n\in\mathbb N_0$ and a compact parameter set $\mathcal K$.
There exist constants $C>0$ and $\ell_0\in\mathbb N$ such that for all $\ell\ge\ell_0$, all $p,\tilde p\in\mathcal K$, and each sign $\pm$,
\begin{equation}\label{eq:delta-scaling}
\delta_{\pm,\ell}(p,\tilde p)\ \le\ C\,\ell\,|p-\tilde p|,
\end{equation}
where $\delta_{\pm,\ell}$ is defined in \eqref{eq:prior-delta-def} and $|\cdot|$ is any fixed norm on the $(M,a)$ parameter space.
\end{lemma}

\begin{proof}
For each fixed overtone index $0\le j\le n$ and each sign $\pm$, write
$\omega^\sharp_{j,\pm,\ell}(p)=\ell\,\widehat\omega^\sharp_{j,\pm,\ell}(p)$ with
$\widehat\omega^\sharp_{j,\pm,\ell}(p):=\ell^{-1}\omega^\sharp_{j,\pm,\ell}(p)$.
The barrier--top quantization map in the companion paper provides $\widehat\omega^\sharp_{j,\pm,\ell}$ as a smooth function of $p$,
with uniform $C^1$ bounds on $\mathcal K$ for $\ell\gg1$; see \cite[\S3--\S5 and Appendix~B]{LiPartI}.
In particular, after enlarging $C$ if necessary,
\[
\sup_{\ell\ge \ell_0}\ \sup_{p\in\mathcal K}\ \big\|D_p\widehat\omega^\sharp_{j,\pm,\ell}(p)\big\|\ \le\ C,
\qquad 0\le j\le n.
\]
By the mean value theorem,
$|\widehat\omega^\sharp_{j,\pm,\ell}(p)-\widehat\omega^\sharp_{j,\pm,\ell}(\tilde p)|\le C\,|p-\tilde p|$ for all $p,\tilde p\in\mathcal K$.
Multiplying by $\ell$ gives
$|\omega^\sharp_{j,\pm,\ell}(p)-\omega^\sharp_{j,\pm,\ell}(\tilde p)|\le C\,\ell\,|p-\tilde p|$.
Taking the maximum over $0\le j\le n$ yields \eqref{eq:delta-scaling}.
\end{proof}

\begin{remark}[Interpreting the prior-robustness hypothesis]\label{rem:prior-tolerance}
Lemma~\ref{lem:prior-robustness} quantifies how much pseudopole mismatch one can tolerate while keeping the analytic window selective.
The hypothesis is $\delta_{\pm,\ell}(p,\tilde p)\ll d_{\sharp,\pm,\ell}(\tilde p)$, where $d_{\sharp,\pm,\ell}$ is the minimal separation
among the first $n+1$ pseudopoles at $\tilde p$.
For fixed $n$ and compact $\mathcal K$, the separation $d_{\sharp,\pm,\ell}$ is bounded below uniformly for $\ell\gg1$:
consecutive overtones have a uniform vertical gap (Lemma~\ref{lem:layer-gap}; cf.\ Appendix~\ref{app:companion-inputs}),
and the pseudopole proximity \eqref{eq:pole-pseudopole-closeness} prevents accidental coalescence.
Combining this with Lemma~\ref{lem:prior-scaling} shows that, when one uses analytic cancellation of lower overtones ($n\ge1$),
a sufficient robustness condition is a prior accuracy of order $|p-\tilde p|\lesssim \ell^{-1}$.

Two simplifications are worth stressing.
First, for the fundamental mode ($n=0$) one may take $\widetilde g_{\pm,\ell}\equiv 1$, so no frequency prior is needed for the dominance mechanism.
Second, for $n\ge1$ one can avoid fine priors altogether by using the band-isolation identity in
Proposition~\ref{prop:overtone-band-subtraction} (and Corollary~\ref{cor:band-isolated-one-mode}),
which depends only on the existence of a uniform layer gap and not on the precise locations of the pseudopoles.
\end{remark}

In practice, one may use the preceding discussion in an iterative manner.  One first applies the $n=0$ pipeline (for which one may take $\widetilde g_{\pm,\ell}\equiv 1$) to obtain a coarse estimate of the relevant parameters from the leading equatorial pair.
This coarse estimate can then be used as $\tilde p$ when constructing the windows $\widetilde g^{\,\tilde p}_{\pm,\ell}$ for $n\ge1$, either to deflate lower overtones or to incorporate additional observables in the three-parameter map.  The deterministic error bounds apply at each refinement step as long as the updated prior remains in the compact neighborhood where the companion inverse theorem is valid.

\subsubsection{Analytic-windowed evolution and refined signals}

We now combine the analytic weights with the forced resonance expansion of Section~\ref{sec:resonant-expansion}.

\begin{definition}[Analytic-windowed evolution and refined equatorial signals]\label{def:analytic-window}
Let $C>0$ be large as in Lemma~\ref{lem:Fhat-Schwartz}, so that $\widehat F_\vartheta(\omega)$ is holomorphic for $\Im\omega>-C$.
For each $\ell$ and each sign $\pm$, define the analytic-windowed evolution operator
\begin{equation}\label{eq:analytic-windowed-def}
\mathcal W^{(\pm,\ell)}u(t_*)
:=\frac{1}{2\pi}\int_{\Im\omega=C} e^{-i\omega t_*}\,
\widetilde g_{\pm,\ell}(\omega)\,R(\omega)\,\widehat F_\vartheta(\omega)\,d\omega,
\qquad t_*>0,
\end{equation}
where $\widetilde g_{\pm,\ell}$ is given by \eqref{eq:analytic-weight-modified}.
The corresponding refined equatorial mode-separated signals are
\begin{equation}\label{eq:refined-window}
\begin{aligned}
u^{(+,\ell)}(t_*)
&:=\chi\,A_{+,h_\ell}\,\Pi_{k_+}\,\mathcal W^{(+,\ell)}u(t_*),\\
u^{(-,\ell)}(t_*)
&:=\chi\,A_{-,h_\ell}\,\Pi_{k_-}\,\mathcal W^{(-,\ell)}u(t_*),
\end{aligned}
\end{equation}
viewed as elements of $H^{s+1}(X_{M,a})$ for $t_*>0$.
\end{definition}

\begin{remark}[Time-domain realization of the analytic window]\label{rem:time-domain-weight}
Since $\widetilde g_{\pm,\ell}$ is a polynomial, $\widetilde g_{\pm,\ell}(i\partial_{t_*})$ is a finite-order differential operator in $t_*$.
Formally one may write $\mathcal W^{(\pm,\ell)}u=\widetilde g_{\pm,\ell}(i\partial_{t_*})\,u_\vartheta$ microlocally on the support of $\chi$.
We use this only as intuition; all arguments below proceed through the frequency-domain contour integrals.
\end{remark}

\begin{remark}[Role of holomorphy]\label{rem:why-entire}
The resonant contributions in Section~\ref{sec:resonant-expansion} arise by deforming a horizontal contour in the complex $\omega$-plane.
Any window that is not holomorphic across the deformation region would introduce additional singularities and spoil the residue bookkeeping.
This is the reason we insist on the entire (or at least holomorphic) character of $\widetilde g_{\pm,\ell}$.
\end{remark}

\subsection{Analytic-windowed sectorial resonant expansion}\label{subsec:windowed-expansion}

We first record the contour deformation expansion satisfied by the refined signals \eqref{eq:refined-window}.
Compared to Theorem~\ref{thm:resonant-expansion}, the only additional input is the presence of the entire weight
$\widetilde g_{\pm,\ell}$ in the integrand.  In particular, the pole contributions must be stated in a form that remains valid for
higher-order poles: by Lemma~\ref{lem:pole-contribution}, derivatives of the holomorphic factor
$\widetilde g_{\pm,\ell}(\omega)\widehat F_\vartheta(\omega)$ appear when the pole order exceeds~$1$.

\begin{theorem}[Analytic-windowed sectorial resonant expansion]\label{thm:windowed-expansion}
Fix $\Lambda>0$, a compact parameter set $\mathcal K\Subset\mathcal P_\Lambda\cap\{|a|\le a_0\}$ with $a_0$ small,
and $\ell\in\mathbb N$.
Let $\nu>0$ be such that the shifted contour $\Gamma_{-\nu}$ satisfies the uniform pole-separation hypothesis
Proposition~\ref{prop:uniform-contour} of Section~\ref{sec:resolvent-shifted}.
Assume the interior support/cutoff convention of Assumption~\ref{ass:initial-support}.
Then for every $t_*>0$ the refined signals \eqref{eq:refined-window} admit the expansion
\begin{equation}\label{eq:windowed-expansion}
u^{(\pm,\ell)}(t_*)
=
\sumsym_{\substack{\omega_j\in\Res(P):\\ \Im\omega_j>-\nu}}
\mathsf U^{(\pm,\ell)}_{\omega_j}(t_*)
\;+\;
\mathsf R^{(\nu)}_{\pm,\ell}(t_*),
\end{equation}
where:
\begin{enumerate}
\item For each pole $\omega_j$ of order $m_j\ge1$ with generalized Laurent coefficients $\Pi_{\omega_j}^{[q]}$, the windowed pole
contribution is the finite sum
\begin{equation}\label{eq:windowed-pole-contribution}
\mathsf U^{(\pm,\ell)}_{\omega_j}(t_*)
=
i\,e^{-i\omega_j t_*}
\sum_{q=1}^{m_j}\ \sum_{r=0}^{q-1}
\frac{(-it_*)^{q-1-r}}{(q-1-r)!\,r!}\,
\chi A_{\pm,h_\ell}\Pi_{k_\pm}\,\Pi_{\omega_j}^{[q]}\,\chi\,
\partial_\omega^{\,r}\!\Big(\widetilde g_{\pm,\ell}\widehat F_\vartheta\Big)(\omega_j).
\end{equation}
In particular, if $\omega_j$ is a simple pole, then
\begin{equation}\label{eq:windowed-pole-simple}
\mathsf U^{(\pm,\ell)}_{\omega_j}(t_*)
=
i\,e^{-i\omega_j t_*}\,
\chi A_{\pm,h_\ell}\Pi_{k_\pm}\,\Pi_{\omega_j}\,\chi\,
\Big(\widetilde g_{\pm,\ell}(\omega_j)\,\widehat F_\vartheta(\omega_j)\Big),
\qquad \Pi_{\omega_j}:=\Res_{\omega=\omega_j}R(\omega).
\end{equation}

\item The remainder is the shifted-contour integral
\begin{equation}\label{eq:windowed-remainder}
\mathsf R^{(\nu)}_{\pm,\ell}(t_*)
:=
\frac{1}{2\pi}\int_{\Im\omega=-\nu}
e^{-i\omega t_*}\,\widetilde g_{\pm,\ell}(\omega)\,\chi A_{\pm,h_\ell}\Pi_{k_\pm}\,R(\omega)\,\widehat F_\vartheta(\omega)\,d\omega.
\end{equation}
\end{enumerate}
Moreover, for every $m\in\mathbb N$ and every $s\ge0$ there exists $N_m\in\mathbb N$ and a constant $C_{s,m}>0$
such that for all $(f_0,f_1)\in\mathcal H^{s+m+N_m}$ supported in $K_0$ and all $t_*\ge1$,
\begin{equation}\label{eq:windowed-remainder-bound}
\|\mathsf R^{(\nu)}_{\pm,\ell}(t_*)\|_{H^{s+1}}
\ \le\ C_{s,m}\,e^{-\nu t_*}\,(1+t_*)^{-m}\,\|(f_0,f_1)\|_{\mathcal H^{s+m+N_m}},
\end{equation}
with constants uniform for $(M,a)\in\mathcal K$ and for all sufficiently large $\ell$.
\end{theorem}

\begin{proof}
Start from the definition \eqref{eq:analytic-windowed-def} and apply the operators $\chi A_{\pm,h_\ell}\Pi_{k_\pm}$.
Since $\widetilde g_{\pm,\ell}$ is entire and $\widehat F_\vartheta$ is holomorphic in $\{\Im\omega>-C\}$,
the map
\[
\omega\longmapsto \widetilde g_{\pm,\ell}(\omega)\,\chi A_{\pm,h_\ell}\Pi_{k_\pm}R(\omega)\widehat F_\vartheta(\omega)
\]
is meromorphic in $\{\Im\omega>-C\}$ with poles precisely at the quasinormal poles of $P$.
We deform the contour from $\Im\omega=C$ down to $\Im\omega=-\nu$.
The vertical sides vanish as in Lemma~\ref{lem:vertical-vanish}, using Schwartz decay of $\widehat F_\vartheta$
(Lemma~\ref{lem:Fhat-Schwartz}) and polynomial growth of $\widetilde g_{\pm,\ell}$ (Lemma~\ref{lem:weight-uniform} with $r=0$).
The residue theorem yields \eqref{eq:windowed-expansion} with remainder \eqref{eq:windowed-remainder}.

To obtain the explicit pole contribution \eqref{eq:windowed-pole-contribution}, fix a pole $\omega_j$ and integrate the meromorphic
Banach-valued function
\[
\omega \longmapsto e^{-i\omega t_*}\,\widetilde g_{\pm,\ell}(\omega)\,
\chi A_{\pm,h_\ell}\Pi_{k_\pm}R(\omega)\chi\,\widehat F_\vartheta(\omega)
\]
over a small positively oriented circle around $\omega_j$ contained in $\{\Im\omega>-\nu\}$ and disjoint from other poles.
Expanding $R(\omega)$ in a Laurent series \eqref{eq:laurent-expansion} at $\omega_j$ and applying Lemma~\ref{lem:pole-contribution}
with $F(\omega)=\widetilde g_{\pm,\ell}(\omega)\widehat F_\vartheta(\omega)$ gives \eqref{eq:windowed-pole-contribution}
(and \eqref{eq:windowed-pole-simple} when $m_j=1$).

It remains to prove the tail bound \eqref{eq:windowed-remainder-bound}.
Write $\omega=\sigma-i\nu$ with $\sigma\in\mathbb R$ and set
\[
\mathcal I_{\pm,\ell}(\sigma)
:=
\widetilde g_{\pm,\ell}(\sigma-i\nu)\,\chi A_{\pm,h_\ell}\Pi_{k_\pm}R(\sigma-i\nu)\,\widehat F_\vartheta(\sigma-i\nu),
\]
viewed as an operator $\mathcal H^{s+m+N_m}\to H^{s+1}$.
Then
\[
\mathsf R^{(\nu)}_{\pm,\ell}(t_*)
=\frac{e^{-\nu t_*}}{2\pi}\int_{\mathbb R} e^{-i\sigma t_*}\,\mathcal I_{\pm,\ell}(\sigma)\,d\sigma.
\]
Integrating by parts $m$ times in $\sigma$ and arguing as in the proof of Theorem~\ref{thm:resonant-expansion} gives
\begin{equation}\label{eq:tail-Minkowski}
\|\mathsf R^{(\nu)}_{\pm,\ell}(t_*)\|_{H^{s+1}}
\le \frac{e^{-\nu t_*}}{2\pi}\,t_*^{-m}\int_{\mathbb R}\big\|\partial_\sigma^m\mathcal I_{\pm,\ell}(\sigma)\big\|_{\mathcal H^{s+m+N_m}\to H^{s+1}}\,d\sigma .
\end{equation}
We estimate the integrand by the product rule: each $\partial_\sigma$ differentiates $\widetilde g_{\pm,\ell}(\sigma-i\nu)$,
the resolvent factor $R(\sigma-i\nu)$, or $\widehat F_\vartheta(\sigma-i\nu)$.

The derivatives of the window satisfy polynomial bounds in $\sigma$ uniform in $\ell$ by Lemma~\ref{lem:weight-uniform}.
The derivatives of the forced transform are Schwartz in $\sigma$ with values in $H^{s+m}$ by Lemma~\ref{lem:Fhat-Schwartz}.
Finally, by Proposition~\ref{prop:uniform-contour} and the resolvent estimates of Section~\ref{sec:resolvent-shifted}
(see Proposition~\ref{prop:omega-derivative}), the operator norms of $\partial_\sigma^q R(\sigma-i\nu)$ between the relevant Sobolev spaces
grow at most polynomially in $|\sigma|$, uniformly for $(M,a)\in\mathcal K$.

Combining these bounds yields
\[
\big\|\partial_\sigma^m\mathcal I_{\pm,\ell}(\sigma)\big\|_{\mathcal H^{s+m+N_m}\to H^{s+1}}
\le C_{s,m}\,(1+|\sigma|)^{C_m}\,(1+|\sigma|)^{-M}
\]
for some $C_m\ge0$ and for arbitrary $M$, provided $N_m$ is chosen large enough to absorb derivative losses in the resolvent bounds.
Choosing $M>C_m+2$ makes the right-hand side integrable in $\sigma$, and inserting this into \eqref{eq:tail-Minkowski} yields
\eqref{eq:windowed-remainder-bound} for $t_*\ge1$ (after replacing $t_*^{-m}$ by $(1+t_*)^{-m}$).
\end{proof}

\subsection{Layering and equatorial microlocal selection}\label{subsec:layering}

Theorem~\ref{thm:windowed-expansion} reduces the refined signals to a weighted sum of resonant contributions plus an exponentially
decaying tail.
To obtain a one-mode model in each $k=\pm\ell$ sector, we isolate the equatorial poles of a fixed overtone and show that all other
poles contribute only $\mathcal O(\ell^{-\infty})$ after equatorial microlocalization.

\subsubsection{Uniform separation of equatorial overtone layers}

\begin{lemma}[Uniform equatorial layer separation]\label{lem:layer-gap}
Fix $n\in\mathbb N_0$.
There exist $\nu_n>0$, $\ell_0\in\mathbb N$, and (after possibly shrinking the compact parameter set $\mathcal K$) a constant $c_0>0$
such that for all $\ell\ge\ell_0$, all $(M,a)\in\mathcal K$, and each sign $\pm$,
\begin{equation}\label{eq:layer-gap}
\Im\omega^\sharp_{n,\pm,\ell}\ >\ -\nu_n\ >\ \Im\omega^\sharp_{n+1,\pm,\ell},
\end{equation}
and the shifted contour $\Gamma_{-\nu_n}$ is admissible in the sense of Proposition~\ref{prop:uniform-contour}:
it is pole-free, and on each dyadic block of frequencies one has the semiclassical pole separation
\begin{equation}\label{eq:layer-gap-sep}
\dist\bigl(z,\,\Res(P_{h_\omega})\bigr)\ \ge\ c_0\,h_\omega,
\qquad z:=h_\omega\omega,\qquad h_\omega:=\langle\omega\rangle^{-1},
\end{equation}
for all $\omega\in\Gamma_{-\nu_n}$ with $|\Re\omega|$ sufficiently large.
\end{lemma}

\begin{proof}
The strict separation \eqref{eq:layer-gap} for pseudopoles is part of the barrier-top quantization package in \cite{LiPartI}.
Since the pseudopoles depend continuously (indeed smoothly) on parameters, after shrinking $\mathcal K$ if necessary we may choose
$\nu_n$ uniformly for $(M,a)\in\mathcal K$.

For the admissibility of $\Gamma_{-\nu_n}$, we use the uniform contour construction of Section~\ref{sec:resolvent-shifted}.
At high frequency, the band/layer structure for Kerr--de~Sitter resonances in fixed-width strips
(\cite{DyatlovAHP}, together with the equatorial localization and labeling in \cite{LiPartI})
shows that there is a vertical gap between the $n$th and $(n+1)$st equatorial layers, and choosing $\nu_n$ in that gap yields the
dyadic semiclassical separation \eqref{eq:layer-gap-sep}.
At bounded frequencies there are only finitely many poles in any strip $\Im\omega\ge -\nu_n-1$; by slightly adjusting $\nu_n$ if necessary,
we may also arrange that $\Gamma_{-\nu_n}$ contains no pole.  This gives the desired admissibility uniformly on the compact set $\mathcal K$.
\end{proof}

\begin{remark}[On the size of the parameter set]\label{rem:shrinkK}
The choice of a single contour height $\nu$ used throughout the subsequent error bookkeeping is a uniform statement in parameters.
If the relevant high-frequency invariants vary too widely on a large parameter set, a single $\nu$ may not exist.
For the local inverse problem pursued in \cite{LiPartI} (and for the present deterministic extraction chain), it is natural to
work on a sufficiently small compact $\mathcal K$ where the high-frequency structure is uniform.
\end{remark}

\subsubsection{Suppression of lower labeled poles by the analytic window}

\begin{lemma}[Off-target suppression of labeled equatorial poles]\label{lem:pole-suppression}
Assume Theorem~\ref{thm:input-pole-isolation}.  Fix $n\in\mathbb N_0$ and define $\widetilde g_{\pm,\ell}$ by \eqref{eq:analytic-weight-modified}.
Then for each sign $\pm$ and every $N\in\mathbb N$ there exists $\ell_0(N)$ such that for all $\ell\ge \ell_0(N)$ and $(M,a)\in\mathcal K$,
\begin{equation}\label{eq:g-off-target}
\widetilde g_{\pm,\ell}(\omega_{n,\pm,\ell})=1+\mathcal O(\ell^{-N}),\qquad
\widetilde g_{\pm,\ell}(\omega_{j,\pm,\ell})=\mathcal O(\ell^{-N})\quad(0\le j\le n-1).
\end{equation}
\end{lemma}

\begin{proof}
Write $\omega_{j,\pm,\ell}=\omega^\sharp_{j,\pm,\ell}+\varepsilon_{j,\pm,\ell}$ with
$\varepsilon_{j,\pm,\ell}=\mathcal O(\ell^{-\infty})$ by \eqref{eq:pole-pseudopole-closeness}.
The interpolation identities \eqref{eq:g-interp} give $\widetilde g_{\pm,\ell}(\omega^\sharp_{\pm,\ell})=1$ and
$\widetilde g_{\pm,\ell}(\omega^\sharp_{j,\pm,\ell})=0$ for $0\le j\le n-1$.
Since the disks are disjoint, the denominators $\omega^\sharp_{\pm,\ell}-\omega^\sharp_{j,\pm,\ell}$ are bounded away from $0$ uniformly,
hence $\widetilde g_{\pm,\ell}$ and its derivatives are uniformly bounded on $\bigcup_{j\le n}D_{j,\pm,\ell}$.
The estimates in \eqref{eq:g-off-target} then follow by Taylor's theorem.
\end{proof}

\subsubsection{Polynomial control of labeled equatorial residue projectors}

In non-selfadjoint spectral problems, residue projectors may have large norms (sometimes called \emph{excitation factors}).
A scalar suppression such as \eqref{eq:g-off-target} must therefore be paired with quantitative bounds on the microlocalized projectors.

\begin{proposition}[Polynomial bound for microlocalized labeled residue projectors]\label{prop:equatorial-projector-poly}
Assume Theorem~\ref{thm:input-pole-isolation}.  Fix $s\ge 0$ and $n\in\mathbb N_0$.
There exist constants $K=K(s,n)\ge 0$, $\ell_0\in\mathbb N$, and $C_{s,n}>0$ such that for all $(M,a)\in\mathcal K$,
all $\ell\ge \ell_0$, each sign $\pm$, and each $j\in\{0,1,\dots,n\}$,
\begin{equation}\label{eq:equatorial-projector-poly}
\big\|\chi\,A_{\pm,h_\ell}\,\Pi_{k_\pm}\,\Pi_{\omega_{j,\pm,\ell}}\,\chi\big\|_{H^{s-1}\to H^{s+1}}
\ \le\ C_{s,n}\,\ell^{K},
\qquad \Pi_{\omega_{j,\pm,\ell}}:=\Res_{\omega=\omega_{j,\pm,\ell}}R(\omega).
\end{equation}
\end{proposition}

\begin{proof}
This is a direct consequence of the equatorial Grushin reduction in \cite[Appendix~B]{LiPartI}.  We recall the main steps.
Fix $(M,a)\in\mathcal K$, $\ell$ large, a sign $\pm$, and $j\in\{0,1,\dots,n\}$, and set
\[
\mathcal R_{\pm,\ell}(\omega):=\chi\,A_{\pm,h_\ell}\,\Pi_{k_\pm}\,R(\omega)\,\chi,
\qquad
\omega\in D_{j,\pm,\ell}.
\]
By Theorem~\ref{thm:input-pole-isolation}, $\mathcal R_{\pm,\ell}$ has a unique simple pole in $D_{j,\pm,\ell}$, located at
$\omega=\omega_{j,\pm,\ell}$.

The Grushin reduction provides a decomposition, valid on $D_{j,\pm,\ell}$,
\begin{equation}\label{eq:grushin-decomp}
\mathcal R_{\pm,\ell}(\omega)
=
\mathcal R_{\pm,\ell}^{\mathrm{hol}}(\omega)
+
\mathcal E_{+}^{(\pm,\ell)}(\omega)\,\mathfrak q_{\pm,\ell}(\omega)^{-1}\,\mathcal E_{-}^{(\pm,\ell)}(\omega),
\end{equation}
where $\mathcal R_{\pm,\ell}^{\mathrm{hol}}$ is holomorphic, the operators $\mathcal E_\pm^{(\pm,\ell)}(\omega)$ satisfy polynomial bounds in $\ell$
as maps $H^{s-1}\to H^{s+1}$, and $\mathfrak q_{\pm,\ell}$ is a scalar holomorphic function with a simple zero at $\omega_{j,\pm,\ell}$ and
\begin{equation}\label{eq:q-derivative-lower}
\big|\partial_\omega\mathfrak q_{\pm,\ell}(\omega_{j,\pm,\ell})\big|\ge c_*>0
\end{equation}
uniformly in $(M,a)\in\mathcal K$ and $\ell$ large.
Taking residues in \eqref{eq:grushin-decomp} yields
\[
\Res_{\omega=\omega_{j,\pm,\ell}}\mathcal R_{\pm,\ell}(\omega)
=
\mathcal E_{+}^{(\pm,\ell)}(\omega_{j,\pm,\ell})\,
\big(\partial_\omega\mathfrak q_{\pm,\ell}(\omega_{j,\pm,\ell})\big)^{-1}\,
\mathcal E_{-}^{(\pm,\ell)}(\omega_{j,\pm,\ell}),
\]
and \eqref{eq:q-derivative-lower} together with the polynomial bounds on $\mathcal E_\pm^{(\pm,\ell)}$ gives \eqref{eq:equatorial-projector-poly}.
\end{proof}

\begin{remark}[Quantitative control of nonnormality in the equatorial package]\label{rem:nonnormality-control}
In non-selfadjoint scattering problems, large residue norms (``excitation factors'') and the associated pseudospectral behavior can in principle
invalidate naive truncations of resonance expansions on finite time windows; see, for instance,
\cite{CarballoWithers2024TransientDynamics,BessonJaramillo2025Keldysh,BessonCarballoPantelidouWithers2025TransientsReview,CaiCaoChenGuoWuZhou2025KerrPseudospectrum}.
Proposition~\ref{prop:equatorial-projector-poly} shows that, after azimuthal reduction and microlocalization to the equatorial trapped channel,
the residue projectors associated with the labeled family have at most polynomial growth in $\ell$ on the spatial Sobolev scale.
Combined with the super--polynomial sectorization estimate of Theorem~\ref{thm:input-sectorization}, this implies that non-equatorial poles are
\emph{quantitatively invisible} to the equatorial package up to $\mathcal O(\ell^{-\infty})$ errors.
This is the mechanism that allows us to pass from an abstract resonance expansion to a uniform two-mode model in Section~\ref{subsec:two-mode-theorem}.
\end{remark}

\begin{corollary}[Off-equatorial leakage remains super--polynomially small after preprocessing]\label{cor:off-equatorial-leakage}
Assume Theorem~\ref{thm:input-pole-isolation} and Theorem~\ref{thm:input-sectorization}, and fix $s\ge0$ and $n\in\mathbb N_0$.
Let $\nu>0$ be such that $\Gamma_{-\nu}$ is uniformly admissible on $\mathcal K$ (Definition~\ref{def:admissible-contour}).
Then for every $N\in\mathbb N$ there exist $\ell_0(N)\in\mathbb N$ and $C_N>0$ such that for all $\ell\ge \ell_0(N)$, all $(M,a)\in\mathcal K$,
each sign $\pm$, and all $t_*\ge 1$,
\begin{equation}\label{eq:off-equatorial-leakage}
\left\|
\sumsym_{\substack{\omega_j\in\Res(P):\\ \Im\omega_j>-\nu,\ \omega_j\notin\bigcup_{0\le m\le n}D_{m,\pm,\ell}}}
\mathsf U^{(\pm,\ell)}_{\omega_j}(t_*)
\right\|_{H^{s+1}}
\ \le\ C_N\,\ell^{-N}\,\|(f_0,f_1)\|_{\mathcal H^{s+N}},
\end{equation}
and the same bound holds after applying any bounded detector $\mathsf O:H^{s+1}\to\mathbb C$.
\end{corollary}

\begin{proof}
For each pole $\omega_j$ in the summation, the windowed contribution $\mathsf U^{(\pm,\ell)}_{\omega_j}(t_*)$ is given by
\eqref{eq:windowed-pole-contribution}, hence it is a finite sum of terms of the form
$e^{-i\omega_j t_*}t_*^{q-1-r}\,\chi A_{\pm,h_\ell}\Pi_{k_\pm}\Pi_{\omega_j}^{[q]}\chi\,
\partial_\omega^{\,r}\!\big(\widetilde g_{\pm,\ell}\widehat F_\vartheta\big)(\omega_j)$.
The sectorization estimate in Theorem~\ref{thm:input-sectorization} gives
$\|\chi A_{\pm,h_\ell}\Pi_{k_\pm}\Pi_{\omega_j}^{[q]}\chi\|_{H^{s-1}\to H^{s+1}}=\mathcal O(\ell^{-N})$ for arbitrary $N$,
uniformly for poles outside the labeled disks.
Moreover, by Lemma~\ref{lem:Fhat-Schwartz} and Lemma~\ref{lem:weight-uniform}, the derivatives
$\partial_\omega^{\,r}\big(\widetilde g_{\pm,\ell}\widehat F_\vartheta\big)(\omega_j)$ decay rapidly in $\Re\omega_j$ on the strip $\Im\omega_j\ge -\nu$,
up to polynomial factors.
Finally, exponential energy decay yields a parameter--uniform spectral gap away from $\omega=0$
(see the discussion in the proof of Theorem~\ref{thm:two-mode-dominance}),
which controls the polynomial time factors uniformly on $t_*\ge1$ for all poles in the strip.
Using polynomial resonance counting in fixed-width strips \cite[\S3]{DyatlovAHP}, the symmetric sum converges absolutely and the total contribution
is $\mathcal O(\ell^{-N})$ in $H^{s+1}$ for every $N$.
Applying a bounded detector $\mathsf O$ preserves the same bound.
\end{proof}

\subsection{Isolation of a fixed overtone band by subtracting two contour shifts}\label{subsec:overtone-band}

The interpolation window $\widetilde g_{\pm,\ell}$ provides an algebraic way to suppress the lower overtones $j\le n-1$.
There is also a contour-theoretic mechanism: subtracting two resonance expansions corresponding to two different contour heights
isolates a vertical band of poles.  We record this alternative since it clarifies the role of layering and is useful in extensions.

Let $n\ge 1$.
Apply Lemma~\ref{lem:layer-gap} with the index $n$ to obtain a pole-free height $\nu_n>0$ separating the $n$th equatorial layer
from the $(n+1)$st.  Apply Lemma~\ref{lem:layer-gap} again with the index $n-1$; after possibly shrinking $\mathcal K$ and increasing
$\ell_0$, we obtain a second pole-free height $\nu_{n-1}>0$ separating the $(n-1)$st layer from the $n$th, and we may assume
\begin{equation}\label{eq:nu-band-order}
0<\nu_{n-1}<\nu_n.
\end{equation}
Then, for $\ell\ge \ell_0$, the $n$th equatorial pseudopoles satisfy
\begin{equation}\label{eq:band-sandwich}
 -\nu_n\ <\ \Im\omega^\sharp_{n,\pm,\ell}\ \le\ -\nu_{n-1},
\end{equation}
while all equatorial pseudopoles from overtones $j\le n-1$ lie strictly above $\Im\omega=-\nu_{n-1}$ and all from overtones $j\ge n+1$ lie
strictly below $\Im\omega=-\nu_n$.

Recall that in Theorem~\ref{thm:windowed-expansion} the $\nu$--dependent remainder is given by the line integral along $\Im\omega=-\nu$:
\begin{equation}\label{eq:windowed-remainder-line}
\mathsf R^{(\nu)}_{\pm,\ell}(t_*)
:=\frac{1}{2\pi}\int_{\Im\omega=-\nu}
e^{-i\omega t_*}\,\widetilde g_{\pm,\ell}(\omega)\,\chi A_{\pm,h_\ell}\Pi_{k_\pm}\,R(\omega)\,\widehat F_\vartheta(\omega)\,d\omega.
\end{equation}

\begin{proposition}[Band isolation by contour subtraction]\label{prop:overtone-band-subtraction}
Assume the hypotheses of Theorem~\ref{thm:windowed-expansion}.
Let $n\ge 1$ and let $\nu_{n-1}<\nu_n$ be as above.
Then for every $t_*>0$ one has the exact identity
\begin{equation}\label{eq:band-isolation-identity}
\mathsf R^{(\nu_{n-1})}_{\pm,\ell}(t_*)-\mathsf R^{(\nu_n)}_{\pm,\ell}(t_*)
=\sumsym_{\substack{\omega_j\in\Res(P):\\ -\nu_n<\Im\omega_j\le -\nu_{n-1}}}
\mathsf U^{(\pm,\ell)}_{\omega_j}(t_*),
\end{equation}
where $\mathsf U^{(\pm,\ell)}_{\omega_j}(t_*)$ are the windowed pole contributions defined in \eqref{eq:windowed-pole-contribution}.
\end{proposition}

\begin{proof}
We integrate the meromorphic $H^{s-1}\to H^{s+1}$--valued function
\[
\omega\ \longmapsto\ e^{-i\omega t_*}\,\widetilde g_{\pm,\ell}(\omega)\,\chi A_{\pm,h_\ell}\Pi_{k_\pm}\,R(\omega)\,\widehat F_\vartheta(\omega)
\]
over the boundary of the rectangle with horizontal edges $\Im\omega=-\nu_{n-1}$ and $\Im\omega=-\nu_n$ and vertical edges $\Re\omega=\pm R$.
The vertical contributions vanish along a pole-avoiding sequence $R\to\infty$ by Lemma~\ref{lem:vertical-vanish}, using Schwartz decay of
$\widehat F_\vartheta$ (Lemma~\ref{lem:Fhat-Schwartz}), polynomial growth of $\|\chi R(\omega)\chi\|$ on compact $\Im\omega$--intervals away from poles,
and the polynomial bounds for $\widetilde g_{\pm,\ell}$ from Lemma~\ref{lem:weight-uniform}.
The residue theorem yields \eqref{eq:band-isolation-identity}; the residue at each pole is precisely the windowed contribution
\eqref{eq:windowed-pole-contribution} by Lemma~\ref{lem:pole-contribution}.
\end{proof}

\begin{corollary}[Band-isolated one-mode model in the equatorial sector]\label{cor:band-isolated-one-mode}
Assume in addition Theorems~\ref{thm:input-pole-isolation} and \ref{thm:input-sectorization}.
Let $n\ge 1$ and choose $\nu_{n-1}<\nu_n$ as above.
Then for every $N\in\mathbb N$ there exists $\ell_0(N)$ such that for all $\ell\ge \ell_0(N)$,
\begin{equation}\label{eq:band-one-mode}
\mathsf O\big(\mathsf R^{(\nu_{n-1})}_{\pm,\ell}(t_*)-\mathsf R^{(\nu_n)}_{\pm,\ell}(t_*)\big)
=a_{\pm,\ell}\,e^{-i\omega_{n,\pm,\ell} t_*}+\mathcal O(\ell^{-N}),
\end{equation}
uniformly for $t_*\in[1,\infty)$ and $(M,a)\in\mathcal K$, where $a_{\pm,\ell}$ are the corresponding detector amplitudes.
\end{corollary}

\begin{proof}
By Proposition~\ref{prop:overtone-band-subtraction}, the left-hand side is a symmetric sum over pole contributions from the strip
$-\nu_n<\Im\omega\le -\nu_{n-1}$.  By the construction of $\nu_{n-1}$ and $\nu_n$, this strip contains no equatorial poles from overtones
$j\le n-1$ and no equatorial poles from overtones $j\ge n+1$.
The sectorization estimate \eqref{eq:sectorization} implies that, after applying the equatorial cutoff and the detector $\mathsf O$,
all poles outside the equatorial disks $D_{n,\pm,\ell}$ contribute $\mathcal O(\ell^{-N})$.
Finally, Theorem~\ref{thm:input-pole-isolation} yields that each $D_{n,\pm,\ell}$ contains exactly one pole, namely $\omega_{n,\pm,\ell}$,
which is simple; thus the contribution is a single exponential.
\end{proof}

\subsection{Two-mode dominance theorem}\label{subsec:two-mode-theorem}

We now combine the analytic-windowed resonance expansion (Theorem~\ref{thm:windowed-expansion}) with the pole suppression
Lemma~\ref{lem:pole-suppression} and the microlocal bounds from Proposition~\ref{prop:equatorial-projector-poly} and
Theorem~\ref{thm:input-sectorization}.

\begin{theorem}[Two-mode dominance in the equatorial high-frequency package]\label{thm:two-mode-dominance}
Fix $\Lambda>0$, a compact parameter set $\mathcal K\Subset\mathcal P_\Lambda\cap\{|a|\le a_0\}$ with $a_0$ small, and an
overtone index $n\in\mathbb N_0$.  Let $\nu>0$ be chosen as in Lemma~\ref{lem:layer-gap}, and assume
Theorems~\ref{thm:input-pole-isolation} and \ref{thm:input-sectorization} as well as Assumption~\ref{ass:initial-support}.
Then there exist $\ell_1\in\mathbb N$ and $N_0\in\mathbb N$ such that for all
$\ell\ge\ell_1$, all $(M,a)\in\mathcal K$, and all initial data $(f_0,f_1)\in \mathcal H^{s+N_0}$ supported in $K_0$,
the refined equatorial signals \eqref{eq:refined-window} satisfy, for $t_*\ge 1$,
\begin{equation}\label{eq:one-mode-pm}
\begin{aligned}
u^{(+,\ell)}(t_*)
&= e^{-i\omega_{+,\ell} t_*}\,\mathcal A_{+,\ell}(f_0,f_1) + \mathcal E_{+,\ell}(t_*),\\
u^{(-,\ell)}(t_*)
&= e^{-i\omega_{-,\ell} t_*}\,\mathcal A_{-,\ell}(f_0,f_1) + \mathcal E_{-,\ell}(t_*),
\end{aligned}
\end{equation}
where:
\begin{enumerate}
\item $\omega_{\pm,\ell}=\omega_{n,\ell,\pm\ell}(M,a)$ are the labeled QNMs in the equatorial $k=\pm\ell$ sectors;
\item the amplitudes $\mathcal A_{\pm,\ell}$ are finite-rank operators $\mathcal H^{s+N_0}\to H^{s+1}(X_{M,a})$ given by the simple-pole residue
formula
\begin{equation}\label{eq:amplitude-formula}
\mathcal A_{\pm,\ell}(f_0,f_1)
:= i\,\chi A_{\pm,h_\ell}\Pi_{k_\pm}\,
\Pi_{\omega_{\pm,\ell}}\,\chi\,
\Big(\widetilde g_{\pm,\ell}(\omega_{\pm,\ell})\,\widehat F_\vartheta(\omega_{\pm,\ell})\Big),
\qquad
\Pi_{\omega_{\pm,\ell}}:=\Res_{\omega=\omega_{\pm,\ell}}R(\omega);
\end{equation}
\item the errors satisfy the following quantitative tail--plus--leakage bound:
for every $m\in\mathbb N$ and every $N'\in\mathbb N$ there exist an integer $N_{m,N'}\in\mathbb N$
and a constant $C_{s,m,N'}>0$ such that for all $(f_0,f_1)\in \mathcal H^{s+m+N_{m,N'}}$ supported in $K_0$,
\begin{equation}\label{eq:two-mode-error}
\|\mathcal E_{\pm,\ell}(t_*)\|_{H^{s+1}}
\le
C_{s,m,N'}\Big(e^{-\nu t_*}\,(1+t_*)^{-m} + \ell^{-N'}\Big)\,
\|(f_0,f_1)\|_{\mathcal H^{s+m+N_{m,N'}}},
\qquad t_*\ge 1,
\end{equation}
with constants uniform in $(M,a)\in\mathcal K$ and $\ell\ge\ell_1$.
\end{enumerate}
In particular, the combined equatorial high-frequency signal $u^{\mathrm{eq},\ell}(t_*):=u^{(+,\ell)}(t_*)+u^{(-,\ell)}(t_*)$
obeys
\begin{equation}\label{eq:two-mode-model}
u^{\mathrm{eq},\ell}(t_*)
=
e^{-i\omega_{+,\ell} t_*}\,\mathcal A_{+,\ell}(f_0,f_1)
+
e^{-i\omega_{-,\ell} t_*}\,\mathcal A_{-,\ell}(f_0,f_1)
+
\mathcal E_{\ell}(t_*),
\end{equation}
with $\|\mathcal E_{\ell}(t_*)\|_{H^{s+1}}$ bounded by the right-hand side of \eqref{eq:two-mode-error} after summing the $\pm$ estimates.
\end{theorem}

\begin{proof}
Apply Theorem~\ref{thm:windowed-expansion}.  The resonant sum in \eqref{eq:windowed-expansion} is a sum of windowed pole contributions
$\mathsf U^{(\pm,\ell)}_{\omega_j}(t_*)$ plus the tail $\mathsf R^{(\nu)}_{\pm,\ell}(t_*)$.
We split the poles $\omega_j$ with $\Im\omega_j>-\nu$ into three classes.

\smallskip
\noindent\emph{(i) Poles in the labeled equatorial disks with $0\le j\le n-1$.}
These are the lower overtones $\omega_{j,\pm,\ell}$, $j\in\{0,1,\dots,n-1\}$.
By Theorem~\ref{thm:input-pole-isolation}, these poles are simple, hence \eqref{eq:windowed-pole-simple} applies.
Lemma~\ref{lem:pole-suppression} gives $|\widetilde g_{\pm,\ell}(\omega_{j,\pm,\ell})|=\mathcal O(\ell^{-N_1})$ for every $N_1$.
Moreover, by Lemma~\ref{lem:Fhat-Schwartz} and $|\omega_{j,\pm,\ell}|\sim \ell$, for every $N_2$ one has
\[
\|\widehat F_\vartheta(\omega_{j,\pm,\ell})\|_{H^{s-1}}
\ \le\ C_{s,N_2}\,\ell^{-N_2}\,\|(f_0,f_1)\|_{\mathcal H^{s+m+N_2}}.
\]
Finally, Proposition~\ref{prop:equatorial-projector-poly} gives a polynomial bound on the microlocalized residue projectors.
Combining these bounds and summing over $0\le j\le n-1$ yields a total contribution $\mathcal O(\ell^{-N'})$ in $H^{s+1}$ for any $N'$.

\smallskip
\noindent\emph{(ii) Poles outside the labeled equatorial disks.}
Let $\omega_j\in\Res(P)$ satisfy $\Im\omega_j>-\nu$ and $\omega_j\notin\bigcup_{0\le m\le n}D_{m,\pm,\ell}$.
For such poles we use the general windowed formula \eqref{eq:windowed-pole-contribution}.
Each term involves a generalized Laurent coefficient $\Pi_{\omega_j}^{[q]}$ multiplied by a derivative
$\partial_\omega^{\,r}(\widetilde g_{\pm,\ell}\widehat F_\vartheta)(\omega_j)$ with $0\le r\le q-1$.
By the product rule, Lemma~\ref{lem:Fhat-Schwartz}, and Lemma~\ref{lem:weight-uniform}, these derivatives have rapid decay in
$\Re\omega_j$ (uniformly on $\Im\omega_j\ge -\nu$) up to polynomial factors, while Theorem~\ref{thm:input-sectorization}
gives $\|\chi A_{\pm,h_\ell}\Pi_{k_\pm}\Pi_{\omega_j}^{[q]}\chi\|_{H^{s-1}\to H^{s+1}}=\mathcal O(\ell^{-N})$ for arbitrary $N$.
Since we work on a compact subextremal slow--rotation set $\mathcal K$, exponential local energy decay implies a uniform spectral gap: there exists $\nu_{\mathrm{gap}}>0$ (depending only on $\mathcal K$) such that the cutoff resolvent has no poles in $\{\Im\omega>-\nu_{\mathrm{gap}}\}$ except for the simple stationary pole at $\omega=0$; see \cite{DyatlovQNM,VasyKerrDeSitter,PetersenVasyKerrDeSitterJEMS}. After azimuthal projection $\Pi_{k_\pm}$ with $k_\pm=\pm\ell\neq0$, the stationary pole is annihilated. Consequently, every pole $\omega_j$ in class~(ii) satisfies $\Im\omega_j\le -\nu_{\mathrm{gap}}$, and polynomial-in-time factors from possible higher-order poles are uniformly controlled on $t_*\ge1$ by the elementary estimate
\begin{equation}\label{eq:t-poly-exp-bound}
\sup_{t_*\ge 1}\ t_*^{p}\,|e^{-i\omega_j t_*}|
\ \le\ \sup_{t\ge0} t^{p}e^{-\nu_{\mathrm{gap}} t}
\ \le\ C_{p}\,\nu_{\mathrm{gap}}^{-p},
\qquad p\in\mathbb N_0,
\end{equation}
valid uniformly for all poles with $\Im\omega_j\le-\nu_{\mathrm{gap}}$.
Using polynomial resonance counting in strips \cite[\S3]{DyatlovAHP}, we may sum over all poles in class~(ii) and obtain a total
contribution $\mathcal O(\ell^{-N'})$ in $H^{s+1}$ for every $N'$.

\smallskip
\noindent\emph{(iii) The target pole $\omega_{\pm,\ell}=\omega_{n,\pm,\ell}$.}
By Theorem~\ref{thm:input-pole-isolation} the target pole is simple, hence \eqref{eq:windowed-pole-simple} applies.
Lemma~\ref{lem:pole-suppression} gives $\widetilde g_{\pm,\ell}(\omega_{\pm,\ell})=1+\mathcal O(\ell^{-N'})$,
and the leading contribution is precisely the exponential term in \eqref{eq:one-mode-pm} with amplitude \eqref{eq:amplitude-formula}.
The multiplicative $\mathcal O(\ell^{-N'})$ error is absorbed into the leakage term in \eqref{eq:two-mode-error}.

Finally, the tail estimate $e^{-\nu t_*}(1+t_*)^{-m}$ follows from \eqref{eq:windowed-remainder-bound} in
Theorem~\ref{thm:windowed-expansion}.  Summing the $\pm$ statements gives \eqref{eq:two-mode-model}.
\end{proof}

\begin{remark}[Bandpass alternative for overtone selection]\label{rem:bandpass-alternative}
The overtone isolation in Theorem~\ref{thm:two-mode-dominance} is achieved by the entire factor $\widetilde g_{\pm,\ell}$,
which suppresses the lower overtones $0,\dots,n-1$ while keeping uniform control on $\Gamma_{-\nu}$.
Corollary~\ref{cor:band-isolated-one-mode} provides an independent mechanism: subtracting two contour shifts at heights
$\nu_{n-1}<\nu_n$ isolates the $n$th overtone layer directly.
The resulting band-isolated signal satisfies a one-mode model with an $\mathcal O(\ell^{-N})$ error and no exponentially decaying tail.
All deterministic extraction statements in Section~\ref{sec:freq-extraction} apply verbatim to either choice of filtered time series.
\end{remark}

\begin{remark}[Nontriviality and genericity of the amplitudes]\label{rem:amplitude-nontrivial}
Theorem~\ref{thm:two-mode-dominance} gives an explicit residue formula for the mode amplitudes.
For fixed $\ell$ and fixed mode label, the set of detectors that annihilate the leading mode is a proper closed hyperplane in the detector space;
see Lemma~\ref{lem:generic-detector}.  Thus for generic equatorially localized initial data and generic detectors the leading scalar amplitudes are nonzero.
A quantitative lower bound (a ``detectability'' condition) is discussed in Section~\ref{subsec:obs-model}, and Appendix~\ref{app:dual-states}
records the dual-state interpretation in the simple-pole regime.
\end{remark}

\begin{lemma}[Generic detectors do not annihilate a fixed mode]\label{lem:generic-detector}
Let $X$ be a Banach space and $0\neq u\in X$. Then
\[
\{\,\Lambda\in X^*: \Lambda(u)=0\,\}
\]
is a proper closed hyperplane of the dual space $X^*$. In particular, its complement is open and dense in $X^*$.
\end{lemma}

\begin{proof}
The map $X^*\ni \Lambda\mapsto \Lambda(u)\in\mathbb C$ is continuous and nonzero, hence its kernel is a closed codimension-one subspace.
Properness follows from Hahn--Banach: there exists $\Lambda_0\in X^*$ with $\Lambda_0(u)\neq 0$.
\end{proof}

\subsection{Implications for the inverse problem}\label{subsec:inverse-interface}

Theorem~\ref{thm:two-mode-dominance} provides the PDE-to-data input needed in Section~\ref{sec:application}:
after equatorial microlocalization, azimuthal selection $k=\pm\ell$, and analytic frequency localization at
$\omega^\sharp_{\pm,\ell}$, the time-domain signal is a single exponential $e^{-i\omega_{\pm,\ell}t_*}$ plus an explicitly controlled tail.
Combining the $+\ell$ and $-\ell$ sectors yields a two-exponential model whose frequencies are precisely the equatorial QNMs
used in the inverse theorem of the companion paper \cite{LiPartI}.

\section{Deterministic frequency extraction stability}\label{sec:freq-extraction}

In this section we record a deterministic, quantitative stability statement for extracting a complex frequency from a scalar ringdown signal.
The point here is not to introduce a new signal-processing algorithm, but to isolate the perturbation mechanism used later in the PDE-to-parameter pipeline:
after the microlocal and analytic localizations of Section~\ref{sec:two-mode}, each equatorial package reduces to a single exponentially damped sinusoid
plus a deterministic tail, and the remaining frequency recovery step is an elementary consequence of time-shift invariance.
Closely related ideas appear in classical Prony-type methods and in the matrix-pencil/ESPRIT literature; see
\cite{HuaSarkarMPM1990,PottsTascheLAA2013,BatenkovYomdinSJAM2013} for representative analyses from different viewpoints.

\subsection{From a field to a scalar time series}\label{subsec:obs}

To avoid any clash with the energy spaces $\mathcal H^s$ fixed in Section~\ref{sec:setup2}, we denote by $\mathcal H_{\mathrm{obs}}$
the Hilbert space in which our time-dependent field is observed.
In our application, $\mathcal H_{\mathrm{obs}}=H^{s+1}(X_{M,a})$.
Let $\mathsf O:\mathcal H_{\mathrm{obs}}\to\mathbb C$ be a bounded linear ``detector''.
Given a time-dependent field $v(t_*)\in\mathcal H_{\mathrm{obs}}$, define the observed scalar signal
\[
y(t_*):=\mathsf O(v(t_*)).
\]
If $v(t_*)=A\,e^{-i\omega t_*}+r(t_*)$ with $A\in\mathcal H_{\mathrm{obs}}$, $\omega\in\mathbb C$, then
\begin{equation}\label{eq:scalar-model}
y(t_*)=a\,e^{-i\omega t_*}+\rho(t_*),
\qquad
a:=\mathsf O(A)\in\mathbb C,\quad \rho(t_*):=\mathsf O(r(t_*)).
\end{equation}
All bounds below will be expressed in terms of $\rho$ and $a$; in particular,
\begin{equation}\label{eq:rho-bound}
|\rho(t_*)|\le \|\mathsf O\|\,\|r(t_*)\|_{\mathcal H_{\mathrm{obs}}}.
\end{equation}
Thus, applying the results of this section with $y=\mathsf O(u^{(\pm,\ell)})$ yields deterministic frequency bounds
for $\omega_{\pm,\ell}$ as soon as the detector does not annihilate the leading amplitude ($a\neq 0$).

\subsection{A robust one-mode estimator based on time-shift invariance}\label{subsec:one-mode-estimator}

We begin with the one-mode model
\begin{equation}\label{eq:one-mode}
y(t)=a\,e^{-i\omega t}+\rho(t),\qquad t\in [T_0,T_0+T],
\end{equation}
where $\omega\in\mathbb C$ is unknown (with $\Im\omega<0$ in ringdown applications),
$a\in\mathbb C$ is an unknown nonzero amplitude, and $\rho$ is an unknown deterministic error (tail + measurement noise).
We write $t$ (instead of $t_*$) for notational simplicity.

\subsubsection{Weighted shift Rayleigh quotient}

Fix a shift step $\Delta>0$ with $\Delta<T$.
Let $w\in L^\infty([T_0,T_0+T-\Delta])$ be a nonnegative weight such that
\begin{equation}\label{eq:weight}
0\le w(t)\le 1,\qquad
w(t)\equiv 1 \text{ on } [T_0+\Delta,\,T_0+T-2\Delta],
\end{equation}
so that both $t$ and $t+\Delta$ lie in the observation window on the essential support of $w$.
Define the weighted inner product and norm
\[
\langle f,g\rangle_w := \int_{T_0}^{T_0+T-\Delta} w(t)\,f(t)\,\overline{g(t)}\,dt,
\qquad
\|f\|_w := \langle f,f\rangle_w^{1/2}.
\]
Let $(S_\Delta f)(t):=f(t+\Delta)$ (defined a.e.\ on $[T_0,T_0+T-\Delta]$).
If $\rho\equiv 0$, then $S_\Delta y = z\,y$ with $z:=e^{-i\omega\Delta}$, so $z$ is an eigenvalue of the shift.

Motivated by this, define the \emph{shift Rayleigh quotient estimator}
\begin{equation}\label{eq:z-hat}
\widehat z := \frac{\langle S_\Delta y,\ y\rangle_w}{\langle y,\ y\rangle_w}
\qquad\text{whenever }\langle y,y\rangle_w>0.
\end{equation}
The estimator $\widehat z$ is determined entirely from the observed scalar time series $y$.
We next quantify its stability under deterministic perturbations.

\begin{remark}[Relation to Prony, matrix pencils, and ESPRIT]\label{rem:prony-esprit}
In the noiseless model $y(t)=Ae^{\omega t}$, the shift operator $(S_\Delta y)(t):=y(t+\Delta)$ acts on the one-dimensional
Krylov space $\mathrm{span}\{y\}$ by multiplication with $z=e^{\omega\Delta}$.
The estimator \eqref{eq:z-hat} is precisely the Rayleigh quotient of $S_\Delta$ restricted to this subspace, and may be viewed
as the $1\times 1$ case of the classical matrix-pencil and ESPRIT constructions for sums of exponentials; see
\cite{HuaSarkarMPM1990,RoyKailath1989ESPRIT} and the discussion in \cite{BatenkovYomdinSJAM2013}.
The novelty here is not the signal-processing mechanism, but the \emph{deterministic} propagation of PDE remainder bounds and
measurement perturbations into explicit frequency error estimates that are uniform in the high-frequency parameter $\ell$.
\end{remark}

\subsubsection{A quantitative perturbation lemma}

Let $y_0(t):=a\,e^{-i\omega t}$ denote the pure exponential component, so $y=y_0+\rho$ and $S_\Delta y_0=z\,y_0$.
The key point is that the shift $S_\Delta$ evaluates $\rho$ on a \emph{shifted} part of the observation window,
so the stability estimate must control both $\rho$ and $S_\Delta\rho$ in the same weighted norm.

\begin{lemma}[Stability of the shift Rayleigh quotient]\label{lem:rayleigh-stability}
Let $y=y_0+\rho$ with $y_0(t)=a e^{-i\omega t}$, $a\neq 0$, and $z=e^{-i\omega\Delta}$.
Assume $\Im\omega\le 0$ so that $|z|\le 1$.
Define the relative error sizes
\begin{equation}\label{eq:eps01}
\varepsilon_0:=\frac{\|\rho\|_w}{\|y_0\|_w},
\qquad
\varepsilon_1:=\frac{\|S_\Delta\rho\|_w}{\|y_0\|_w},
\qquad
\varepsilon:=\max\{\varepsilon_0,\varepsilon_1\}.
\end{equation}
If $\varepsilon_0\le 1/4$, then $\langle y,y\rangle_w\ge (1-2\varepsilon_0)\|y_0\|_w^2\ge \frac12\|y_0\|_w^2$ and
\begin{equation}\label{eq:z-error}
|\widehat z - z|
\ \le\
\frac{(\varepsilon_0+\varepsilon_1)(1+\varepsilon_0)}{1-2\varepsilon_0}.
\end{equation}
In particular, if $\varepsilon\le 1/8$, then
\begin{equation}\label{eq:z-error-crude}
|\widehat z-z|\ \le\ 3\,\varepsilon,
\qquad
|\widehat z|\ \le\ |z|+3\,\varepsilon.
\end{equation}
\end{lemma}

\begin{proof}
Write $D:=\langle y,y\rangle_w$ and $N:=\langle S_\Delta y,y\rangle_w$.
We estimate the denominator first:
\[
D=\langle y_0+\rho,\ y_0+\rho\rangle_w
=\|y_0\|_w^2 + 2\Re\langle y_0,\rho\rangle_w + \|\rho\|_w^2.
\]
By Cauchy--Schwarz,
\[
|\langle y_0,\rho\rangle_w|\le \|y_0\|_w\|\rho\|_w=\varepsilon_0\|y_0\|_w^2,
\qquad
\|\rho\|_w^2=\varepsilon_0^2\|y_0\|_w^2,
\]
hence
\[
D\ge \|y_0\|_w^2-2\varepsilon_0\|y_0\|_w^2=(1-2\varepsilon_0)\|y_0\|_w^2.
\]
If $\varepsilon_0\le 1/4$ this gives $D\ge \frac12\|y_0\|_w^2>0$, so $\widehat z$ is well-defined.

For the numerator we use $S_\Delta y_0=z\,y_0$:
\[
N=\langle S_\Delta(y_0+\rho),\ y_0+\rho\rangle_w
=\langle z y_0,\ y_0\rangle_w+\langle z y_0,\rho\rangle_w+\langle S_\Delta\rho,\ y_0\rangle_w+\langle S_\Delta\rho,\rho\rangle_w.
\]
Subtracting $zD$ and regrouping,
\begin{align*}
N-zD
&=\langle S_\Delta\rho,\ y_0\rangle_w - z\langle \rho,\ y_0\rangle_w
  +\langle S_\Delta\rho,\rho\rangle_w - z\langle \rho,\rho\rangle_w\\
&=\langle S_\Delta\rho-z\rho,\ y_0\rangle_w+\langle S_\Delta\rho-z\rho,\ \rho\rangle_w\\
&=\langle S_\Delta\rho-z\rho,\ y_0+\rho\rangle_w.
\end{align*}
We bound the perturbation term in the weighted norm:
\[
\|S_\Delta\rho-z\rho\|_w
\le \|S_\Delta\rho\|_w + |z|\,\|\rho\|_w
\le \|S_\Delta\rho\|_w + \|\rho\|_w
=(\varepsilon_1+\varepsilon_0)\|y_0\|_w,
\]
using $|z|\le 1$.
Therefore,
\begin{align*}
|N-zD|
&\le \|S_\Delta\rho-z\rho\|_w\,\|y_0+\rho\|_w\\
&\le (\varepsilon_0+\varepsilon_1)\|y_0\|_w\big(\|y_0\|_w+\|\rho\|_w\big)
=(\varepsilon_0+\varepsilon_1)(1+\varepsilon_0)\|y_0\|_w^2.
\end{align*}
Dividing by the denominator lower bound yields
\[
|\widehat z-z|=\left|\frac{N}{D}-z\right|=\frac{|N-zD|}{D}
\le \frac{(\varepsilon_0+\varepsilon_1)(1+\varepsilon_0)\|y_0\|_w^2}{(1-2\varepsilon_0)\|y_0\|_w^2}
=
\frac{(\varepsilon_0+\varepsilon_1)(1+\varepsilon_0)}{1-2\varepsilon_0},
\]
which is \eqref{eq:z-error}.

For the crude bound, assume $\varepsilon\le 1/8$. Then $\varepsilon_0\le 1/8$ and $\varepsilon_1\le 1/8$, so
\[
|\widehat z-z|
\le
\frac{(2\varepsilon)(1+\varepsilon)}{1-2\varepsilon}
\le
\frac{2\varepsilon\cdot (9/8)}{3/4}
=3\varepsilon,
\]
proving \eqref{eq:z-error-crude}. The bound on $|\widehat z|$ follows from $|\widehat z|\le |z|+|\widehat z-z|$.
\end{proof}

\begin{remark}[A convenient sufficient condition for $\varepsilon_1\lesssim \varepsilon_0$]\label{rem:shift-compatible}
If the weight satisfies a shift-compatibility condition of the form
$w(t-\Delta)\le C_w\,w(t)$ a.e.\ on $[T_0+\Delta,T_0+T-\Delta]$ (after extending $w$ by $0$ outside its domain),
then a change of variables gives $\|S_\Delta f\|_w\le C_w^{1/2}\|f\|_w$ and hence $\varepsilon_1\le C_w^{1/2}\varepsilon_0$.
In particular, for the unweighted choice $w\equiv 1$ one has $C_w=1$.
We do \emph{not} impose such a condition in general; instead, we keep $\varepsilon_1$ explicit since it is directly controlled
by the same tail bounds as $\varepsilon_0$ in our PDE application.
\end{remark}

\subsection{From shift eigenvalue to frequency: branch control}\label{subsec:log-branch}

The estimator \eqref{eq:z-hat} returns $\widehat z\approx z=e^{-i\omega\Delta}$.
To convert this into an estimate for $\omega$, we must choose a branch of the complex logarithm.

\subsubsection{Logarithm on a controlled neighborhood}

We denote by $\Log$ the principal branch on $\mathbb C\setminus(-\infty,0]$.

\begin{lemma}[Local Lipschitz control for the logarithm]\label{lem:log-Lip}
Let $z\in\mathbb C\setminus\{0\}$ and assume $\widehat z$ satisfies $|\widehat z-z|\le \frac12|z|$.
Then $\widehat z$ is nonzero and
\begin{equation}\label{eq:log-Lip}
|\Log \widehat z - \Log z|\ \le\ \frac{2}{|z|}\,|\widehat z-z|,
\end{equation}
where $\Log$ is any holomorphic branch of the logarithm on a simply connected domain containing
the closed segment between $z$ and $\widehat z$ and avoiding $0$.
\end{lemma}

\begin{proof}
Consider the straight-line path $\gamma(s)=z+s(\widehat z-z)$, $s\in[0,1]$.
The assumption implies $|\gamma(s)|\ge |z|-|\widehat z-z|\ge |z|/2$ for all $s$, so $\gamma$ avoids $0$.
For any holomorphic branch of $\Log$ on a domain containing $\gamma$, the fundamental theorem of calculus yields
\[
\Log \widehat z-\Log z = \int_0^1 \frac{\gamma'(s)}{\gamma(s)}\,ds = \int_0^1 \frac{\widehat z-z}{\gamma(s)}\,ds.
\]
Taking absolute values and using $|\gamma(s)|\ge |z|/2$ gives \eqref{eq:log-Lip}.
\end{proof}

\subsubsection{Branch selection using a pseudopole prior}

The map $\omega\mapsto z=e^{-i\omega\Delta}$ is $2\pi/\Delta$ periodic in $\Re\omega$.
Thus, even if one determines $z$ exactly, the frequency is only determined modulo $2\pi/\Delta$ in real part.
In our application this ambiguity is removed by the semiclassical prior $\omega^\sharp_{\pm,\ell}$ (the pseudopole)
and the analytic localization window of Section~\ref{sec:two-mode}.
The next lemma makes this branch selection explicit and reduces it to the principal logarithm of a ratio close to $1$.

\begin{lemma}[Branch selection from a prior]\label{lem:branch-selection}
Fix $\Delta>0$ and let $\omega^\sharp\in\mathbb C$ be a prior.
Set $z^\sharp:=e^{-i\omega^\sharp\Delta}$.
Assume that $z=e^{-i\omega\Delta}$ and $\widehat z$ satisfy
\begin{equation}\label{eq:branch-hyp}
|z-z^\sharp|\le \frac14|z^\sharp|,
\qquad
|\widehat z-z|\le \frac12|z|.
\end{equation}
Then $\widehat z/z^\sharp$ belongs to the disk $\{\,\zeta:|\zeta-1|\le 5/8\,\}\subset\mathbb C\setminus(-\infty,0]$,
so the principal logarithm $\Log$ is holomorphic at $\widehat z/z^\sharp$.
Define
\begin{equation}\label{eq:log-prior}
\Log_{\omega^\sharp}(\widehat z)
:=
\Log\!\left(\frac{\widehat z}{z^\sharp}\right)-i\omega^\sharp\Delta,
\qquad
\widehat\omega:=\frac{i}{\Delta}\,\Log_{\omega^\sharp}(\widehat z).
\end{equation}
Then $\Log_{\omega^\sharp}$ is a holomorphic logarithm branch on a neighborhood of $\widehat z$, satisfies
$\exp(\Log_{\omega^\sharp}(\widehat z))=\widehat z$ and $\Log_{\omega^\sharp}(z^\sharp)=-i\omega^\sharp\Delta$, and moreover
\begin{equation}\label{eq:log-prior-Lip}
|\widehat\omega-\omega|
\ \le\
\frac{2}{\Delta\,|z|}\,|\widehat z-z|.
\end{equation}
\end{lemma}

\begin{proof}
From $|z-z^\sharp|\le \frac14|z^\sharp|$ we obtain $|z|\ge |z^\sharp|-|z-z^\sharp|\ge \frac34|z^\sharp|$.
Hence
\[
|\widehat z-z^\sharp|
\le |\widehat z-z|+|z-z^\sharp|
\le \frac12|z|+\frac14|z^\sharp|
\le \frac12\cdot \frac34|z^\sharp|+\frac14|z^\sharp|
=\frac58|z^\sharp|.
\]
Dividing by $|z^\sharp|$ gives $|\widehat z/z^\sharp-1|\le 5/8$, proving that $\widehat z/z^\sharp$ lies in the stated disk.
This disk is contained in the half-plane $\{\Re\zeta\ge 3/8\}$ and therefore avoids $(-\infty,0]$; thus the principal $\Log$
is holomorphic there, and \eqref{eq:log-prior} defines a holomorphic function of $\widehat z$.
The identity $\exp(\Log_{\omega^\sharp}(\widehat z))=\widehat z$ follows from $\exp(\Log(\widehat z/z^\sharp))=\widehat z/z^\sharp$
and $\exp(-i\omega^\sharp\Delta)=z^\sharp$.

For the Lipschitz bound, note that
\[
\widehat\omega-\omega
=
\frac{i}{\Delta}\Big(\Log(\widehat z/z^\sharp)-\Log(z/z^\sharp)\Big),
\qquad
\frac{\widehat z}{z^\sharp}-\frac{z}{z^\sharp}=\frac{\widehat z-z}{z^\sharp}.
\]
Moreover, $|z/z^\sharp|=|z|/|z^\sharp|\ge 3/4$ and the hypothesis $|\widehat z-z|\le \frac12|z|$ implies
$\big|\frac{\widehat z}{z^\sharp}-\frac{z}{z^\sharp}\big|\le \frac12|z/z^\sharp|$.
Applying Lemma~\ref{lem:log-Lip} to $z/z^\sharp$ and $\widehat z/z^\sharp$ gives
\[
|\Log(\widehat z/z^\sharp)-\Log(z/z^\sharp)|
\le \frac{2}{|z/z^\sharp|}\,\Big|\frac{\widehat z-z}{z^\sharp}\Big|
=\frac{2}{|z|}\,|\widehat z-z|.
\]
Multiplying by $1/\Delta$ yields \eqref{eq:log-prior-Lip}.
\end{proof}

\subsubsection{A deterministic one-mode frequency stability theorem}

We now combine Lemma~\ref{lem:rayleigh-stability} with the logarithm control above.
Define the frequency estimator
\begin{equation}\label{eq:omega-hat}
\widehat\omega
:=
\frac{i}{\Delta}\,\Log(\widehat z),
\qquad \widehat z \text{ as in \eqref{eq:z-hat}},
\end{equation}
where $\Log$ is a chosen logarithm branch.
In applications we will take $\Log=\Log_{\omega^\sharp}$ from Lemma~\ref{lem:branch-selection}
with $\omega^\sharp=\omega^\sharp_{\pm,\ell}$.

\begin{theorem}[Deterministic one-mode frequency extraction]\label{thm:one-mode-extraction}
Let $y(t)=a e^{-i\omega t}+\rho(t)$ on $[T_0,T_0+T]$ with $a\neq 0$ and $\Im\omega\le 0$.
Fix $\Delta\in(0,T)$ and a weight $w$ satisfying \eqref{eq:weight}, and define $\widehat z$ by \eqref{eq:z-hat}.
Let $\varepsilon_0,\varepsilon_1,\varepsilon$ be as in \eqref{eq:eps01} and set $z=e^{-i\omega\Delta}$.

Assume $\varepsilon\le \min\{1/8,\,|z|/20\}$.
Then $\widehat z\neq 0$, $|\widehat z-z|\le \frac12|z|$, and for any holomorphic logarithm branch $\Log$ on a simply connected domain
containing the segment between $z$ and $\widehat z$ and avoiding $0$, the estimator \eqref{eq:omega-hat} satisfies
\begin{equation}\label{eq:omega-error}
|\widehat\omega-\omega|
\ \le\
\frac{10}{\Delta\,|z|}\,\varepsilon
\ =\
\frac{10}{\Delta}\,e^{-\Im\omega\Delta}\,\frac{\max\{\|\rho\|_w,\|S_\Delta\rho\|_w\}}{\|a e^{-i\omega t}\|_w}.
\end{equation}
\end{theorem}

\begin{proof}
By Lemma~\ref{lem:rayleigh-stability}, if $\varepsilon\le 1/8$ then $|\widehat z-z|\le 3\varepsilon$.
The additional assumption $\varepsilon\le |z|/20$ implies $|\widehat z-z|\le 3|z|/20<|z|/2$.
Lemma~\ref{lem:log-Lip} then yields
\[
|\Log\widehat z-\Log z|
\le \frac{2}{|z|}\,|\widehat z-z|
\le \frac{6}{|z|}\,\varepsilon
\le \frac{10}{|z|}\,\varepsilon.
\]
Finally,
\[
|\widehat\omega-\omega|
=
\frac{1}{\Delta}|\Log\widehat z-\Log z|
\le
\frac{10}{\Delta\,|z|}\varepsilon,
\]
which is \eqref{eq:omega-error}. Since $|z|=|e^{-i\omega\Delta}|=e^{\Im\omega\Delta}$, the last identity follows.
\end{proof}

\begin{remark}[Choosing the shift step $\Delta$]\label{rem:choose-Delta}
The bound \eqref{eq:omega-error} makes the dependence on the shift step explicit.  Taking $\Delta$ too small amplifies the prefactor $\Delta^{-1}$ and provides little phase separation, so that the ratio $\widehat z$ becomes sensitive to the residual term.  Taking $\Delta$ too large reduces the available fitting interval $[T_0,T_0+T-\Delta]$ and introduces the factor $|z|^{-1}=e^{-\Im\omega\Delta}$, which can significantly magnify errors when the damping is strong.  A practical choice is to take $\Delta$ as a fixed fraction of the window length $T$ while keeping $T-\Delta$ comfortably larger than the expected damping time, so that both the signal energy and the phase separation remain visible.  No optimal choice is claimed here; the point is that the stability constant is explicit and can be balanced against the observation constraints.
\end{remark}

\begin{remark}[Rayleigh-quotient viewpoint and relation to Prony-type methods]\label{rem:rayleigh-prony}
The estimator \eqref{eq:z-hat} can be viewed as a weighted \emph{shift Rayleigh quotient}.
Indeed, if we equip $L^2([T_0,T_0+T-\Delta])$ with the weighted inner product $\langle\cdot,\cdot\rangle_w$
from \eqref{eq:z-hat}, then
\[
\widehat z
=\frac{\langle S_\Delta y,\,y\rangle_w}{\langle y,\,y\rangle_w}
\]
is the Rayleigh quotient of the shift operator $S_\Delta$ evaluated at the test vector $y$.
In the exact one-mode case $\rho\equiv0$, $y$ is an eigenfunction of $S_\Delta$ and the quotient returns the corresponding eigenvalue
$z=e^{-i\omega\Delta}$.
For multi-exponential signals $y(t)=\sum_{j=1}^J a_j e^{-i\omega_j t}$, higher-dimensional variants arise by restricting $S_\Delta$
to $\mathrm{span}\{y,S_\Delta y,\dots,S_\Delta^{J-1}y\}$, leading to generalized eigenvalue problems that underlie
Prony/matrix-pencil/ESPRIT constructions; see, for instance, \cite{HuaSarkarMPM1990,PottsTascheLAA2013,BatenkovYomdinSJAM2013}.
The two-exponential conditioning bounds in Proposition~\ref{prop:prony-conditioning} quantify the expected loss of stability as
frequencies approach each other (a regime relevant near coalescence and exceptional points).
\end{remark}

\begin{remark}[Interpretation of the branch condition]\label{rem:branch}
Lemma~\ref{lem:branch-selection} shows that once a prior $\omega^\sharp$ localizes $z=e^{-i\omega\Delta}$ to a small neighborhood of
$z^\sharp=e^{-i\omega^\sharp\Delta}$, the correct logarithm branch is selected by taking the principal logarithm of $\widehat z/z^\sharp$.
In our Kerr--de~Sitter application, the pseudopole $\omega^\sharp_{\pm,\ell}$ and the analytic window of Section~\ref{sec:two-mode}
provide such localization, with an error $\mathcal O(\ell^{-\infty})$.
\end{remark}

\subsection{Applying the one-mode theorem to the equatorial ringdown package}\label{subsec:apply-equatorial}

We now specialize to the PDE output of Theorem~\ref{thm:two-mode-dominance}.
Fix $\ell\ge\ell_1$ and parameters $(M,a)\in\mathcal K$.
For definiteness, consider the $+$ sector (the $-$ sector is identical).
Let $v(t):=u^{(+,\ell)}(t)$ be the $\mathcal H_{\mathrm{obs}}$-valued signal (here $\mathcal H_{\mathrm{obs}}=H^{s+1}(X_{M,a})$),
so that by \eqref{eq:one-mode-pm}:
\[
v(t)=e^{-i\omega_{+,\ell} t}A_{+,\ell}+\mathcal E_{+,\ell}(t),\qquad t\ge 1,
\]
with $A_{+,\ell}:=\mathcal A_{+,\ell}(f_0,f_1)\in H^{s+1}(X_{M,a})$ and $\omega_{+,\ell}=\omega_{n,\ell,\ell}(M,a)$.

Applying a detector $\mathsf O$ gives a scalar signal $y(t)=\mathsf O(v(t))$ of the form
\begin{equation}\label{eq:y-equatorial}
y(t)=a_{+,\ell}\,e^{-i\omega_{+,\ell} t}+\rho_{+,\ell}(t),
\qquad
a_{+,\ell}:=\mathsf O(A_{+,\ell}),
\qquad
\rho_{+,\ell}(t):=\mathsf O(\mathcal E_{+,\ell}(t)).
\end{equation}
By \eqref{eq:rho-bound} and \eqref{eq:two-mode-error}, for every $m\in\mathbb N$ and every $N'\in\mathbb N$,
\begin{equation}\label{eq:rho-equatorial-bound}
\begin{aligned}
|\rho_{+,\ell}(t)|
&\le
C_{s,m,N'}\,\|\mathsf O\|\Big(e^{-\nu t}(1+t)^{-m}+\ell^{-N'}\Big)\,
\|(f_0,f_1)\|_{\mathcal H^{s+m+N_{m,N'}}},\\
&\hspace{1.2cm} t\ge 1.
\end{aligned}
\end{equation}
In particular, on any window $[T_0,T_0+T]$ with $T_0\ge 1$, we have
\begin{equation}\label{eq:rho-L2}
\begin{aligned}
\|\rho_{+,\ell}\|_{L^2([T_0,T_0+T])}
&\le
C_{s,m,N'}\,\|\mathsf O\|\Big(e^{-\nu T_0}(1+T_0)^{-m}+\ell^{-N'}\Big)\,\sqrt{T}\\
&\qquad\times
\|(f_0,f_1)\|_{\mathcal H^{s+m+N_{m,N'}}}.
\end{aligned}
\end{equation}

To apply Theorem~\ref{thm:one-mode-extraction}, we compare the residual energies
$\|\rho_{+,\ell}\|_w$ and $\|S_\Delta\rho_{+,\ell}\|_w$ to the dominant-mode energy $\|a_{+,\ell}e^{-i\omega_{+,\ell}t}\|_w$.
A convenient lower bound for $\|a e^{-i\omega t}\|_w$ is obtained by integrating $|e^{-i\omega t}|^2=e^{2\Im\omega\,t}$.

\begin{lemma}[Lower bound for the dominant-mode energy]\label{lem:mode-energy}
Let $y_0(t)=a e^{-i\omega t}$ with $a\neq 0$ and $\Im\omega<0$.
Let $w$ satisfy \eqref{eq:weight}. Then
\begin{equation}\label{eq:mode-energy}
\|y_0\|_w^2
\ \ge\
|a|^2\int_{T_0+\Delta}^{T_0+T-2\Delta} e^{2\Im\omega\,t}\,dt
\ =
|a|^2\,\frac{e^{2\Im\omega\,(T_0+\Delta)}-e^{2\Im\omega\,(T_0+T-2\Delta)}}{-2\Im\omega}.
\end{equation}
In particular,
\begin{equation}\label{eq:mode-energy-crude}
\|y_0\|_w
\ \ge\
|a|\,e^{\Im\omega\,(T_0+T-2\Delta)}\,\sqrt{T-3\Delta},
\qquad (T>3\Delta).
\end{equation}
\end{lemma}

\begin{proof}
Since $w\equiv 1$ on $[T_0+\Delta,T_0+T-2\Delta]$, we have
\begin{align*}
\|y_0\|_w^2
&=\int_{T_0}^{T_0+T-\Delta} w(t)\,|a|^2 e^{2\Im\omega\, t}\,dt\\
&\ge |a|^2\int_{T_0+\Delta}^{T_0+T-2\Delta} e^{2\Im\omega\, t}\,dt.
\end{align*}
which gives \eqref{eq:mode-energy}. The crude bound \eqref{eq:mode-energy-crude} follows from the inequality
$\int_I e^{2\Im\omega t}\,dt\ge |I|\cdot \min_I e^{2\Im\omega t}$.
\end{proof}

\begin{lemma}[Controlling the shifted residual by the $L^2$ tail]\label{lem:shifted-residual}
Let $\rho\in L^2([T_0,T_0+T])$ and let $w$ satisfy \eqref{eq:weight}. Then
\begin{equation}\label{eq:shifted-residual}
\max\{\|\rho\|_w,\ \|S_\Delta\rho\|_w\}\ \le\ \|\rho\|_{L^2([T_0,T_0+T])}.
\end{equation}
\end{lemma}

\begin{proof}
Since $0\le w\le 1$ and the $w$-integration domain is contained in $[T_0,T_0+T]$,
\[
\|\rho\|_w^2=\int_{T_0}^{T_0+T-\Delta} w(t)|\rho(t)|^2\,dt
\le \int_{T_0}^{T_0+T}|\rho(t)|^2\,dt.
\]
Similarly,
\begin{align*}
\|S_\Delta\rho\|_w^2
&=\int_{T_0}^{T_0+T-\Delta} w(t)\,|\rho(t+\Delta)|^2\,dt\\
&\le \int_{T_0}^{T_0+T-\Delta}|\rho(t+\Delta)|^2\,dt
= \int_{T_0+\Delta}^{T_0+T}|\rho(s)|^2\,ds\\
&\le \int_{T_0}^{T_0+T}|\rho(s)|^2\,ds.
\end{align*}
Taking square roots yields \eqref{eq:shifted-residual}.
\end{proof}

Combining Lemmas~\ref{lem:mode-energy}--\ref{lem:shifted-residual} with \eqref{eq:rho-L2} yields an explicit relative error bound.
We summarize the result in a corollary.

\begin{corollary}[Equatorial QNM extraction error bound]\label{cor:equatorial-extraction}
Fix $\ell\ge\ell_1$ and let $y(t)$ be the observed $+$--equatorial signal \eqref{eq:y-equatorial}
on $[T_0,T_0+T]$ with $T_0\ge 1$ and $T>3\Delta$.
Assume $a_{+,\ell}=\mathsf O(A_{+,\ell})\neq 0$ and $\Im\omega_{+,\ell}\le -c<0$ on $\mathcal K$.
Define
\begin{equation}\label{eq:eps-equatorial}
\varepsilon_{+,\ell}
:=
\frac{\max\{\|\rho_{+,\ell}\|_w,\ \|S_\Delta\rho_{+,\ell}\|_w\}}{\|a_{+,\ell}e^{-i\omega_{+,\ell}t}\|_w}.
\end{equation}
Then for every $m\in\mathbb N$ there exist a constant $C_{s,m}>0$ and an integer $N=N(s,m)$ such that
\begin{equation}\label{eq:eps-equatorial-bound}
\varepsilon_{+,\ell}
\ \le\
C_{s,m}\,
\frac{\|\mathsf O\|}{|a_{+,\ell}|}\,
e^{-(\nu+\Im\omega_{+,\ell})\,T_0}\,(1+T_0)^{-m}\,
\frac{\sqrt{T}}{\sqrt{T-3\Delta}}\,
\|(f_0,f_1)\|_{\mathcal H^{s+m+N}},
\end{equation}
and if $\varepsilon_{+,\ell}\le \min\{1/8,|z_{+,\ell}|/20\}$ (with $z_{+,\ell}=e^{-i\omega_{+,\ell}\Delta}$),
the estimator \eqref{eq:omega-hat} satisfies
\begin{equation}\label{eq:omega-equatorial}
|\widehat\omega_{+,\ell}-\omega_{+,\ell}|
\ \le\
\frac{10}{\Delta\,|z_{+,\ell}|}\,\varepsilon_{+,\ell}.
\end{equation}
An identical statement holds for the $-$ sector.
\end{corollary}

\begin{proof}
Lemma~\ref{lem:shifted-residual} gives
$\max\{\|\rho_{+,\ell}\|_w,\|S_\Delta\rho_{+,\ell}\|_w\}\le \|\rho_{+,\ell}\|_{L^2([T_0,T_0+T])}$.
Using \eqref{eq:rho-L2} we obtain an explicit upper bound for the numerator in \eqref{eq:eps-equatorial}.
For the denominator we use Lemma~\ref{lem:mode-energy} and the uniform bound $\Im\omega_{+,\ell}\le -c<0$ on $\mathcal K$,
which yields the lower bound \eqref{eq:mode-energy-crude} with constants absorbed into $C_{s,m}$.
This gives \eqref{eq:eps-equatorial-bound}.
Finally, apply Theorem~\ref{thm:one-mode-extraction} to obtain \eqref{eq:omega-equatorial}.
\end{proof}

\begin{remark}[Interpretation]\label{rem:interpret}
The factor $e^{-(\nu+\Im\omega)T_0}$ in \eqref{eq:eps-equatorial-bound} reflects a basic signal-to-tail competition:
the dominant mode decays like $e^{\Im\omega\,t}$, while the tail decays like $e^{-\nu t}$.
Since $\nu$ is chosen \emph{below} the dominant mode (i.e.\ $\nu>|\Im\omega|$ by construction in Section~\ref{sec:two-mode}),
the relative tail size improves exponentially as $T_0$ increases, up to the point where the absolute signal becomes too small for measurement.
Our bounds make this tradeoff explicit and deterministic.
\end{remark}

\subsection{Optional: extracting two frequencies from a two-exponential scalar signal}\label{subsec:two-exp}

The preceding analysis is the most efficient route for our Kerr--de~Sitter application because the PDE localizations isolate each mode separately.
For completeness, we record a deterministic stability statement for the \emph{two-exponential} model, which may be useful if one only observes the combined signal.
This is a very small instance of Prony-type reconstruction; see \cite{PottsTascheLAA2013,HuaSarkarMPM1990,BatenkovYomdinSJAM2013} for broader frameworks
and sharp discussions of conditioning.

Consider
\begin{equation}\label{eq:two-exp-model}
y(t)=a_1 e^{-i\omega_1 t}+a_2 e^{-i\omega_2 t}+\rho(t),
\qquad t\in[T_0,T_0+T],
\end{equation}
with $\omega_1\neq\omega_2$ and $a_1a_2\neq 0$.
Fix $\Delta>0$ with $2\Delta<T$ and define discrete samples
\[
y_j:=y(T_0+j\Delta),\qquad j=0,1,2,3.
\]
In the noiseless case $\rho\equiv 0$, the sequence $(y_j)$ satisfies the order-$2$ recurrence
\begin{equation}\label{eq:recurrence}
y_{j+2}=s_1 y_{j+1}-s_2 y_j,
\qquad s_1=z_1+z_2,\quad s_2=z_1 z_2,
\quad z_k:=e^{-i\omega_k\Delta}.
\end{equation}
Solving for $(s_1,s_2)$ from $j=0,1$ yields $z_1,z_2$ as the roots of $\lambda^2-s_1\lambda+s_2=0$.
The next lemma quantifies stability under deterministic perturbations of the four samples.

\begin{lemma}[Four-sample Prony stability (local)]\label{lem:prony-stability}
Assume $\rho\equiv 0$ in \eqref{eq:two-exp-model} and let $y_j$ be the noiseless samples
$y_j=a_1 z_1^j+a_2 z_2^j$ with $a_1a_2\neq 0$ and $z_1\neq z_2$.
Set
\[
\Delta_0:=y_0 y_2 - y_1^2 = a_1 a_2 (z_1-z_2)^2\neq 0,
\qquad
R:=\max\{|z_1|,|z_2|\}.
\]
Suppose perturbed samples $\widetilde y_j=y_j+e_j$ are given, with $|e_j|\le \eta$.
Let $(\widetilde s_1,\widetilde s_2)$ be obtained by solving the perturbed linear system
\[
\begin{pmatrix}\widetilde y_1 & -\widetilde y_0\\ \widetilde y_2 & -\widetilde y_1\end{pmatrix}
\binom{\widetilde s_1}{\widetilde s_2}
=
\binom{\widetilde y_2}{\widetilde y_3},
\]
and let $\widetilde z_1,\widetilde z_2$ be the roots of $\lambda^2-\widetilde s_1\lambda+\widetilde s_2=0$.
Assume in addition that the noise level is small in the sense
\begin{equation}\label{eq:prony-smallness}
\eta\ \le\ c_0\,|a_1a_2|\,|z_1-z_2|^4,
\end{equation}
for a sufficiently small absolute constant $c_0>0$.
Then, after labeling $(\widetilde z_1,\widetilde z_2)$ to match $(z_1,z_2)$, one has
\begin{equation}\label{eq:prony-z}
|\widetilde z_k-z_k|
\ \le\
C(R,a_1,a_2)\,\frac{\eta}{|a_1 a_2|\,|z_1-z_2|^3},
\qquad k=1,2.
\end{equation}
Consequently, choosing a logarithm branch as in Lemma~\ref{lem:branch-selection}, one obtains
\begin{equation}\label{eq:prony-omega}
|\widetilde\omega_k-\omega_k|
\ \le\
\frac{C(R,a_1,a_2)}{\Delta\,|z_k|}\,\frac{\eta}{|a_1 a_2|\,|z_1-z_2|^3}.
\end{equation}
\end{lemma}

\begin{proof}
\emph{Step 1: stability of the recurrence coefficients.}
For the noiseless data, the linear system for $(s_1,s_2)$ has determinant $\Delta_0=y_0y_2-y_1^2\neq 0$, and Cramer's rule gives
\begin{equation}\label{eq:s12-explicit}
s_1=\frac{y_0y_3-y_1y_2}{\Delta_0},\qquad s_2=\frac{y_1y_3-y_2^2}{\Delta_0}.
\end{equation}
For the perturbed data, the matrix determinant is $\widetilde\Delta_0:=\widetilde y_0\widetilde y_2-\widetilde y_1^2$ and
\[
\widetilde s_1=\frac{\widetilde y_0\widetilde y_3-\widetilde y_1\widetilde y_2}{\widetilde\Delta_0},
\qquad
\widetilde s_2=\frac{\widetilde y_1\widetilde y_3-\widetilde y_2^2}{\widetilde\Delta_0}.
\]
Write $\delta s_\ell:=\widetilde s_\ell-s_\ell$ and $\delta\Delta_0:=\widetilde\Delta_0-\Delta_0$.
A direct expansion yields
\begin{equation}\label{eq:delta-det-bound}
|\delta\Delta_0|
\le
|e_0|\,|y_2|+|y_0|\,|e_2|+2|y_1|\,|e_1|+|e_0|\,|e_2|+|e_1|^2
\le
\eta\big(|y_0|+2|y_1|+|y_2|\big)+2\eta^2.
\end{equation}
Similarly, expanding the numerators shows
\begin{equation}\label{eq:delta-num-bound}
\big|(\widetilde y_0\widetilde y_3-\widetilde y_1\widetilde y_2)-(y_0y_3-y_1y_2)\big|
+\big|(\widetilde y_1\widetilde y_3-\widetilde y_2^2)-(y_1y_3-y_2^2)\big|
\le
C_y\,\eta,
\end{equation}
where $C_y$ depends only on $|y_0|,\dots,|y_3|$.
Since $y_j=a_1 z_1^j+a_2 z_2^j$, we have the crude bound
\[
|y_j|\le |a_1|\,|z_1|^j+|a_2|\,|z_2|^j \le (|a_1|+|a_2|)\,R^j
\le (|a_1|+|a_2|)\,\max\{1,R^3\},\qquad j=0,1,2,3.
\]
(In the ringdown regime $\Im\omega_k\le 0$, one has $|z_k|\le 1$ and hence $R\le 1$, so the right-hand side reduces to $|a_1|+|a_2|$.)
Thus $C_y$ can be bounded in terms of $\max\{1,R^3\}$ and $|a_1|+|a_2|$.

Under the smallness assumption \eqref{eq:prony-smallness} (with $c_0$ chosen sufficiently small depending on $R,a_1,a_2$ through $C_y$),
\eqref{eq:delta-det-bound} implies $|\delta\Delta_0|\le \frac12|\Delta_0|$, hence
\begin{equation}\label{eq:det-lower}
|\widetilde\Delta_0|\ \ge\ \frac12\,|\Delta_0|.
\end{equation}
Combining \eqref{eq:det-lower} with \eqref{eq:delta-num-bound} yields
\begin{equation}\label{eq:ds12}
|\delta s_1|+|\delta s_2|
\ \le\
C_1(R,a_1,a_2)\,\frac{\eta}{|\Delta_0|}
=
C_1(R,a_1,a_2)\,\frac{\eta}{|a_1a_2|\,|z_1-z_2|^2}.
\end{equation}

\emph{Step 2: stability of the roots.}
Let $p(\lambda)=\lambda^2-s_1\lambda+s_2=(\lambda-z_1)(\lambda-z_2)$ and
$\widetilde p(\lambda)=\lambda^2-\widetilde s_1\lambda+\widetilde s_2$.
Set $\delta p(\lambda):=\widetilde p(\lambda)-p(\lambda)=-\delta s_1\,\lambda+\delta s_2$.

Choose $r:=|z_1-z_2|/4$.
By \eqref{eq:ds12} and the smallness assumption \eqref{eq:prony-smallness}, we may ensure (by shrinking $c_0$ further if needed) that
$|\delta s_1|+|\delta s_2|\le c\,|z_1-z_2|^2$ with $c>0$ small enough so that
\begin{equation}\label{eq:rouche-cond}
|\delta p(\lambda)|<|p(\lambda)|
\qquad\text{for all }\lambda\text{ with }|\lambda-z_k|=r,\ k=1,2.
\end{equation}
Indeed, on $|\lambda-z_1|=r$ we have $|\lambda-z_2|\ge |z_1-z_2|-r=3|z_1-z_2|/4$, hence
\[
|p(\lambda)|=|\lambda-z_1|\,|\lambda-z_2|
\ge r\cdot \frac34|z_1-z_2|
=\frac{3}{16}|z_1-z_2|^2,
\]
while $|\delta p(\lambda)|\le |\delta s_1|\,|\lambda|+|\delta s_2|$ and $|\lambda|\le |z_1|+r\le |z_1|+|z_1-z_2|/4$.
The corresponding estimate for the circle around $z_2$ is identical.

By Rouch\'e's theorem, \eqref{eq:rouche-cond} implies that $\widetilde p$ has exactly one zero in each disk
$D_k:=\{|\lambda-z_k|<r\}$.
Label these zeros as $\widetilde z_k\in D_k$.

Now use the identity $\widetilde p(\widetilde z_k)=0$, i.e.\ $p(\widetilde z_k)=\delta s_1\,\widetilde z_k-\delta s_2$.
Since $p(\widetilde z_k)=(\widetilde z_k-z_k)(\widetilde z_k-z_{3-k})$ and $\widetilde z_k\in D_k$, we have
$|\widetilde z_k-z_{3-k}|\ge |z_1-z_2|-r=3|z_1-z_2|/4$.
Therefore,
\[
\begin{aligned}
|\widetilde z_k-z_k|
&\le
\frac{|\delta s_1|\,|\widetilde z_k|+|\delta s_2|}{|\widetilde z_k-z_{3-k}|}\\
&\le
\frac{4}{3|z_1-z_2|}\Bigl(|\delta s_1|(|z_k|+r)+|\delta s_2|\Bigr)\\
&\le
C_2(R)\,\frac{|\delta s_1|+|\delta s_2|}{|z_1-z_2|}.
\end{aligned}
\]
Combining this with \eqref{eq:ds12} yields \eqref{eq:prony-z}.

Finally, \eqref{eq:prony-omega} follows from Lemma~\ref{lem:log-Lip} (or Lemma~\ref{lem:branch-selection} if a prior is available)
applied to $z_k$ and $\widetilde z_k$.
\end{proof}

\begin{remark}[Two-exponential extraction and the main pipeline]\label{rem:why-not-two}
Lemma~\ref{lem:prony-stability} shows that recovering both nodes from only four consecutive samples suffers a severe separation-dependent loss:
even in the local regime where the reconstruction map is Lipschitz, the error constant scales at least like
$|a_1a_2|^{-1}|z_1-z_2|^{-3}$.
The factor $|z_1-z_2|^{-2}$ comes from inverting the $2\times 2$ Hankel matrix (whose determinant is $a_1a_2(z_1-z_2)^2$),
and an additional $|z_1-z_2|^{-1}$ comes from converting perturbed symmetric functions into perturbed roots.
In Kerr--de~Sitter, the equatorial splitting satisfies $|\Re\omega_{+,\ell}-\Re\omega_{-,\ell}|\sim |a|$,
so $|z_1-z_2|\sim \Delta |a|$ for small $\Delta$ and the constant can blow up like $(\Delta|a|)^{-3}$ when $a$ is small.
The PDE localization of Section~\ref{sec:two-mode} avoids this loss by separating the two modes \emph{before} extraction,
reducing the problem to two stable one-mode extractions.
\end{remark}

\subsection{Output for the inverse problem}\label{subsec:output-inverse}

Combining Section~\ref{sec:two-mode} with the one-mode stability Theorem~\ref{thm:one-mode-extraction} yields:
\begin{itemize}
\item a deterministic procedure to estimate $\omega_{n,\ell,\pm\ell}(M,a)$ from a time-domain signal after
equatorial microlocalization and azimuthal selection;
\item an explicit bound for the frequency error in terms of the tail size (controlled by the shifted-contour resolvent)
and any additional measurement perturbation.
\end{itemize}
In the next section we feed these frequency error bounds into the parameter recovery stability estimate
from the first paper to obtain an explicit parameter-bias bound for ringdown-based inversion.

\section{Application: parameter bias bound for ringdown-based inversion}\label{sec:application}

This section closes the deterministic pipeline
\[
\text{time-domain ringdown data}\ \Longrightarrow\ \text{QNM frequencies}\ \Longrightarrow\ (M,a),
\]
and produces an explicit \emph{parameter bias bound} in terms of:
(i) the PDE tail remainder on a shifted contour (Sections~\ref{sec:resolvent-shifted}--\ref{sec:two-mode}),
(ii) deterministic measurement perturbations in the time series, and
(iii) the conditioning of the frequency-extraction functional (Section~\ref{sec:freq-extraction}).
The estimates below are entirely deterministic and keep track of the parameters that govern conditioning:
start time $T_0$, window length $T$, shift step $\Delta$, and (crucially) nondegeneracy of the chosen detector for the extracted modes.

\subsection{The inversion map from equatorial QNMs (input from the companion paper)}\label{subsec:inv-map}

Fix $\Lambda>0$, a compact parameter set
\[
\mathcal K \Subset \mathcal P_\Lambda\cap\{|a|\le a_0\},
\]
with $a_0>0$ sufficiently small, and fix an overtone index
\begin{equation}\label{eq:overtone-range-app}
n\in\mathbb N_0.
\end{equation}
For $\ell\gg1$ let
\[
\omega_{\pm,\ell}(M,a):=\omega_{n,\ell,\pm\ell}(M,a)
\]
denote the labeled simple equatorial QNMs (Section~\ref{sec:two-mode}).
Define the normalized real observables
\begin{equation}\label{eq:UV-def-app}
\begin{aligned}
U_\ell(M,a)
&:=\Re\frac{\omega_{+,\ell}(M,a)+\omega_{-,\ell}(M,a)}{2\ell},\\
V_\ell(M,a)
&:=\Re\frac{\omega_{+,\ell}(M,a)-\omega_{-,\ell}(M,a)}{2\ell}.
\end{aligned}
\end{equation}
and the \emph{equatorial QNM data map}
\begin{equation}\label{eq:G-def-app}
\mathcal G_\ell(M,a):=(U_\ell(M,a),\,V_\ell(M,a))\in\mathbb R^2.
\end{equation}

The following theorem is the frequency-to-parameter stability input proved in the companion paper
(we restate it for convenience).

\begin{theorem}[Local inversion from equatorial QNMs; stability]\label{thm:inv-companion}
There exist $\ell_0\in\mathbb N$ and constants $c_*,C_*>0$ such that for every $\ell\ge\ell_0$:
\begin{enumerate}
\item $\mathcal G_\ell:\mathcal K\to\mathbb R^2$ is real-analytic and locally invertible at every point of $\mathcal K$,
with Jacobian uniformly bounded away from $0$:
\begin{equation}\label{eq:jac-lower-app}
\inf_{(M,a)\in\mathcal K}|\det D\mathcal G_\ell(M,a)|\ \ge\ c_*.
\end{equation}
\item (Quantitative local stability.) For all $(M,a),(M',a')\in\mathcal K$ sufficiently close,
\begin{equation}\label{eq:stab-G}
|(M,a)-(M',a')|
\ \le\
C_*\,|\mathcal G_\ell(M,a)-\mathcal G_\ell(M',a')|.
\end{equation}
\end{enumerate}
Moreover, in terms of the frequencies,
\begin{equation}\label{eq:stab-omega-app}
|(M,a)-(M',a')|
\ \le\
\frac{C_*}{\ell}\Big(|\Re(\omega_{\mathrm{av}}-\omega_{\mathrm{av}}')|+|\Re(\omega_{\mathrm{dif}}-\omega_{\mathrm{dif}}')|\Big),
\end{equation}
where $\omega_{\mathrm{av}}=\frac12(\omega_{+,\ell}+\omega_{-,\ell})$ and $\omega_{\mathrm{dif}}=\frac12(\omega_{+,\ell}-\omega_{-,\ell})$.
\end{theorem}

\begin{remark}[A condition number viewpoint]\label{rem:cond-number}
It is useful to interpret \eqref{eq:stab-G}--\eqref{eq:stab-omega-app} as an inverse-problem condition number statement.
Define the \emph{raw-frequency map}
\[
\mathcal F_\ell(M,a):=\big(\Re\omega_{\mathrm{av}}(M,a),\,\Re\omega_{\mathrm{dif}}(M,a)\big)\in\mathbb R^2.
\]
Then $\mathcal G_\ell=\mathcal F_\ell/\ell$ by \eqref{eq:UV-def-app}. Differentiating gives
$D\mathcal F_\ell=\ell\,D\mathcal G_\ell$ and consequently
\[
\big\|(D\mathcal F_\ell)^{-1}\big\|
\ =\ \frac{1}{\ell}\,\big\|(D\mathcal G_\ell)^{-1}\big\|.
\]
Thus the local inverse from \emph{unscaled} frequencies carries an additional factor $1/\ell$ compared to the inverse of $\mathcal G_\ell$.
In particular, the constant $C_*/\ell$ in \eqref{eq:stab-omega-app} is precisely the (uniform) Lipschitz constant of the local inverse
when the data are taken to be $(\Re\omega_{\mathrm{av}},\Re\omega_{\mathrm{dif}})$.
This is the source of the high-frequency conditioning gain in Theorem~\ref{thm:param-bias} below.
\end{remark}

\subsection{Time-domain observation model and detectability of the separated modes}\label{subsec:obs-model}

Fix $\ell\ge\ell_1$ (as in Theorem~\ref{thm:two-mode-dominance}) and let $(M,a)\in\mathcal K$ be the true parameters.
Let $u$ be the (scalar) wave solution with initial data $(f_0,f_1)\in\mathcal H^{s+N}$.

\medskip

\noindent\textbf{Equatorial mode-separated signals.}
We work with the refined equatorial windowed signals $u^{(\pm,\ell)}(t)$ constructed in
\eqref{eq:refined-window}:
\[
u^{(+,\ell)}(t)=\chi A_{+,h_\ell}\Pi_{+\ell}\,\mathcal W^{(+,\ell)} u(t),
\qquad
u^{(-,\ell)}(t)=\chi A_{-,h_\ell}\Pi_{-\ell}\,\mathcal W^{(-,\ell)} u(t),
\]
where $\mathcal W^{(\pm,\ell)}$ is the inverse Laplace integral with the modified analytic weight $\widetilde g_{\pm,\ell}$.
By Theorem~\ref{thm:two-mode-dominance}, for $t\ge 1$,
\begin{equation}\label{eq:u-mode-separated}
\begin{aligned}
u^{(\pm,\ell)}(t)
&=e^{-i\omega_{\pm,\ell}t}\,A_{\pm,\ell} + E_{\pm,\ell}(t),\\
\|E_{\pm,\ell}(t)\|_{H^{s+1}}
&\le C_{s,m,N'}\Big(e^{-\nu t}(1+t)^{-m}+\ell^{-N'}\Big)\,
\|(f_0,f_1)\|_{\mathcal H^{s+m+N_{m,N'}}}.
\end{aligned}
\end{equation}

\begin{remark}[Choice of the leakage exponent]\label{rem:leakage}
The term $\ell^{-N'}$ in \eqref{eq:u-mode-separated} encodes semiclassically small leakage from other poles in the same
equatorial package; it is uniform on fixed compact parameter sets.
For the later stability estimates we fix a large $N'$ once and for all, and take $m$ as large as needed.
\end{remark}

\medskip

\noindent\textbf{Scalar detector and ideal signals.}
Let $\mathsf O:H^{s+1}(X_{M,a})\to\mathbb C$ be a bounded linear functional (Section~\ref{subsec:obs}).
Define the \emph{ideal} mode-separated scalar signals
\[
y_{\pm}(t):=\mathsf O\big(u^{(\pm,\ell)}(t)\big).
\]
Then by \eqref{eq:u-mode-separated},
\begin{equation}\label{eq:y-model}
y_{\pm}(t)=a_{\pm,\ell}\,e^{-i\omega_{\pm,\ell}t}+\rho_{\pm,\ell}(t),
\qquad
a_{\pm,\ell}:=\mathsf O(A_{\pm,\ell}),\qquad
\rho_{\pm,\ell}(t):=\mathsf O(E_{\pm,\ell}(t)).
\end{equation}
The tail terms satisfy, for every $m$,
\begin{equation}\label{eq:rho-tail-bound}
\begin{aligned}
|\rho_{\pm,\ell}(t)|
&\le \|\mathsf O\|\,\|E_{\pm,\ell}(t)\|_{H^{s+1}}\\
&\le C_{s,m,N'}\,\|\mathsf O\|\Big(e^{-\nu t}(1+t)^{-m}+\ell^{-N'}\Big)\,
\|(f_0,f_1)\|_{\mathcal H^{s+m+N_{m,N'}}},\\
&\hspace{1.2cm} t\ge 1.
\end{aligned}
\end{equation}

\begin{assumption}[Mode detectability on the chosen window]\label{ass:detectability-app}
On the observation window under consideration, we assume that the detector does not annihilate the extracted dominant components:
\begin{equation}\label{eq:detectability}
a_{+,\ell}\neq 0,\qquad a_{-,\ell}\neq 0.
\end{equation}
Equivalently, for any admissible weight $w$ that is not a.e.\ zero on $[T_0,T_0+T-\Delta]$, the reference signals $y_{0,\pm}(t):=a_{\pm,\ell}e^{-i\omega_{\pm,\ell}t}$ satisfy $\|y_{0,\pm}\|_w>0$.
All frequency and parameter estimates below are conditional on \eqref{eq:detectability} and on a quantitative signal-to-residual smallness
condition, expressed through the relative residual sizes $\varepsilon_\pm$ in \eqref{eq:eps-def}.
\end{assumption}

\begin{remark}[Dual-state interpretation of detectability]\label{rem:dual-state-detect}
Assumption~\ref{ass:detectability-app} admits a clean operator-theoretic interpretation in terms of \emph{left} and \emph{right}
resonant states.
For a simple pole $\omega=\omega_{0}$ of the meromorphic family $R(\omega)=P(\omega)^{-1}$, Appendix~\ref{app:dual-states}
(recall also the general framework of Keldysh expansions \cite{GohbergSigal1971,KatoPerturbation})
shows that the residue projector has rank one and can be written, after a choice of normalization, as
\begin{equation}\label{eq:rank-one-residue}
\Pi_{\omega_0} f=\frac{\langle f,v_{\omega_0}\rangle}{\langle \partial_\omega P(\omega_0)u_{\omega_0},v_{\omega_0}\rangle}\,u_{\omega_0},
\end{equation}
where $u_{\omega_0}$ is a (right) resonant state and $v_{\omega_0}$ is the corresponding dual (left) resonant state.
Consequently, for the equatorial poles $\omega_{\pm,\ell}$ appearing in Theorem~\ref{thm:two-mode-dominance},
the residue contributions take the form
\[
A_{\pm,\ell}=\chi\,A_{\pm,h_\ell}\Pi_{k_\pm}\Pi_{\omega_{\pm,\ell}}\widehat F_\vartheta(\omega_{\pm,\ell})
=\Bigl(\frac{\langle \widehat F_\vartheta(\omega_{\pm,\ell}),v_{\pm,\ell}\rangle}{\langle \partial_\omega P(\omega_{\pm,\ell})u_{\pm,\ell},v_{\pm,\ell}\rangle}\Bigr)\,
\chi\,A_{\pm,h_\ell}\Pi_{k_\pm}u_{\pm,\ell}.
\]
Applying the observation functional $\mathsf O$ therefore yields
\begin{equation}\label{eq:amplitude-dual-pairing}
\begin{split}
a_{\pm,\ell}
&=\frac{\langle \widehat F_\vartheta(\omega_{\pm,\ell}),v_{\pm,\ell}\rangle}
{\langle \partial_\omega P(\omega_{\pm,\ell})u_{\pm,\ell},v_{\pm,\ell}\rangle}\\
&\qquad\times \mathsf O\!\left(\chi\,A_{\pm,h_\ell}\Pi_{k_\pm}u_{\pm,\ell}\right).
\end{split}
\end{equation}
The \emph{detectability} condition $a_{\pm,\ell}\neq 0$ thus fails precisely when either the chosen observable annihilates the
projected right resonant state, or when the data are orthogonal (in the dual pairing) to the left resonant state.
In particular, for fixed $(M,a)$ and $\ell$, the set of non-detectable data is a codimension-one subspace, and nonnormality
enters through the normalization denominator in \eqref{eq:amplitude-dual-pairing}, which is the standard excitation-factor
mechanism.
\end{remark}

\begin{lemma}[Generic initial data do not annihilate a fixed extracted mode]\label{lem:generic-data-app}
Fix $(M,a)\in\mathcal K$, $\ell\ge\ell_1$, and a bounded detector $\mathsf O:H^{s+1}\to\mathbb C$.
For each sign $\pm$, the map
\[
\mathcal H^{s+N}\ni (f_0,f_1)\longmapsto a_{\pm,\ell}=\mathsf O\big(A_{\pm,\ell}(f_0,f_1)\big)\in\mathbb C
\]
is a continuous linear functional. If it is not identically zero, then its kernel is a proper closed hyperplane, and its complement is open and dense.
\end{lemma}

\begin{proof}
By \eqref{eq:u-mode-separated}, $(f_0,f_1)\mapsto A_{\pm,\ell}(f_0,f_1)\in H^{s+1}$ is linear.
Its continuity in the $\mathcal H^{s+N}$ topology follows from the construction in Theorem~\ref{thm:two-mode-dominance}:
$A_{\pm,\ell}$ is obtained by composing bounded operators on the data with the finite-rank residue projector at the simple pole $\omega_{\pm,\ell}$.
Composing with the bounded functional $\mathsf O$ gives a continuous linear functional on $\mathcal H^{s+N}$.
If it is not identically zero, its kernel is a proper closed hyperplane, hence its complement is open and dense.
\end{proof}

\begin{proposition}[Generic detectability for a fixed mode and a fixed detector]\label{prop:generic-detectability}
Fix $(M,a)\in\mathcal K$ and $\ell\ge \ell_1$, and let $\mathsf O:H^{s+1}\to\mathbb C$ be a bounded detector.
For each sign $\pm$, exactly one of the following alternatives holds:
\begin{enumerate}
\item[(a)] The detector is \emph{blind} to the extracted $\pm$ mode in the sense that
$a_{\pm,\ell}(f_0,f_1)=0$ for all data $(f_0,f_1)\in\mathcal H^{s+N}$.
\item[(b)] The set of data $(f_0,f_1)\in\mathcal H^{s+N}$ for which $a_{\pm,\ell}(f_0,f_1)\neq 0$ is open and dense.
\end{enumerate}
If both signs fall under alternative \textup{(b)}, then the joint detectability set
\[
\big\{(f_0,f_1)\in\mathcal H^{s+N}: a_{+,\ell}(f_0,f_1)\neq 0\ \text{and}\ a_{-,\ell}(f_0,f_1)\neq 0\big\}
\]
is open and dense.
\end{proposition}

\begin{proof}
For each sign $\pm$, Lemma~\ref{lem:generic-data-app} shows that the amplitude map
$(f_0,f_1)\mapsto a_{\pm,\ell}(f_0,f_1)$ is a continuous linear functional on $\mathcal H^{s+N}$.
If it is identically zero, we are in alternative \textup{(a)}. Otherwise its kernel is a proper closed hyperplane and its complement is open and dense,
which is alternative \textup{(b)}.
If both amplitude functionals are nontrivial, then the intersection of the two open dense complements is again open and dense.
\end{proof}

\begin{remark}[Characterization of blind detectors]\label{rem:blind-detector}
In the rank-one setting \eqref{eq:rank-one-residue}--\eqref{eq:amplitude-dual-pairing}, the range of the residue contribution $A_{\pm,\ell}$
is contained in the one-dimensional space spanned by $\chi A_{\pm,h_\ell}\Pi_{k_\pm}u_{\pm,\ell}$.
Consequently, if $\chi A_{\pm,h_\ell}\Pi_{k_\pm}u_{\pm,\ell}\neq 0$, then alternative \textup{(a)} in
Proposition~\ref{prop:generic-detectability} holds if and only if
\[
\mathsf O\!\left(\chi A_{\pm,h_\ell}\Pi_{k_\pm}u_{\pm,\ell}\right)=0,
\]
a codimension-one condition on $\mathsf O$ in the dual space of $H^{s+1}$.
In particular, blindness is nongeneric among bounded detectors, and can be mitigated in practice by using multiple independent channels,
as described below.
\end{remark}

\begin{remark}[Avoiding accidental annihilation]\label{rem:multi-detector-app}
Assumption~\ref{ass:detectability-app} is intrinsic: if a detector annihilates a mode, no deterministic method can recover its frequency from that channel.
If one has access to multiple bounded detectors $\mathsf O_1,\dots,\mathsf O_J$, one may run the entire extraction procedure in each channel and select
the channel(s) with the largest observed dominant energy $\|\widetilde y_{j,\pm}\|_w$ on the chosen window.
This reduces the practical risk of accidental annihilation and is compatible with the deterministic framework developed here.
\end{remark}

\medskip

\noindent\textbf{Measured signals and measurement noise.}
We allow an additional deterministic measurement perturbation $\eta_\pm(t)$ on the observation window
$[T_0,T_0+T]$:
\begin{equation}\label{eq:measured-y}
\widetilde y_\pm(t)=y_\pm(t)+\eta_\pm(t)
=
a_{\pm,\ell}\,e^{-i\omega_{\pm,\ell}t}+\underbrace{\big(\rho_{\pm,\ell}(t)+\eta_\pm(t)\big)}_{=:r_\pm(t)},
\qquad t\in[T_0,T_0+T],
\end{equation}
where $T_0\ge 1$ and $T>0$.
We measure errors in the weighted norm $\|\cdot\|_w$ of Section~\ref{subsec:one-mode-estimator},
with a fixed shift $\Delta\in(0,T)$ and a weight $w$ satisfying \eqref{eq:weight}.
When we use the explicit lower bounds from Lemma~\ref{lem:mode-energy}, we additionally assume $T>3\Delta$.

\subsection{Frequency estimation from each separated mode}\label{subsec:freq-estimation-app}

Apply the shift Rayleigh quotient estimator \eqref{eq:z-hat} (Section~\ref{subsec:one-mode-estimator})
to each $\widetilde y_\pm$ on the window $[T_0,T_0+T]$ with step $\Delta$ and weight $w$:
\[
\widehat z_\pm := \frac{\langle S_\Delta \widetilde y_\pm,\ \widetilde y_\pm\rangle_w}{\langle \widetilde y_\pm,\ \widetilde y_\pm\rangle_w},
\qquad
\widehat\omega_\pm := \frac{i}{\Delta}\,\Log(\widehat z_\pm),
\]
where $\Log$ is chosen consistently with the prior $\omega^\sharp_{\pm,\ell}$ as in Lemma~\ref{lem:branch-selection}
(or equivalently Remark~\ref{rem:branch}).

Define the relative residual sizes
\begin{equation}\label{eq:eps-def}
\varepsilon_\pm:=\frac{\max\{\|r_\pm\|_w,\ \|S_\Delta r_\pm\|_w\}}{\|a_{\pm,\ell}e^{-i\omega_{\pm,\ell}t}\|_w}
=
\frac{\max\{\|\rho_{\pm,\ell}+\eta_\pm\|_w,\ \|S_\Delta(\rho_{\pm,\ell}+\eta_\pm)\|_w\}}{\|a_{\pm,\ell}e^{-i\omega_{\pm,\ell}t}\|_w}.
\end{equation}

\begin{proposition}[Deterministic frequency error from tail+measurement]\label{prop:freq-error-app}
Assume \eqref{eq:detectability} and $\Im\omega_{\pm,\ell}\le 0$.
If $\varepsilon_\pm\le \min\{1/8,\ |z_\pm|/20\}$ where $z_\pm=e^{-i\omega_{\pm,\ell}\Delta}$, then
\begin{equation}\label{eq:freq-error-app}
|\widehat\omega_\pm-\omega_{\pm,\ell}|
\ \le\
\frac{10}{\Delta\,|z_\pm|}\,\varepsilon_\pm.
\end{equation}
In particular,
\begin{equation}\label{eq:realpart-error}
|\Re(\widehat\omega_\pm-\omega_{\pm,\ell})|\ \le\ |\widehat\omega_\pm-\omega_{\pm,\ell}|.
\end{equation}
\end{proposition}

\begin{proof}
This is Theorem~\ref{thm:one-mode-extraction} applied to $\widetilde y_\pm$ with residual $r_\pm$.
The definition \eqref{eq:eps-def} matches the relative error quantity $\varepsilon$ in Theorem~\ref{thm:one-mode-extraction}
(with $\rho$ replaced by $r_\pm$).
\end{proof}

\begin{corollary}[A sufficient condition for staying in the labeled frequency disk]\label{cor:freq-disk}
Assume Theorem~\ref{thm:input-pole-isolation} and fix the overtone index $n$ used throughout Section~\ref{sec:two-mode}.
Let $c_{\mathrm{sep}}>0$ be the disk radius in \eqref{eq:disks} and write $\omega_{\pm,\ell}:=\omega_{n,\pm,\ell}$ and
$\omega^\sharp_{\pm,\ell}:=\omega^\sharp_{n,\pm,\ell}$.
Choose $\ell$ sufficiently large so that $|\omega_{\pm,\ell}-\omega^\sharp_{\pm,\ell}|\le c_{\mathrm{sep}}/4$ (which holds by
\eqref{eq:pole-pseudopole-closeness}).
If, in addition to the hypotheses of Proposition~\ref{prop:freq-error-app}, the residual size satisfies
\begin{equation}\label{eq:eps-disk-sufficient}
\varepsilon_\pm\ \le\ \frac{c_{\mathrm{sep}}}{40}\,\Delta\,|z_\pm|,
\qquad z_\pm=e^{-i\omega_{\pm,\ell}\Delta},
\end{equation}
then the extracted frequency obeys $\widehat\omega_\pm\in D_{n,\pm,\ell}$.
\end{corollary}

\begin{proof}
By Proposition~\ref{prop:freq-error-app} and \eqref{eq:eps-disk-sufficient},
\[
|\widehat\omega_\pm-\omega_{\pm,\ell}|
\le
\frac{10}{\Delta\,|z_\pm|}\,\varepsilon_\pm
\le \frac{c_{\mathrm{sep}}}{4}.
\]
Together with $|\omega_{\pm,\ell}-\omega^\sharp_{\pm,\ell}|\le c_{\mathrm{sep}}/4$, this gives
$|\widehat\omega_\pm-\omega^\sharp_{\pm,\ell}|\le c_{\mathrm{sep}}/2$, hence $\widehat\omega_\pm\in D_{n,\pm,\ell}$ by \eqref{eq:disks}.
\end{proof}

\begin{corollary}[Explicit control of $\varepsilon_\pm$ and the frequency error]\label{cor:freq-error-explicit}
Assume Assumption~\ref{ass:detectability-app} and fix window parameters $T_0\ge 1$, $T>0$, $\Delta\in(0,T)$ and a weight $w$ as in \eqref{eq:weight}.
Let $y_{0,\pm}(t):=a_{\pm,\ell}e^{-i\omega_{\pm,\ell}t}$ and write the measured signals as in \eqref{eq:measured-y},
$\widetilde y_\pm=y_{0,\pm}+r_\pm$ with $r_\pm=\rho_{\pm,\ell}+\eta_\pm$.
Then
\begin{equation}\label{eq:eps-explicit}
\varepsilon_\pm
\ \le\
\frac{\|r_\pm\|_{L^2([T_0,T_0+T])}}{\|y_{0,\pm}\|_w}
\ \le\
\frac{\|\rho_{\pm,\ell}\|_{L^2([T_0,T_0+T])}+\|\eta_\pm\|_{L^2([T_0,T_0+T])}}{\|y_{0,\pm}\|_w}.
\end{equation}
Moreover, integrating the pointwise tail bound \eqref{eq:rho-tail-bound} yields, for every $m\in\mathbb N$ and $N'\in\mathbb N$,
\begin{equation}\label{eq:rho-L2-app}
\|\rho_{\pm,\ell}\|_{L^2([T_0,T_0+T])}
\le
C_{s,m,N'}\,\|\mathsf O\|\Big(e^{-\nu T_0}(1+T_0)^{-m}+\ell^{-N'}\Big)\,\sqrt{T}\,
\|(f_0,f_1)\|_{\mathcal H^{s+m+N_{m,N'}}}.
\end{equation}
In particular, if $\varepsilon_\pm\le \min\{1/8,|z_\pm|/20\}$ (with $z_\pm=e^{-i\omega_{\pm,\ell}\Delta}$), then
\begin{equation}\label{eq:freq-error-explicit}
|\widehat\omega_\pm-\omega_{\pm,\ell}|
\ \le\
\frac{10}{\Delta\,|z_\pm|}\,
\frac{\|\rho_{\pm,\ell}\|_{L^2([T_0,T_0+T])}+\|\eta_\pm\|_{L^2([T_0,T_0+T])}}{\|y_{0,\pm}\|_w}.
\end{equation}
If, in addition, $T>3\Delta$, then Lemma~\ref{lem:mode-energy} provides the lower bound
\[
\|y_{0,\pm}\|_w
\ \ge\
|a_{\pm,\ell}|\left(\int_{T_0+\Delta}^{T_0+T-2\Delta} e^{2\Im\omega_{\pm,\ell} t}\,dt\right)^{1/2},
\]
which makes \eqref{eq:eps-explicit}--\eqref{eq:freq-error-explicit} fully explicit in terms of the amplitude $|a_{\pm,\ell}|$.
\end{corollary}

\begin{proof}
The first inequality in \eqref{eq:eps-explicit} is Lemma~\ref{lem:shifted-residual} applied to $r_\pm$.
The second inequality is the triangle inequality $\|r_\pm\|_{L^2}\le \|\rho_{\pm,\ell}\|_{L^2}+\|\eta_\pm\|_{L^2}$.
The tail estimate \eqref{eq:rho-L2-app} follows by integrating \eqref{eq:rho-tail-bound} over the window.
Finally, \eqref{eq:freq-error-explicit} is the frequency bound \eqref{eq:freq-error-app} with $\varepsilon_\pm$ controlled by \eqref{eq:eps-explicit}.
\end{proof}

\begin{lemma}[Separating tail and measurement contributions]\label{lem:eps-split}
Let $y_{0,\pm}(t):=a_{\pm,\ell}e^{-i\omega_{\pm,\ell}t}$ and $r_\pm=\rho_{\pm,\ell}+\eta_\pm$.
Then
\begin{equation}\label{eq:eps-split}
\varepsilon_\pm
\ \le\
\varepsilon_\pm^{\mathrm{tail}}+\varepsilon_\pm^{\mathrm{meas}},
\end{equation}
where
\begin{equation}\label{eq:eps-tail-meas}
\varepsilon_\pm^{\mathrm{tail}}:=\frac{\max\{\|\rho_{\pm,\ell}\|_w,\ \|S_\Delta\rho_{\pm,\ell}\|_w\}}{\|y_{0,\pm}\|_w},
\qquad
\varepsilon_\pm^{\mathrm{meas}}:=\frac{\max\{\|\eta_\pm\|_w,\ \|S_\Delta\eta_\pm\|_w\}}{\|y_{0,\pm}\|_w}.
\end{equation}
Moreover, since $0\le w\le 1$,
\begin{equation}\label{eq:eps-L2}
\max\{\|f\|_w,\ \|S_\Delta f\|_w\}\ \le\ \|f\|_{L^2([T_0,T_0+T])}
\qquad \text{for all } f\in L^2([T_0,T_0+T]),
\end{equation}
and therefore
\[
\varepsilon_\pm^{\mathrm{tail}}
\le \frac{\|\rho_{\pm,\ell}\|_{L^2([T_0,T_0+T])}}{\|y_{0,\pm}\|_w},
\qquad
\varepsilon_\pm^{\mathrm{meas}}
\le \frac{\|\eta_\pm\|_{L^2([T_0,T_0+T])}}{\|y_{0,\pm}\|_w}.
\]
\end{lemma}

\begin{proof}
The first inequality \eqref{eq:eps-split} follows from the triangle inequality:
\[
\begin{aligned}
\max\{\|r_\pm\|_w,\|S_\Delta r_\pm\|_w\}
&\le
\max\{\|\rho_{\pm,\ell}\|_w,\|S_\Delta\rho_{\pm,\ell}\|_w\}\\
&\quad+
\max\{\|\eta_\pm\|_w,\|S_\Delta\eta_\pm\|_w\}.
\end{aligned}
\]
and dividing by $\|y_{0,\pm}\|_w$ gives \eqref{eq:eps-split}.
The bound \eqref{eq:eps-L2} is Lemma~\ref{lem:shifted-residual} applied to $f$.
\end{proof}

\subsection{From frequency errors to parameter errors}\label{subsec:freq-to-param}

Define the estimated QNM observables from the extracted frequencies:
\begin{equation}\label{eq:UhatVhat}
\widehat U_\ell:=\Re\frac{\widehat\omega_+ + \widehat\omega_-}{2\ell},
\qquad
\widehat V_\ell:=\Re\frac{\widehat\omega_+ - \widehat\omega_-}{2\ell},
\qquad
\widehat{\mathcal G}_\ell:=(\widehat U_\ell,\widehat V_\ell).
\end{equation}
Let $(\widehat M,\widehat a)$ be defined by local inversion:
\begin{equation}\label{eq:MH-def}
(\widehat M,\widehat a):=\mathcal G_\ell^{-1}\big(\widehat{\mathcal G}_\ell\big),
\end{equation}
where $\mathcal G_\ell^{-1}$ is the local inverse guaranteed by Theorem~\ref{thm:inv-companion}
on a neighborhood of the true point $(M,a)$.

\begin{lemma}[Observable error in terms of frequency errors]\label{lem:UV-error}
Let $\delta\omega_\pm:=\widehat\omega_\pm-\omega_{\pm,\ell}$. Then
\begin{equation}\label{eq:UV-error}
|\widehat U_\ell-U_\ell|
\ \le\
\frac{|\delta\omega_+|+|\delta\omega_-|}{2\ell},
\qquad
|\widehat V_\ell-V_\ell|
\ \le\
\frac{|\delta\omega_+|+|\delta\omega_-|}{2\ell},
\end{equation}
and hence
\begin{equation}\label{eq:G-error}
|\widehat{\mathcal G}_\ell-\mathcal G_\ell(M,a)|
\ \le\
\frac{\sqrt2}{2\ell}\,\big(|\delta\omega_+|+|\delta\omega_-|\big).
\end{equation}
\end{lemma}

\begin{proof}
By definition,
\[
\widehat U_\ell-U_\ell
=
\Re\frac{\delta\omega_+ + \delta\omega_-}{2\ell},
\qquad
\widehat V_\ell-V_\ell
=
\Re\frac{\delta\omega_+ - \delta\omega_-}{2\ell}.
\]
Thus $|\widehat U_\ell-U_\ell|\le \frac{|\delta\omega_+|+|\delta\omega_-|}{2\ell}$ and similarly for $V$.
The bound \eqref{eq:G-error} follows from $|(x,y)|\le \sqrt2\max\{|x|,|y|\}$ and \eqref{eq:UV-error}.
\end{proof}

We can now state the main parameter-bias bound.

\subsection{Assumptions and quantitative inputs for the bias estimate}\label{subsec:bias-inputs}

For the reader's convenience, we record here the precise ingredients that enter the deterministic
parameter-bias bound.  Throughout this subsection we fix $(M,a)\in\mathcal K$ and an integer
$\ell\ge\max\{\ell_0,\ell_1\}$.

\medskip

\noindent\textbf{(A1) Two-mode ringdown model and tail control.}
By the two-mode dominance result (Theorem~\ref{thm:two-mode-dominance}), the mode-separated
signals produced by the analytic windows satisfy, for $t\ge 1$,
\[
\mathsf O\big(u^{(\pm,\ell)}(t)\big)
=a_{\pm,\ell}e^{-i\omega_{\pm,\ell}t}+\rho_{\pm,\ell}(t),
\]
with $\rho_{\pm,\ell}$ obeying the explicit tail bounds \eqref{eq:rho-tail-bound}.
In particular, the residual on a window $[T_0,T_0+T]$ admits an a priori estimate in terms of
$(T_0,T)$, the shifted-contour decay rate $\nu$, and the high-frequency leakage order $\ell^{-N'}$.

\medskip

\noindent\textbf{(A2) Detectability of the extracted modes.}
We assume that the chosen detector does not annihilate the dominant extracted components,
\begin{equation*}
a_{+,\ell}\neq 0,\qquad a_{-,\ell}\neq 0,
\end{equation*}
that is, Assumption~\ref{ass:detectability-app} holds on the window.

\medskip

\noindent\textbf{(A3) Frequency extraction with a small residual.}
We measure the observed signals
\(\widetilde y_\pm(t)=a_{\pm,\ell}e^{-i\omega_{\pm,\ell}t}+r_\pm(t)\)
on $[T_0,T_0+T]$ in the weighted norm $\|\cdot\|_w$ of Section~\ref{subsec:one-mode-estimator}, with
shift $\Delta\in(0,T)$, and define the relative residual sizes $\varepsilon_\pm$ by \eqref{eq:eps-def}.
When $\varepsilon_\pm$ satisfy the quantitative signal-to-residual smallness conditions appearing in
Proposition~\ref{prop:freq-error-app}, the shift Rayleigh quotient estimator yields deterministic
frequency errors bounded by \eqref{eq:freq-error-app}.

\medskip

\noindent\textbf{(A4) Conditioning of the parameter inverse map.}
The companion paper provides a local inverse of the equatorial frequency map
$\mathcal G_\ell=(U_\ell,V_\ell)$ with a uniform stability constant $C_*$,
as stated in Theorem~\ref{thm:inv-companion}; equivalently, the Jacobian of $\mathcal G_\ell$
is uniformly nondegenerate on $\mathcal K$.

\medskip

With these inputs fixed, the next theorem gives a deterministic parameter bias bound in terms of
the residual sizes $\varepsilon_\pm$.

\begin{theorem}[Parameter bias bound from ringdown-based inversion]\label{thm:param-bias}
Fix $\Lambda>0$, $\mathcal K$, and $n$ as in \eqref{eq:overtone-range-app}.
Choose $\ell\ge\max\{\ell_0,\ell_1\}$ so that:
(i) the two-mode dominance theorem (Theorem~\ref{thm:two-mode-dominance}) holds, and
(ii) the inverse map theorem (Theorem~\ref{thm:inv-companion}) holds.
Fix an observation window $[T_0,T_0+T]$ with $T_0\ge 1$, a shift $\Delta\in(0,T)$, and a weight $w$ as in \eqref{eq:weight}.

Assume the inputs \textup{(A1)--(A4)} from Section~\ref{subsec:bias-inputs}, and define $\varepsilon_\pm$ by \eqref{eq:eps-def}.
If $\varepsilon_\pm\le \min\{1/8,\ |z_\pm|/20\}$ with $z_\pm=e^{-i\omega_{\pm,\ell}\Delta}$, then the reconstructed parameters
\eqref{eq:MH-def} satisfy the deterministic error bound
\begin{equation}\label{eq:param-bias-bound}
|(\widehat M,\widehat a)-(M,a)|
\ \le\
\frac{5\sqrt2\,C_*}{\Delta\,\ell}\,
\left(\frac{\varepsilon_+}{|z_+|}+\frac{\varepsilon_-}{|z_-|}\right),
\end{equation}
where $C_*$ is the inverse stability constant from Theorem~\ref{thm:inv-companion}.

Moreover, splitting $\varepsilon_\pm$ as in Lemma~\ref{lem:eps-split}, one obtains a decomposition into
a \emph{PDE tail bias} term and a \emph{measurement bias} term:
\begin{equation}\label{eq:param-bias-split}
|(\widehat M,\widehat a)-(M,a)|
\ \le\
\mathcal B_{\mathrm{tail}}+\mathcal B_{\mathrm{meas}},
\end{equation}
where
\begin{equation}\label{eq:tail-bias}
\mathcal B_{\mathrm{tail}}
:=
\frac{5\sqrt2\,C_*}{\Delta\,\ell}\left(\frac{\varepsilon^{\mathrm{tail}}_+}{|z_+|}+\frac{\varepsilon^{\mathrm{tail}}_-}{|z_-|}\right),
\qquad
\mathcal B_{\mathrm{meas}}
:=
\frac{5\sqrt2\,C_*}{\Delta\,\ell}\left(\frac{\varepsilon^{\mathrm{meas}}_+}{|z_+|}+\frac{\varepsilon^{\mathrm{meas}}_-}{|z_-|}\right).
\end{equation}
In particular, using \eqref{eq:rho-tail-bound} together with Lemma~\ref{lem:shifted-residual} and Lemma~\ref{lem:mode-energy},
one obtains explicit upper bounds for $\varepsilon^{\mathrm{tail}}_\pm$ (and hence for $\mathcal B_{\mathrm{tail}}$)
in terms of $(T_0,T,\Delta,w)$, the tail decay rate $\nu$, and the amplitudes $a_{\pm,\ell}$.
\end{theorem}

\begin{proof}
By Proposition~\ref{prop:freq-error-app},
\[
|\delta\omega_\pm|
=|\widehat\omega_\pm-\omega_{\pm,\ell}|
\le
\frac{10}{\Delta\,|z_\pm|}\,\varepsilon_\pm.
\]
Insert this into Lemma~\ref{lem:UV-error}:
\[
|\widehat{\mathcal G}_\ell-\mathcal G_\ell(M,a)|
\le
\frac{\sqrt2}{2\ell}\left(\frac{10}{\Delta|z_+|}\varepsilon_+ + \frac{10}{\Delta|z_-|}\varepsilon_-\right)
=
\frac{5\sqrt2}{\Delta\ell}\left(\frac{\varepsilon_+}{|z_+|}+\frac{\varepsilon_-}{|z_-|}\right).
\]
Now apply the inverse stability estimate \eqref{eq:stab-G} from Theorem~\ref{thm:inv-companion}:
\[
|(\widehat M,\widehat a)-(M,a)|
\le
C_*\,|\widehat{\mathcal G}_\ell-\mathcal G_\ell(M,a)|,
\]
which yields \eqref{eq:param-bias-bound}.
The split \eqref{eq:param-bias-split}--\eqref{eq:tail-bias} follows by combining \eqref{eq:param-bias-bound}
with Lemma~\ref{lem:eps-split}.
\end{proof}

\begin{corollary}[Explicit bias bound in terms of the PDE tail and measurement noise]\label{cor:param-bias-explicit}
Assume the hypotheses of Theorem~\ref{thm:param-bias} and, in addition, that $T>3\Delta$.
Let $y_{0,\pm}(t):=a_{\pm,\ell}e^{-i\omega_{\pm,\ell}t}$ and $r_\pm=\rho_{\pm,\ell}+\eta_\pm$ as in \eqref{eq:measured-y}.
Then
\begin{equation}\label{eq:param-bias-explicit}
|(\widehat M,\widehat a)-(M,a)|
\ \le\
\frac{5\sqrt2\,C_*}{\Delta\,\ell}\sum_{\varsigma\in\{+,-\}}
\frac{1}{|z_\varsigma|}\,
\frac{\|r_\varsigma\|_{L^2([T_0,T_0+T])}}{\|y_{0,\varsigma}\|_w}.
\end{equation}
Moreover, by Lemma~\ref{lem:mode-energy},
\[
\|y_{0,\varsigma}\|_w\ \ge\ |a_{\varsigma,\ell}|\left(\int_{T_0+\Delta}^{T_0+T-2\Delta} e^{2\Im\omega_{\varsigma,\ell} t}\,dt\right)^{1/2},
\]
and therefore
\begin{equation}\label{eq:param-bias-explicit2}
\begin{aligned}
|(\widehat M,\widehat a)-(M,a)|
&\le
\frac{5\sqrt2\,C_*}{\Delta\,\ell}\sum_{\varsigma\in\{+,-\}}
\frac{1}{|z_\varsigma|}\,
\frac{\|\rho_{\varsigma,\ell}\|_{L^2([T_0,T_0+T])}+\|\eta_\varsigma\|_{L^2([T_0,T_0+T])}}{|a_{\varsigma,\ell}|}\\
&\qquad\times
\left(\int_{T_0+\Delta}^{T_0+T-2\Delta} e^{2\Im\omega_{\varsigma,\ell} t}\,dt\right)^{-1/2}.
\end{aligned}
\end{equation}
In particular, integrating the pointwise bound \eqref{eq:rho-tail-bound} gives
\begin{align*}
\|\rho_{\varsigma,\ell}\|_{L^2([T_0,T_0+T])}
&\le
C_{s,m,N'}\,\|\mathsf O\|\Big(e^{-\nu T_0}(1+T_0)^{-m}+\ell^{-N'}\Big)\,\sqrt{T}\\
&\qquad\times \|(f_0,f_1)\|_{\mathcal H^{s+m+N_{m,N'}}}.
\end{align*}
\end{corollary}

\begin{proof}
By Lemma~\ref{lem:shifted-residual} and $0\le w\le 1$, one has
$\max\{\|r_\varsigma\|_w,\|S_\Delta r_\varsigma\|_w\}\le \|r_\varsigma\|_{L^2([T_0,T_0+T])}$.
Thus $\varepsilon_\varsigma\le \|r_\varsigma\|_{L^2([T_0,T_0+T])}/\|y_{0,\varsigma}\|_w$.
Insert this into \eqref{eq:param-bias-bound} to obtain \eqref{eq:param-bias-explicit}.
The inequality \eqref{eq:param-bias-explicit2} follows from Lemma~\ref{lem:mode-energy} and the triangle inequality
$\|r_\varsigma\|_{L^2}\le \|\rho_{\varsigma,\ell}\|_{L^2}+\|\eta_\varsigma\|_{L^2}$.
Finally, integrating \eqref{eq:rho-tail-bound} over $[T_0,T_0+T]$ yields the displayed $L^2$ bound for $\rho_{\varsigma,\ell}$.
\end{proof}

\subsection{Three-parameter ringdown inversion: deterministic bias bound for \texorpdfstring{$(M,a,\Lambda)$}{(M,a,Lambda)}}\label{subsec:application-3}

In many applications the cosmological constant is treated as known.  From a mathematical viewpoint, however, it is natural to ask
whether the same high-frequency equatorial package can support local recovery of the \emph{three} Kerr--de~Sitter parameters
$(M,a,\Lambda)$.
At the spectral level this requires a third independent real observable.
The companion paper~\cite{LiPartI} shows that, away from $a=0$, adding one damping observable extracted from the imaginary part
of a single equatorial mode yields such a three-parameter inverse theorem.
In this subsection we propagate the deterministic PDE-to-data error chain through that three-parameter inverse map.

\subsubsection*{Three-parameter inverse map from equatorial QNMs (input from the companion paper)}

Let
\begin{equation}\label{eq:K-3-compact-app}
\mathcal K^{(3)}\Subset \bigl\{(M,a,\Lambda):\ \Lambda>0,\ (M,a)\in\mathcal P_\Lambda,\ 0<a_1\le |a|\le a_0\bigr\}
\end{equation}
be a compact three-parameter slow-rotation set with $a_1>0$.
Fix an overtone index $n$ as in \eqref{eq:overtone-range-app}, and for $\ell\gg1$ let
\[
\omega_{\pm,\ell}(M,a,\Lambda):=\omega_{n,\ell,\pm\ell}(M,a,\Lambda)
\]
denote the labeled simple equatorial QNMs.
Define $U_\ell$ and $V_\ell$ as in \eqref{eq:UV-def-app} (now viewed as functions of $(M,a,\Lambda)$), and define the
additional damping observable
\begin{equation}\label{eq:Wtilde-def-app}
\widetilde W_\ell(M,a,\Lambda)
:=-\frac{\Im\,\omega_{+,\ell}(M,a,\Lambda)}{n+\tfrac12}.
\end{equation}
Set
\begin{equation}\label{eq:H-def-app}
\mathcal H_\ell(M,a,\Lambda):=\bigl(U_\ell(M,a,\Lambda),\,V_\ell(M,a,\Lambda),\,\widetilde W_\ell(M,a,\Lambda)\bigr)\in\mathbb R^3.
\end{equation}

\begin{theorem}[Local three-parameter inversion from equatorial QNMs; stability]\label{thm:inv-companion-3}
Fix $n$ and $\mathcal K^{(3)}$ as above.
Then there exist $\ell_0\in\mathbb N$ and $C_*^{(3)}>0$ such that for every $\ell\ge\ell_0$ and every point
$\mu:=(M,a,\Lambda)\in\mathcal K^{(3)}$ there exists a neighborhood $\mathcal N\subset\mathcal K^{(3)}$ of $\mu$ on which
the map $\mathcal H_\ell$ is injective and satisfies the quantitative stability estimate
\begin{equation}\label{eq:stab-H}
|\mu-\mu'|
\ \le\
C_*^{(3)}\,\bigl|\mathcal H_\ell(\mu)-\mathcal H_\ell(\mu')\bigr|,
\qquad \mu,\mu'\in\mathcal N.
\end{equation}
Moreover, $\mathcal H_\ell$ is real-analytic on $\mathcal K^{(3)}$.
\end{theorem}

\begin{proof}
This is proved in the companion paper~\cite[Theorem~51]{LiPartI} (see also \cite[\S~5.8]{LiPartI} for the
construction of the damping observable and the underlying Jacobian mechanism).
\end{proof}

\begin{remark}[Uniformity in $\Lambda$]\label{rem:uniform-in-Lambda}
The forward analysis in Sections~\ref{sec:setup2}--\ref{sec:freq-extraction} was presented with $\Lambda>0$ fixed.
All inputs used there (Fredholm/meromorphic resolvent theory, normally hyperbolic trapping estimates, sectorization in the equatorial
package, and the Grushin reduction underlying Proposition~\ref{prop:equatorial-projector-poly}) depend smoothly on the background
parameters, hence remain valid uniformly on compact parameter sets; a more formal uniformity statement is recorded in Appendix~\ref{app:lambda-uniformity}.
In particular, after restricting to $\mathcal K^{(3)}$ in \eqref{eq:K-3-compact-app} and shrinking auxiliary constants if needed,
all constants in Theorem~\ref{thm:two-mode-dominance} and Proposition~\ref{prop:freq-error-app} may be taken uniform on
$\mathcal K^{(3)}$.
\end{remark}

\subsubsection*{Propagation of time-domain errors to the three-parameter data map}

Let $\mu=(M,a,\Lambda)\in\mathcal K^{(3)}$ be the true parameters and let $\ell$ be sufficiently large so that
Theorem~\ref{thm:two-mode-dominance}, Proposition~\ref{prop:freq-error-app}, and Theorem~\ref{thm:inv-companion-3} apply.
As in Section~\ref{subsec:obs-model}, we observe the mode-separated scalar signals
\(\widetilde y_\pm(t)=a_{\pm,\ell}e^{-i\omega_{\pm,\ell}t}+r_\pm(t)\)
on $[T_0,T_0+T]$, define $\varepsilon_\pm$ by \eqref{eq:eps-def}, and construct the estimated frequencies
$\widehat\omega_\pm$ by \eqref{eq:omega-hat}.
Define the estimated three-component data
\begin{equation}\label{eq:Hhat-def}
\widehat{\mathcal H}_\ell
:=\bigl(\widehat U_\ell,\,\widehat V_\ell,\,\widehat{\widetilde W}_\ell\bigr),
\end{equation}
where
\begin{equation}\label{eq:UVW-hat-def}
\begin{aligned}
\widehat U_\ell&:=\Re\frac{\widehat\omega_+ + \widehat\omega_-}{2\ell},
&\qquad
\widehat V_\ell&:=\Re\frac{\widehat\omega_+ - \widehat\omega_-}{2\ell},
\\
\widehat{\widetilde W}_\ell&:=-\frac{\Im\,\widehat\omega_+}{n+\tfrac12}.
\end{aligned}
\end{equation}

\begin{lemma}[Data-map perturbation bound in terms of frequency errors]\label{lem:H-error}
Let $\delta\omega_\pm:=\widehat\omega_\pm-\omega_{\pm,\ell}$.
Then
\begin{equation}\label{eq:H-error}
\bigl|\widehat{\mathcal H}_\ell-\mathcal H_\ell(\mu)\bigr|
\ \le\
\frac{\sqrt2}{2\ell}\bigl(|\delta\omega_+|+|\delta\omega_-|\bigr)
+\frac{|\delta\omega_+|}{n+\tfrac12}.
\end{equation}
\end{lemma}

\begin{proof}
By \eqref{eq:UVW-hat-def} and \eqref{eq:UV-def-app} (with $\Lambda$ treated as a variable),
\[
|\widehat U_\ell-U_\ell|\le \frac{1}{2\ell}\bigl(|\Re\delta\omega_+|+|\Re\delta\omega_-|\bigr)\le \frac{1}{2\ell}\bigl(|\delta\omega_+|+|\delta\omega_-|\bigr),
\]
and similarly
\[
|\widehat V_\ell-V_\ell|\le \frac{1}{2\ell}\bigl(|\delta\omega_+|+|\delta\omega_-|\bigr).
\]
Moreover,
\[
\bigl|\widehat{\widetilde W}_\ell-\widetilde W_\ell\bigr|
=\frac{1}{n+\tfrac12}\,|\Im\delta\omega_+|\le \frac{|\delta\omega_+|}{n+\tfrac12}.
\]
Combining these estimates and using
$\sqrt{|x|^2+|y|^2}\le \sqrt2\,\max\{|x|,|y|\}\le \frac{\sqrt2}{2}(|x|+|y|)$
gives \eqref{eq:H-error}.
\end{proof}

\subsubsection*{Three-parameter deterministic bias bound}

\begin{theorem}[Three-parameter parameter bias bound from ringdown-based inversion]\label{thm:param-bias-3}
Fix $n$ and $\mathcal K^{(3)}$ as in \eqref{eq:K-3-compact-app}.
Choose $\ell$ large so that Theorem~\ref{thm:two-mode-dominance}, Proposition~\ref{prop:freq-error-app}, and
Theorem~\ref{thm:inv-companion-3} hold on $\mathcal K^{(3)}$.
Fix an observation window $[T_0,T_0+T]$ with $T_0\ge 1$, a shift $\Delta\in(0,T)$, and a weight $w$ as in \eqref{eq:weight}.

Assume the mode detectability condition \eqref{eq:detectability} and define $\varepsilon_\pm$ by \eqref{eq:eps-def}.
If $\varepsilon_\pm\le \min\{1/8,\ |z_\pm|/20\}$ with $z_\pm=e^{-i\omega_{\pm,\ell}\Delta}$, then the reconstructed parameters
\(\widehat\mu:=(\widehat M,\widehat a,\widehat\Lambda)\) defined by inverting $\widehat{\mathcal H}_\ell$ satisfy
\begin{equation}\label{eq:param-bias-bound-3}
|\widehat\mu-\mu|
\ \le\
\frac{5\sqrt2\,C_*^{(3)}}{\Delta\,\ell}\left(\frac{\varepsilon_+}{|z_+|}+\frac{\varepsilon_-}{|z_-|}\right)
+\frac{10\,C_*^{(3)}}{\Delta\,(n+\tfrac12)}\,\frac{\varepsilon_+}{|z_+|}.
\end{equation}
\end{theorem}

\begin{proof}
By Proposition~\ref{prop:freq-error-app},
\begin{equation}\label{eq:dw-bounds-3}
|\delta\omega_\pm|
\le \frac{10}{\Delta\,|z_\pm|}\,\varepsilon_\pm.
\end{equation}
Insert \eqref{eq:dw-bounds-3} into Lemma~\ref{lem:H-error} to obtain
\[
\bigl|\widehat{\mathcal H}_\ell-\mathcal H_\ell(\mu)\bigr|
\le
\frac{\sqrt2}{2\ell}\left(\frac{10}{\Delta|z_+|}\varepsilon_+ + \frac{10}{\Delta|z_-|}\varepsilon_-\right)
+\frac{1}{n+\tfrac12}\,\frac{10}{\Delta|z_+|}\varepsilon_+.
\]
Using $\frac{\sqrt2}{2}\cdot 10=5\sqrt2$ gives
\begin{equation}\label{eq:H-error-2}
\bigl|\widehat{\mathcal H}_\ell-\mathcal H_\ell(\mu)\bigr|
\le
\frac{5\sqrt2}{\Delta\,\ell}\left(\frac{\varepsilon_+}{|z_+|}+\frac{\varepsilon_-}{|z_-|}\right)
+\frac{10}{\Delta\,(n+\tfrac12)}\,\frac{\varepsilon_+}{|z_+|}.
\end{equation}
Finally, apply the three-parameter inverse stability estimate \eqref{eq:stab-H} from Theorem~\ref{thm:inv-companion-3}:
\[
|\widehat\mu-\mu|
\le
C_*^{(3)}\,\bigl|\widehat{\mathcal H}_\ell-\mathcal H_\ell(\mu)\bigr|,
\]
and combine with \eqref{eq:H-error-2} to obtain \eqref{eq:param-bias-bound-3}.
\end{proof}

\subsection{Practical remarks on conditioning and choice of window parameters}\label{subsec:practical}

Theorem~\ref{thm:param-bias} is deterministic and makes explicit several important conditioning mechanisms.

\medskip

\noindent\textbf{(i) The high-frequency gain and condition number.}
The factor $1/\ell$ in \eqref{eq:param-bias-bound} is inherited from the normalization in \eqref{eq:UV-def-app} and from the companion-paper stability bound
\eqref{eq:stab-omega-app}. Equivalently, it is the $1/\ell$ improvement in the condition number of the raw-frequency inverse map discussed in
Remark~\ref{rem:cond-number}.

\medskip

\noindent\textbf{(ii) Tail versus signal.}
The dominant mode decays like $e^{\Im\omega t}$ while the remainder decays like $e^{-\nu t}$.
Since $\nu$ is chosen strictly above the mode damping (Section~\ref{sec:two-mode}),
the ratio improves like $e^{-(\nu-|\Im\omega|)T_0}=e^{-(\nu+\Im\omega)T_0}$ as $T_0$ increases, cf.\ Remark~\ref{rem:interpret}.
This gives a deterministic justification for choosing a ringdown start time \emph{after} prompt response has decayed, subject to the obvious limitation
that the absolute signal $|a_{\pm,\ell}e^{-i\omega_{\pm,\ell}t}|$ also decreases with $T_0$.

\medskip

\noindent\textbf{(iii) Shift step $\Delta$ and branch control.}
The frequency extraction error scales like $1/\Delta$ in \eqref{eq:freq-error-app},
while the branch ambiguity scales like $2\pi/\Delta$ in $\Re\omega$.
In the present setting, the pseudopole prior $\omega^\sharp_{\pm,\ell}$ and the analytic window $\widetilde g_{\pm,\ell}$ localize each frequency
to a neighborhood where Lemma~\ref{lem:branch-selection} selects a consistent logarithm branch.

\medskip

\noindent\textbf{(iv) Detectability and amplitude dependence.}
All bounds depend on the amplitudes through $\|y_{0,\pm}\|_w^{-1}$, hence on $|a_{\pm,\ell}|^{-1}$.
This is intrinsic: if a detector nearly annihilates a mode, then the corresponding frequency is poorly observed on that channel.
Assumption~\ref{ass:detectability-app} makes the nonvanishing requirement explicit, and Lemma~\ref{lem:generic-data-app}
shows that the requirement is generic for fixed $(M,a)$ and fixed $\ell$.
Quantitatively, the smallness conditions $\varepsilon_\pm\le 1/8$ in Proposition~\ref{prop:freq-error-app} are exactly
signal-to-residual inequalities on the chosen window; they encode how large the dominant-mode energy must be relative to the tail and the measurement perturbation
for deterministic extraction to be stable.

\section{Pseudospectral stability in the equatorial package}\label{sec:pseudospectral}

The stationary wave operator is non-selfadjoint, and QNM poles alone do not control the size of the resolvent away from the spectrum:
pseudospectral effects and transient growth can occur even when all poles lie strictly in the lower half-plane
\cite{JaramilloMacedoAlSheikh2021PRX,GasperinJaramillo2022ScalarProduct,CarballoWithers2024TransientDynamics,CaiCaoChenGuoWuZhou2025KerrPseudospectrum}.
In our setting, however, we do not work with the full resolvent but rather with an \emph{equatorially microlocalized} resolvent in a fixed semiclassical package.
The purpose of this section is to record a quantitative bound showing that, within this package, the localized resolvent cannot develop large pseudomode
growth far away from the labeled poles: its blow-up is confined to a controlled neighborhood of the poles, with polynomial dependence on~$\ell$.

\subsection{A microlocalized resolvent bound near labeled equatorial poles}\label{subsec:pseudo-local}

Fix an overtone index $n\in\mathbb N_0$ and a Sobolev order $s\ge0$.
For each $\ell$ and sign $\pm$, define the localized resolvent family
\begin{equation}\label{eq:localized-resolvent}
\mathcal R_{\pm,\ell}(\omega):=\chi\,A_{\pm,h_\ell}\,\Pi_{k_\pm}\,R(\omega)\,\chi,
\qquad k_\pm=\pm\ell,\qquad h_\ell=\ell^{-1}.
\end{equation}
\begin{definition}[Localized $\epsilon$-pseudospectrum in the equatorial package]\label{def:local-pseudospectrum}
Fix $s\ge0$ and $n\in\mathbb N_0$.  For each $\ell$ and sign $\pm$, and for $\epsilon>0$, we define the localized $\epsilon$-pseudospectrum
\begin{equation}\label{eq:Sigma-eps}
\Sigma^{(\pm,\ell)}_{\epsilon}
:=\Bigl\{\omega\in\bigcup_{0\le j\le n}D_{j,\pm,\ell}:\ \big\|\mathcal R_{\pm,\ell}(\omega)\big\|_{H^{s-1}\to H^{s+1}}>\epsilon^{-1}\Bigr\}.
\end{equation}
This set measures the growth of the microlocalized resolvent within the labeled equatorial neighborhood.
\end{definition}

On each disk $D_{j,\pm,\ell}$ around a labeled equatorial pole $\omega_{j,\pm,\ell}$ (with $0\le j\le n$),
the Grushin reduction from \cite[Appendix~B]{LiPartI} yields the decomposition \eqref{eq:grushin-decomp},
with a scalar quantization function $\mathfrak q_{\pm,\ell}$ having a simple zero at $\omega_{j,\pm,\ell}$ and satisfying the uniform normalization
\eqref{eq:q-derivative-lower}.

\begin{lemma}[Quantization function controls distance to a simple pole]\label{lem:q-controls-distance}
Possibly after shrinking the disk radius in the definition of $D_{j,\pm,\ell}$ (uniformly in $(M,a)$ and~$\ell$),
there exist constants $c_q,r_q>0$ such that for all $(M,a)\in\mathcal K$, all $\ell$ sufficiently large,
each sign $\pm$, each $j\in\{0,1,\dots,n\}$, and all $\omega\in D_{j,\pm,\ell}$ with $|\omega-\omega_{j,\pm,\ell}|\le r_q$,
\begin{equation}\label{eq:q-distance}
|\mathfrak q_{\pm,\ell}(\omega)|\ \ge\ c_q\,|\omega-\omega_{j,\pm,\ell}|.
\end{equation}
\end{lemma}

\begin{proof}
Fix $(M,a)$, $\ell$, $\pm$, and $j$.
Since $\mathfrak q_{\pm,\ell}$ is holomorphic on $D_{j,\pm,\ell}$ and has a simple zero at $\omega_{j,\pm,\ell}$, Taylor's theorem gives
\[
\mathfrak q_{\pm,\ell}(\omega)
=(\omega-\omega_{j,\pm,\ell})\,\partial_\omega\mathfrak q_{\pm,\ell}(\omega_{j,\pm,\ell})
+\frac12(\omega-\omega_{j,\pm,\ell})^2\,\partial_\omega^2\mathfrak q_{\pm,\ell}(\xi),
\]
for some $\xi$ on the line segment between $\omega$ and $\omega_{j,\pm,\ell}$.
The Grushin reduction provides uniform bounds on $\partial_\omega^2\mathfrak q_{\pm,\ell}$ on $D_{j,\pm,\ell}$
(see \cite[Appendix~B]{LiPartI}); thus there is $C_q>0$ with
$\sup_{D_{j,\pm,\ell}}|\partial_\omega^2\mathfrak q_{\pm,\ell}|\le C_q$ uniformly in $(M,a)\in\mathcal K$ and $\ell$ large.
Choose $r_q>0$ so that $\frac12 r_q C_q\le \frac12 c_*$, where $c_*>0$ is the lower bound in \eqref{eq:q-derivative-lower}.
Then for $|\omega-\omega_{j,\pm,\ell}|\le r_q$,
\[
|\mathfrak q_{\pm,\ell}(\omega)|
\ge |\omega-\omega_{j,\pm,\ell}|\,|\partial_\omega\mathfrak q_{\pm,\ell}(\omega_{j,\pm,\ell})|
-\frac12|\omega-\omega_{j,\pm,\ell}|^2\,\sup_{D_{j,\pm,\ell}}|\partial_\omega^2\mathfrak q_{\pm,\ell}|
\ge \frac12 c_*\,|\omega-\omega_{j,\pm,\ell}|.
\]
This proves \eqref{eq:q-distance} with $c_q=\tfrac12c_*$.
\end{proof}

\begin{proposition}[Localized resolvent bound near labeled equatorial poles]\label{prop:localized-resolvent-bound}
Assume Theorem~\ref{thm:input-pole-isolation}. Fix $s\ge0$ and $n\in\mathbb N_0$.
There exist constants $K\ge0$, $\ell_0\in\mathbb N$, and $C>0$ such that for all $(M,a)\in\mathcal K$,
all $\ell\ge\ell_0$, each sign $\pm$, each $j\in\{0,1,\dots,n\}$, and all $\omega\in D_{j,\pm,\ell}\setminus\{\omega_{j,\pm,\ell}\}$,
\begin{equation}\label{eq:localized-resolvent-bound}
\big\|\mathcal R_{\pm,\ell}(\omega)\big\|_{H^{s-1}\to H^{s+1}}
\ \le\ C\,\ell^{K}\left(1+\frac{1}{|\omega-\omega_{j,\pm,\ell}|}\right).
\end{equation}
\end{proposition}

\begin{proof}
Fix $(M,a)$, $\ell$, $\pm$, and $j$.
On $D_{j,\pm,\ell}$, the decomposition \eqref{eq:grushin-decomp} gives
\[
\mathcal R_{\pm,\ell}(\omega)
=
\mathcal R_{\pm,\ell}^{\mathrm{hol}}(\omega)
+
\mathcal E_{+}^{(\pm,\ell)}(\omega)\,\mathfrak q_{\pm,\ell}(\omega)^{-1}\,\mathcal E_{-}^{(\pm,\ell)}(\omega).
\]
By \cite[Appendix~B]{LiPartI}, the holomorphic term and the Grushin operators satisfy polynomial bounds in~$\ell$ as maps $H^{s-1}\to H^{s+1}$,
uniformly on $D_{j,\pm,\ell}$. Thus there exist $K\ge0$ and $C_0>0$ such that
\[
\sup_{\omega\in D_{j,\pm,\ell}}\big\|\mathcal R_{\pm,\ell}^{\mathrm{hol}}(\omega)\big\|_{H^{s-1}\to H^{s+1}}
\le C_0\ell^{K},
\qquad
\sup_{\omega\in D_{j,\pm,\ell}}\big\|\mathcal E_{\pm}^{(\pm,\ell)}(\omega)\big\|_{H^{s-1}\to H^{s+1}}
\le C_0\ell^{K}.
\]
For $\omega$ with $|\omega-\omega_{j,\pm,\ell}|\le r_q$, Lemma~\ref{lem:q-controls-distance} yields
$|\mathfrak q_{\pm,\ell}(\omega)|^{-1}\le c_q^{-1}|\omega-\omega_{j,\pm,\ell}|^{-1}$.
Combining these bounds gives \eqref{eq:localized-resolvent-bound} for such $\omega$.
If $|\omega-\omega_{j,\pm,\ell}|\ge r_q$, then $|\omega-\omega_{j,\pm,\ell}|^{-1}\le r_q^{-1}$ and the same inequality follows (with a larger constant) from the uniform polynomial bounds on $D_{j,\pm,\ell}$.
\end{proof}

\begin{corollary}[Localized pseudospectral disks near labeled poles]\label{cor:pseudospectral-disks}
Assume Theorem~\ref{thm:input-pole-isolation} and fix $s\ge0$ and $n\in\mathbb N_0$.
Let $\Sigma^{(\pm,\ell)}_{\epsilon}$ be the localized $\epsilon$-pseudospectrum from Definition~\ref{def:local-pseudospectrum}, and let
$\mathcal Z_{\pm,\ell}:=\{\omega_{j,\pm,\ell}:0\le j\le n\}$.
Then there exist $K\ge0$, $\ell_0\in\mathbb N$, and $C>0$ such that for every $\ell\ge\ell_0$ and every $\epsilon>0$,
\begin{equation}\label{eq:pseudo-disks}
\Sigma^{(\pm,\ell)}_{\epsilon}
\ \subset\ \bigcup_{\omega_j\in\mathcal Z_{\pm,\ell}}B\bigl(\omega_j,\,C\ell^{K}\epsilon\bigr).
\end{equation}
\end{corollary}

\begin{proof}
Let $\omega\in\Sigma^{(\pm,\ell)}_{\epsilon}$.  Then $\omega\in D_{j,\pm,\ell}$ for some $0\le j\le n$ and $\|\mathcal R_{\pm,\ell}(\omega)\|_{H^{s-1}\to H^{s+1}}>\epsilon^{-1}$.
By Proposition~\ref{prop:localized-resolvent-bound},
\[
\epsilon^{-1}< C\,\ell^{K}\left(1+\frac{1}{|\omega-\omega_{j,\pm,\ell}|}\right).
\]
If $\epsilon\ge (2C\ell^{K})^{-1}$ then the conclusion \eqref{eq:pseudo-disks} is trivial after enlarging the constant.
Otherwise $\epsilon^{-1}\ge 2C\ell^{K}$, hence the right-hand side forces
$\frac{1}{|\omega-\omega_{j,\pm,\ell}|}\ge \frac12 C^{-1}\ell^{-K}\epsilon^{-1}$, i.e.
$|\omega-\omega_{j,\pm,\ell}|\le 2C\,\ell^{K}\epsilon$.
This yields \eqref{eq:pseudo-disks}.
\end{proof}

\begin{remark}[Scope of the pseudospectral statement]\label{rem:pseudo-scope}
The inclusion \eqref{eq:pseudo-disks} is a \emph{localized} pseudospectral statement: it controls the growth of the resolvent only after
restricting to a fixed equatorial high-frequency package and inserting microlocal cutoffs.
It does not claim that the \emph{global} Kerr (or Kerr--de~Sitter) pseudospectrum is small; rather it identifies a regime in which the
relevant channel is quantitatively stable, which is precisely what is needed for the deterministic bias bounds in
Section~\ref{sec:application}.
\end{remark}

\subsection{A remark on extensions to higher spin}\label{subsec:spin}

The analysis in this paper is presented for the scalar wave equation, since the barrier-top equatorial quantization and Grushin reduction
used in Proposition~\ref{prop:equatorial-projector-poly} are currently established in that setting.
At a structural level, however, several steps are not tied to spin $0$:
the meromorphic/Fredholm framework for stationary families, the contour-shift resonance expansion with remainder,
the analytic-window calculus (entire weights), and the deterministic frequency-extraction stability theory apply to any stationary operator family
for which (i) one has a resonance expansion in a strip and (ii) one can control the relevant microlocalized spectral projectors.
For Kerr--de~Sitter, a robust linear and nonlinear stability framework for the Einstein vacuum equations with $\Lambda>0$
was developed by Hintz--Vasy \cite{HintzVasyKdSStability}, suggesting that analogous microlocal resolvent technology should be available
for tensorial operators once an appropriate high-frequency labeling and sectorization package is established.
Carrying out the equatorial semiclassical quantization and the associated projector bounds for the Teukolsky system (or for gauge-fixed linearized
Einstein operators) remains an interesting open direction.

\section{Discussion and outlook}\label{sec:discussion}

This paper is the second part of a series whose goal is to place a concrete version of high-frequency black-hole spectroscopy
on a fully deterministic mathematical footing.
Building on the semiclassical quantization and labeling theory developed in the companion paper~\cite{LiPartI}, we work in the
slow-rotation compact regime and focus on an equatorial high-frequency package with a fixed overtone index.
Starting from the resolvent-based resonant expansion with remainder established in Section~\ref{sec:resonant-expansion}, we show that a
suitable \emph{analytic} preprocessing of the time signal isolates the target overtone with quantitative control
(Section~\ref{sec:two-mode}).  We then establish deterministic stability of one- and two-frequency extraction from finite-length,
perturbed data (Section~\ref{sec:freq-extraction}), and propagate the resulting frequency errors through the local inverse map from
equatorial quasinormal frequencies to parameters, obtaining a quantitative parameter bias bound
(Section~\ref{sec:application}).

\paragraph{Positioning relative to current ringdown debates.}
Several recent works emphasize that nonnormality can manifest itself through pseudospectral sensitivity, transient growth, and non-modal
dynamics even when all QNMs are damped \cite{JaramilloMacedoAlSheikh2021PRX,CarballoWithers2024TransientDynamics,
CarballoPantelidouWithers2025NonModal,BessonCarballoPantelidouWithers2025TransientsReview,CaiCaoChenGuoWuZhou2025KerrPseudospectrum}.
Our results are compatible with these phenomena because the main theorems concern \emph{refined} signals obtained by an explicit
sequence of reductions and localizations: azimuthal projection, equatorial microlocal filtering, and an entire (polynomial) analytic
window.  The deterministic remainder bounds quantify what is discarded by these steps on a prescribed late-time window.
In particular, the two-mode dominance theorem should not be read as a claim that a generic unfiltered ringdown waveform is universally
two-mode; rather it identifies a semiclassical regime and a controlled preprocessing map in which a two-mode model is justified.

\paragraph{Comparison with single-frequency inverse results.}
In the mathematical inverse literature on black-hole spectroscopy, a useful reference point is the result of
Uhlmann--Wang~\cite{UhlmannWang2023BHMass}, who recover the black-hole mass from a single quasinormal mode.
Our setting is different in two complementary ways.
First, we work in a high-frequency equatorial package and use a \emph{pair} of frequencies $\omega_{\pm,\ell}$ to recover a \emph{pair}
of parameters $(M,a)$, with a conditioning gain of order $1/\ell$ in the slow-rotation regime.
Second, rather than proving identifiability from exact spectral data, we quantify how finite-window ringdown truncation and
deterministic measurement perturbations propagate through the extraction step and the inverse map; this leads to the explicit bias
bound \eqref{eq:param-bias-bound}.

\paragraph{Analytic pole selection and implementable filtering.}
A key constraint in time-domain ringdown analysis is that ``mode selection'' must be compatible with the analytic structure of the
Laplace transform: a non-holomorphic frequency cutoff destroys the contour-shift argument and generally has no clean interpretation in
the resolvent calculus.  Our approach therefore uses an \emph{entire} (indeed, polynomial) weight $\widetilde g_{\pm,\ell}(\omega)$
which vanishes at all poles in the relevant overtone strip except for the target pole $\omega_{\pm,\ell}$, and equals $1$ at
$\omega_{\pm,\ell}$.
The construction proceeds by Lagrange interpolation on the pseudopole lattice and a low-degree correction factor which suppresses the
growth of the interpolation polynomial on the shifted contour; see Appendix~\ref{app:analytic-window} and
Proposition~\ref{prop:equatorial-projector-poly}.
This is precisely the point where we avoid making any global claim about the size of residue projectors:
the window is designed so that the \emph{weighted} pole sum is absolutely convergent for the windowed signal, while the unweighted full
resonant expansion is used only in the sense of convergence of truncated sums (Section~\ref{sec:resonant-expansion}).

\medskip
\noindent\textbf{Time-domain meaning.}
Because $\widetilde g_{\pm,\ell}$ is a polynomial in $\omega$, multiplication by $\widetilde g_{\pm,\ell}(\omega)$ on the Laplace side
corresponds to applying a finite-order differential operator $\widetilde g_{\pm,\ell}(i\partial_{t_*})$ in the time domain.
In discrete data, this can be implemented by stable finite-difference combinations of a fixed number of time shifts (depending on the
chosen overtone index $n$ but \emph{independent} of $\ell$).  This provides a direct bridge between analytic pole calculus and an
algorithmically realizable preprocessing step.

\paragraph{Relation to QNM filters in data analysis.}
There is a natural connection between our analytic window construction and recent work in the gravitational-wave literature on
QNM filters, which aim to remove or isolate selected QNMs by applying a frequency-domain filter and then transforming back to the
time domain; see, for instance, Ma et al.~\cite{MaEtAl2022QNMFilters}.  The common theme is that (in a regime where QNM expansions are
meaningful) one seeks a preprocessing map whose action on an ideal sum of damped exponentials is transparent.
Our analytic windows differ in two respects that are important for the deterministic resolvent analysis pursued here.
First, the weight $\widetilde g_{\pm,\ell}$ is entire (indeed, polynomial), so it is automatically compatible with contour deformation
on the Laplace side and yields a remainder term that can be bounded directly by shifted-contour resolvent estimates.
Second, the window is constructed to localize a single overtone in a specific semiclassical package and to interact cleanly with the
microlocal sectorization estimate from the companion paper, thereby producing quantitative two-mode dominance in a Sobolev topology.

\paragraph{Nonnormality, excitation factors, and pseudospectra.}
Quasinormal-mode problems are fundamentally non-selfadjoint, and spectral data can be highly sensitive to perturbations.
This sensitivity is naturally quantified by pseudospectra and has been explored in the black-hole context in
\cite{JaramilloMacedoAlSheikh2021PRX}; see also the discussion of the structural role of the underlying scalar product in
\cite{GasperinJaramillo2022ScalarProduct}.  Recent numerical studies compute pseudospectra for Kerr in the scalar case by casting the
problem into a non-selfadjoint eigenvalue formulation on hyperboloidal slices and examining the norm dependence of the resolvent; see
for instance~\cite{CaiCaoChenGuoWuZhou2025KerrPseudospectrum} and the broader hyperboloidal perspective in
\cite{MacedoZenginoglu2024Hyperboloidal}.

From the viewpoint of deterministic ringdown inversion, the main analytic difficulty is not only that the spectrum may be sensitive,
but also that residue projectors (often referred to as \emph{excitation factors} in the physics literature) may have large operator
norms and can in principle mask small scalar weights.  In the equatorial barrier-top package the Grushin reduction from~\cite{LiPartI}
yields an explicit polynomial bound for the relevant microlocalized residue projectors
(Proposition~\ref{prop:equatorial-projector-poly}).
Combined with the $\mathcal O(\ell^{-\infty})$ suppression produced by the analytic weights $\widetilde g_{\pm,\ell}$, this gives a
two-mode dominance statement that is robust against hidden large excitation factors in the regime considered here.
Understanding how far this polynomial-control regime extends (and how it deteriorates outside it) is an important direction for
further work.

\paragraph{Transient dynamics and superradiant amplification.}
The analysis in this series is intentionally late-time and deterministic: we extract frequencies from a window $[T_0,T_0+T]$ chosen
beyond the initial prompt response, and we treat everything outside the selected finite-dimensional model as a controlled remainder.
Recent work emphasizes, however, that nonnormality can manifest itself through large \emph{transient} effects even when all QNMs are
damped.  In particular, the behavior of truncated QNM sums and the possibility of transient growth/superradiant amplification have been
analyzed in \cite{CarballoWithers2024TransientDynamics,CarballoPantelidouWithers2025NonModal}; see also the overview
\cite{BessonCarballoPantelidouWithers2025TransientsReview}.
From an inverse-problem perspective, such transients may create an additional bias source if the observation window overlaps the growth
phase.  The deterministic framework developed here suggests two mathematically clean ways to incorporate these phenomena: either enlarge
$T_0$ so that the transient component is absorbed into the exponentially damped tail, or extend the signal model to include a finite
number of non-modal contributions, leading naturally to higher-rank pencil/Prony estimators and to condition numbers governed by
Vandermonde-type matrices.  The two-frequency stability results in Section~\ref{subsec:two-exp} provide a first step toward such
multi-component deterministic models.

\paragraph{Scope of the two-mode regime and failure mechanisms.}
The two-mode dominance result (Theorem~\ref{thm:two-mode-dominance}) rests on three structural inputs:
\begin{enumerate}
\item \emph{Pole isolation and layer gaps}: the target pole $\omega_{\pm,\ell}$ is isolated in an overtone strip, and the shifted
contour can be placed in a pole-free region with uniform resolvent bounds (Section~\ref{subsec:param-uniformity}).
\item \emph{Microlocal sectorization}: after applying the equatorial/azimuthal localization $A_{\pm,h_\ell}\Pi_{k_\pm}$, all poles
outside the equatorial package contribute only $\mathcal O(\ell^{-\infty})$ to the windowed signal (Appendix~\ref{app:companion-inputs}).
\item \emph{Late-time suppression of transients}: the start time $T_0$ is chosen so that the remainder term produced by contour
shifting is small compared to the selected pole contribution (Sections~\ref{sec:resonant-expansion} and \ref{sec:two-mode}).
\end{enumerate}
In addition, detectability (Assumption~\ref{ass:detectability-app}) excludes the degenerate case where the selected mode is annihilated
by the observable or by dual orthogonality.

These hypotheses highlight concrete mechanisms by which the two-mode regime can fail.
Near the boundary of the subextremal parameter set $\mathcal P_\Lambda$ (e.g.\ in near-extremal limits), spectral gaps may shrink and
different QNM families can approach each other, making pole isolation delicate and potentially enhancing nonnormal effects.
Even away from such limits, poor choices of observable or initial data can yield very small amplitudes, amplifying the impact of any
residual tail.  Finally, for early start times the prompt response and transient growth dominate the signal and invalidate any
asymptotic two-mode model.

\paragraph{Exceptional points and near-coalescence.}
Exceptional points and avoided crossings provide another mechanism by which inverse problems become ill-conditioned: as two modes
approach coalescence, spectral projectors and inverse maps can lose regularity, and small data errors can translate into large parameter
bias.  This phenomenon has been highlighted in black-hole ringdown studies from the viewpoint of non-Hermitian physics; see for
instance~\cite{MacedoKatagiriKubotaMotohashi2025EP} and the analysis of resonant excitation in~\cite{Motohashi2025ResonantExcitation}.
In the present series the same mechanism appears transparently through the condition-number interpretation of the parameter map
(Section~\ref{subsec:practical}): the constants in the bias bounds reflect the local inverse Lipschitz constant, which necessarily
deteriorates in near-coalescence regimes.
A systematic extension of the deterministic extraction and bias analysis to neighborhoods of exceptional points would require a careful
treatment of generalized resonant states and possibly higher-order poles, together with sharp stability bounds for confluent Prony
systems, cf.~\cite{BatenkovYomdinSJAM2013}.

\paragraph{Completeness, Keldysh-type expansions, and global mode decompositions.}
The present work is deliberately local in the spectral plane: it relies on isolating a fixed overtone band and a pair of equatorial
branches, rather than on global completeness of quasinormal modes.  At the same time, there has been notable recent progress on global
expansion frameworks.  A hyperboloidal Keldysh-type approach to quasinormal mode expansions has been developed in
\cite{BessonJaramillo2025Keldysh}, and complementary perspectives on complete mode decompositions beyond pure QNM sums have been
proposed, for instance, in \cite{ArnaudoCarballoWithers2025CompleteModes}.  Connecting such global expansion theories with the
deterministic inverse analysis pursued here would require uniform quantitative control of remainders (tails, branch-cut contributions,
and high overtones) in norms compatible with the inverse map, together with robust bounds on the growth of the relevant spectral
projectors.

\paragraph{Further directions.}
The present series focuses on the scalar wave equation on Kerr--de~Sitter, where the meromorphic resolvent theory and the high-frequency
semiclassical machinery are particularly clean.  Extending the inverse framework to other field spins and to more general observation
models (for instance, multiple channels or non-equatorial projections) would be valuable.
Another natural direction is to relax the small-rotation restriction by combining the equatorial semiclassical package of \cite{LiPartI}
with global estimates for the full subextremal Kerr--de~Sitter family, such as \cite{PetersenVasyKerrDeSitterJEMS}.
Finally, while the de~Sitter setting provides exponential decay and avoids asymptotically flat branch-cut tails, it would be of
substantial interest to develop analogues of the present deterministic inverse bounds for asymptotically flat Kerr, where polynomial
late-time tails and continuum contributions are unavoidable and must be incorporated into the error bookkeeping.
\appendix

\section{Fourier--Laplace conventions and the distributional source term}\label{app:laplace}

This appendix fixes sign conventions and records the distributional identities used in Section~\ref{sec:setup2}.

\subsection{Forward Fourier--Laplace transform}

For a distribution $v$ supported in $t\ge0$ we use the forward Fourier--Laplace transform
\[
\widehat v(\omega):=\int_0^\infty e^{i\omega t}\,v(t)\,dt,
\qquad \Im\omega>C_0,
\]
and the inverse formula (for suitable $C>C_0$)
\[
v(t)=\frac{1}{2\pi}\int_{\Im\omega=C} e^{-i\omega t}\,\widehat v(\omega)\,d\omega.
\]
With this convention one has, for test functions $\phi$,
\[
\langle \delta,\phi\rangle=\phi(0),\qquad \langle \delta',\phi\rangle=-\phi'(0),
\]
and therefore
\begin{equation}\label{eq:delta-transform}
\widehat{\delta}(\omega)=1,\qquad \widehat{\delta'}(\omega)=-i\omega.
\end{equation}

\subsection{Resolvent identity for the Cauchy problem}

Let $u$ solve the wave equation $Pu=0$ with initial data $(u,\partial_{t_*}u)|_{t_*=0}=(f_0,f_1)$, and let $u^+:=\mathbfone_{t_*\ge0}u$.
A direct computation in distributions gives
\[
Pu^+=\delta'(t_*)\,f_0+\delta(t_*)\,(f_1+Qf_0).
\]
Taking the Fourier--Laplace transform and using \eqref{eq:delta-transform} yields
\[
P(\omega)\widehat u^+(\omega)=(f_1+Qf_0)-i\omega f_0,
\]
so that, whenever $R(\omega)=P(\omega)^{-1}$ exists,
\[
\widehat u^+(\omega)=R(\omega)\bigl(f_1+(Q-i\omega)f_0\bigr).
\]
This is precisely Lemma~\ref{lem:resolvent-identity}.

\subsection{Motivation for the forced formulation}

In the main text we avoid distributional source terms by passing to the forced problem with a smooth cutoff in time:
we set $u_\vartheta(t_*):=\vartheta(t_*)u(t_*)$ with a fixed $\vartheta\in C^\infty$ supported in $(0,\infty)$ and equal to $1$ for $t_*\ge1$.
Then $u_\vartheta$ satisfies $Pu_\vartheta=F_\vartheta$ with $F_\vartheta$ smooth and compactly supported in time, and
$u_\vartheta=u$ for $t_*\ge1$.  This is the starting point for the contour deformation arguments in
Section~\ref{sec:resonant-expansion} and for the analytically windowed evolution in Section~\ref{sec:two-mode}.

\section{Dual resonant states, biorthogonality, and excitation amplitudes}\label{app:dual-states}

This appendix records a standard rank-one formula for the residue projector at a simple pole of a meromorphic
Fredholm resolvent and explains how the \emph{scalar} amplitude seen by a detector factors through a pairing
with a \emph{dual} (left) resonant state.  While the contour deformation arguments of Section~\ref{sec:resonant-expansion}
do not rely on these identities, they provide an operator-theoretic interpretation of the detectability requirement
in Section~\ref{subsec:obs-model} and clarify how nonnormality may enter through the size of the residue projector.
For background on meromorphic Fredholm families and Keldysh-type expansions we refer to
\cite{GohbergSigal1971,KatoPerturbation,BessonJaramillo2025Keldysh}.

\subsection{Rank-one structure of the residue at a simple pole}\label{subsec:rank-one-residue}

Let $X,Y$ be complex Banach spaces and let $P(\omega):X\to Y$ be a holomorphic family of Fredholm operators of index $0$
in a neighbourhood of $\omega_0\in\mathbb C$.  Assume that $P(\omega)$ is invertible for some $\omega$ in this neighbourhood,
and denote the meromorphic inverse by $R(\omega)=P(\omega)^{-1}$.

Suppose that $\omega_0$ is a \emph{simple} pole of $R$ and that
\[
\dim\ker P(\omega_0)=\dim\ker P(\omega_0)^*=1,
\]
where $P(\omega_0)^*:Y^*\to X^*$ denotes the Banach space adjoint.  Choose nonzero vectors
\[
u_0\in\ker P(\omega_0)\subset X,\qquad v_0\in\ker P(\omega_0)^*\subset Y^*,
\]
and use the duality pairing $\langle\cdot,\cdot\rangle:Y\times Y^*\to\mathbb C$.  Since $v_0$ annihilates $\Ran P(\omega_0)$,
we have $\langle P(\omega_0)w,v_0\rangle=0$ for all $w\in X$.

\begin{lemma}[Residue projector at a simple pole]\label{lem:rank-one-residue}
Under the above assumptions,
\begin{equation}\label{eq:denominator-nonzero}
\big\langle \partial_\omega P(\omega_0)u_0,\,v_0\big\rangle\neq 0,
\end{equation}
and the residue operator $\Pi_{\omega_0}:=\Res_{\omega=\omega_0}R(\omega):Y\to X$ is the rank-one operator
\begin{equation}\label{eq:rank-one-projector}
\Pi_{\omega_0}f
=
\frac{\langle f,\,v_0\rangle}{\langle \partial_\omega P(\omega_0)u_0,\,v_0\rangle}\,u_0,
\qquad f\in Y.
\end{equation}
In particular, $\Ran \Pi_{\omega_0}=\ker P(\omega_0)$ and $\ker \Pi_{\omega_0}\supset \Ran P(\omega_0)$.
\end{lemma}

\begin{proof}
Since $R$ has a simple pole at $\omega_0$, there exists a holomorphic family $R_0(\omega)$ near $\omega_0$ and a finite-rank operator
$\Pi_{\omega_0}$ such that
\begin{equation}\label{eq:laurent-simple}
R(\omega)=\frac{1}{\omega-\omega_0}\,\Pi_{\omega_0}+R_0(\omega).
\end{equation}
Expand
\[
P(\omega)=P(\omega_0)+(\omega-\omega_0)\,\partial_\omega P(\omega_0)+\mathcal O\big((\omega-\omega_0)^2\big)
\]
and use $P(\omega)R(\omega)=\Id_Y$ for $\omega\neq \omega_0$.  Substituting \eqref{eq:laurent-simple} gives
\[
(\omega-\omega_0)^{-1}\,P(\omega_0)\Pi_{\omega_0}
+\Big(\partial_\omega P(\omega_0)\Pi_{\omega_0}+P(\omega_0)R_0(\omega_0)\Big)
+\mathcal O(\omega-\omega_0)
=\Id_Y.
\]
The coefficient of $(\omega-\omega_0)^{-1}$ must vanish, so $P(\omega_0)\Pi_{\omega_0}=0$ and hence
$\Ran \Pi_{\omega_0}\subset \ker P(\omega_0)=\mathrm{span}\{u_0\}$.  Thus there exists $\ell\in Y^*$ such that
\begin{equation}\label{eq:rank-one-form}
\Pi_{\omega_0}f=\langle f,\ell\rangle\,u_0,\qquad f\in Y.
\end{equation}
Similarly, the identity $R(\omega)P(\omega)=\Id_X$ implies $\Pi_{\omega_0}P(\omega_0)=0$, so $\ell$ annihilates $\Ran P(\omega_0)$.
Since $\dim\ker P(\omega_0)^*=1$, the cokernel is one-dimensional and therefore $\ell$ is a nonzero multiple of $v_0$,
say $\ell=c\,v_0$.

Applying the functional $v_0$ to the constant term identity and using $v_0\circ P(\omega_0)=0$ yields
\[
\big\langle \partial_\omega P(\omega_0)\Pi_{\omega_0}f,\,v_0\big\rangle
=
\langle f,\,v_0\rangle,
\qquad f\in Y.
\]
With \eqref{eq:rank-one-form} and $\ell=c\,v_0$, the left-hand side equals
$c\,\langle f,v_0\rangle\,\langle \partial_\omega P(\omega_0)u_0,v_0\rangle$.
Since $v_0\neq 0$, this forces \eqref{eq:denominator-nonzero} and
$c=\langle \partial_\omega P(\omega_0)u_0,v_0\rangle^{-1}$, which gives \eqref{eq:rank-one-projector}.
\end{proof}

\begin{remark}[Normalisation]\label{rem:dual-normalisation}
The representation \eqref{eq:rank-one-projector} is invariant under rescaling $u_0\mapsto \alpha u_0$ and $v_0\mapsto \beta v_0$.
One may therefore impose the normalisation
\begin{equation}\label{eq:dual-normalisation}
\big\langle \partial_\omega P(\omega_0)u_0,\,v_0\big\rangle = 1,
\end{equation}
in which case $\Pi_{\omega_0}f=\langle f,v_0\rangle u_0$.
In Hilbert settings where $X=Y$ (or after identifying $X$ with a graph space on $Y$), this corresponds to the familiar
right/left eigenvector normalisation for non-selfadjoint eigenvalue problems.
\end{remark}

\subsection{Detector amplitudes as dual pairings}\label{subsec:amplitude-dual-pairing}

We now connect Lemma~\ref{lem:rank-one-residue} to the resonance expansion of Section~\ref{sec:resonant-expansion}.
Recall that the stationary family for the wave operator is defined by \eqref{eq:stationary-family},
and that the resonant term associated with a pole $\omega_j$ is obtained by taking the residue of the integrand
$e^{-i\omega t_*}\chi R(\omega)\chi\,\widehat F_\vartheta(\omega)$ (cf.\ Theorem~\ref{thm:resonant-expansion}).

Assume that $\omega_0$ is a pole of $R(\omega)$ which is simple and has one-dimensional resonant and co-resonant spaces,
so that Lemma~\ref{lem:rank-one-residue} applies.  Let $u_0$ be a (right) resonant state and $v_0$ a dual (left) resonant
state, normalised by \eqref{eq:dual-normalisation}.  Then the residue contribution can be written as
\begin{equation}\label{eq:resonant-term-dual}
\chi\,\Pi_{\omega_0}\,\chi\,\widehat F_\vartheta(\omega_0)
=
\big\langle \chi\,\widehat F_\vartheta(\omega_0),\,v_0\big\rangle\,\chi u_0.
\end{equation}
Thus, for any bounded detector $\mathsf O:H^{s+1}\to\mathbb C$ supported where $\chi\equiv 1$, the corresponding scalar
amplitude satisfies
\begin{equation}\label{eq:scalar-amplitude-dual}
\mathsf O\big(\chi\,\Pi_{\omega_0}\,\chi\,\widehat F_\vartheta(\omega_0)\big)
=
\big\langle \chi\,\widehat F_\vartheta(\omega_0),\,v_0\big\rangle\;\mathsf O(\chi u_0).
\end{equation}
In particular, the nonvanishing of an observed mode amplitude factors into two independent requirements:
(i) the detector must not annihilate the spatial profile $\chi u_0$, and (ii) the initial data (equivalently, the forcing
$\widehat F_\vartheta(\omega_0)$) must have nontrivial pairing with the dual state $v_0$.

\begin{remark}[Excitation factors and nonnormality]\label{rem:excitation-factor}
If one uses an unnormalised pair $(u_0,v_0)$, the amplitude involves the denominator
$\langle \partial_\omega P(\omega_0)u_0,v_0\rangle$ in \eqref{eq:rank-one-projector}.  Smallness of this quantity corresponds
to a large residue operator norm and to large ``excitation factors'' in the physics literature; it also features in recent
work on avoided crossings and exceptional points in black-hole ringdown
\cite{Motohashi2025ResonantExcitation,MacedoKatagiriKubotaMotohashi2025EP}.
This is a precise sense in which nonnormality may magnify observed ringdown amplitudes, even when the decay rates are fixed.
\end{remark}

\subsection{Relation to Keldysh expansions}\label{subsec:keldysh-connection}

The rank-one formula \eqref{eq:rank-one-projector} is the simplest instance of a more general principle:
near a pole of a meromorphic Fredholm resolvent, the principal part can be written explicitly in terms of
right/left root vectors (Jordan chains) and their biorthogonality relations.  For higher-order poles this produces the
polynomial prefactors in time that already appear in the general resonance expansion of Theorem~\ref{thm:resonant-expansion};
in matrix and operator-pencil language this is precisely the content of Keldysh-type theorems.
In the black-hole setting, a hyperboloidal formulation makes these biorthogonal constructions particularly natural and has
recently been exploited to obtain spectral versions of resonance expansions; see~\cite{BessonJaramillo2025Keldysh} and the
references therein, as well as the complementary convergent mode-sum perspective developed in~\cite{ArnaudoCarballoWithers2025CompleteModes}.

In the present paper we do not assume any global completeness of resonant states.  Instead, we isolate a finite set of poles
in the equatorial package (Sections~\ref{sec:two-mode}--\ref{sec:application}) and control all remaining poles by a contour
remainder and by microlocal sectorization.  Appendix~\ref{app:dual-states} is included only to provide a structural
interpretation of the amplitudes and of Assumption~\ref{ass:detectability-app}.

\section{Analytic pole selection and band isolation}\label{app:analytic-window}

This appendix isolates the complex-analytic mechanism behind the ``analytic pole selection'' used in
Section~\ref{sec:two-mode}.  The guiding principle is:

\medskip
\noindent\emph{Residue calculus is only available for meromorphic integrands.}
\medskip

Consequently, frequency localization in a contour representation cannot be achieved by inserting
non-holomorphic cutoffs (for instance depending on $\Re\omega$), but it \emph{can} be achieved by
multiplying by carefully chosen \emph{entire} weights of controlled growth.

\subsection{Why non-holomorphic cutoffs cannot be used in residue calculus}

Consider a contour integral of the form
\[
\int_{\Im\omega=C} e^{-i\omega t}\,R(\omega)\widehat F(\omega)\,d\omega,
\]
where $R(\omega)$ is meromorphic.  If one inserts a factor depending on $\Re\omega$---for instance
$\psi(h\Re\omega)$ for a smooth cutoff $\psi$---then the integrand is no longer holomorphic (or meromorphic)
as a function of the complex variable $\omega$.  In that case one cannot justify deforming the contour and collecting
residues'' by Cauchy's theorem.

For this reason, Section~\ref{subsec:analytic-localization} uses the \emph{entire} weights
$\widetilde g_{\pm,\ell}$ defined in \eqref{eq:analytic-weight}; multiplying by an entire function preserves meromorphicity.

\subsection{Lagrange weights at pseudopoles: an abstract lemma}

We now record a general interpolation mechanism which is not specific to Kerr--de~Sitter, and which may be
useful whenever one has a finite collection of (semi)classical pseudopoles approximating true poles.

Fix $n\in\mathbb N_0$ and let $\{\Omega_j^\sharp\}_{j=0}^n\subset\mathbb C$ be pairwise distinct.
For each $m\in\{0,1,\dots,n\}$ define the Lagrange polynomial
\begin{equation}\label{eq:lagrange-weight}
G_m(\omega)
:=
\prod_{\substack{0\le j\le n\\ j\neq m}}
\frac{\omega-\Omega_j^\sharp}{\Omega_m^\sharp-\Omega_j^\sharp},
\qquad \omega\in\mathbb C.
\end{equation}
Then $G_m$ is entire of degree $n$ and satisfies $G_m(\Omega_m^\sharp)=1$ and $G_m(\Omega_j^\sharp)=0$ for $j\neq m$.

\begin{proposition}[Robust pole selection by pseudopole interpolation]\label{prop:pseudopole-interp}
Let $\{\Omega_j^\sharp\}_{j=0}^n$ and $G_m$ be as above, and set
\[
d_\sharp:=\min_{j\neq k}|\Omega_j^\sharp-\Omega_k^\sharp|>0.
\]
Assume that $\{\Omega_j\}_{j=0}^n\subset\mathbb C$ satisfies
\[
|\Omega_j-\Omega_j^\sharp|\le \delta,\qquad j=0,1,\dots,n,
\]
for some $\delta\in(0,d_\sharp/8]$.
Then for every $m\in\{0,\dots,n\}$ one has
\begin{align}
\label{eq:interp-robust-m}
|G_m(\Omega_m)-1|
&\le C_n^{\mathrm{Lag}}\,\frac{\delta}{d_\sharp},\\[0.3em]
\label{eq:interp-robust-off}
|G_m(\Omega_j)|
&\le C_n^{\mathrm{Lag}}\,\frac{\delta}{d_\sharp},
\qquad j\in\{0,\dots,n\}\setminus\{m\},
\end{align}
where $C_n^{\mathrm{Lag}}>0$ depends only on $n$.

Moreover, for every fixed $\eta\in\mathbb R$ there is a constant $C_{\eta,n}$ such that
\begin{equation}\label{eq:interp-growth}
|G_m(\sigma+i\eta)|\le C_{\eta,n}\,(1+|\sigma|)^n,\qquad \sigma\in\mathbb R,
\end{equation}
uniformly in $m$.
\end{proposition}

\begin{proof}
Fix $m$.

\emph{Step 1: uniform bounds near the interpolation nodes.}
Let $r:=d_\sharp/4$ and consider the closed disks
$B_j:=\{|\omega-\Omega_j^\sharp|\le r\}$, $j=0,\dots,n$.
They are pairwise disjoint by definition of $d_\sharp$.

For $\omega\in B_j$ and $k\neq j$, we have
\[
|\omega-\Omega_k^\sharp|
\le |\Omega_j^\sharp-\Omega_k^\sharp|+r
\le \Big(1+\frac14\Big)|\Omega_j^\sharp-\Omega_k^\sharp|,
\]
and also
\[
|\omega-\Omega_k^\sharp|
\ge |\Omega_j^\sharp-\Omega_k^\sharp|-r
\ge \Big(1-\frac14\Big)|\Omega_j^\sharp-\Omega_k^\sharp|.
\]
Using \eqref{eq:lagrange-weight} and the fact that each denominator satisfies
$|\Omega_m^\sharp-\Omega_k^\sharp|\ge d_\sharp$, we obtain a crude bound
\begin{equation}\label{eq:Gm-sup-disk}
\sup_{\omega\in \bigcup_{j=0}^n B_j}|G_m(\omega)|
\le \Big(\frac{5}{3}\Big)^n.
\end{equation}
(The constant is immaterial; any uniform bound depending only on $n$ is sufficient.)

\emph{Step 2: Lipschitz bound on each disk.}
Since $G_m$ is holomorphic, Cauchy's estimate on each disk $B_j$ gives
\[
\sup_{\omega\in B_j}|G_m'(\omega)|
\le
\frac{1}{r}\sup_{\omega\in B_j}|G_m(\omega)|
\le
\frac{1}{r}\Big(\frac{5}{3}\Big)^n
=
C_n^{\mathrm{Lag}}\,\frac{1}{d_\sharp},
\]
with $C_n^{\mathrm{Lag}}:=4(5/3)^n$.

\emph{Step 3: evaluate at perturbed nodes.}
Since $\delta\le d_\sharp/8=r/2$, each perturbed node $\Omega_j$ lies in $B_j$.
For $j=m$ we write
\[
G_m(\Omega_m)-1=G_m(\Omega_m)-G_m(\Omega_m^\sharp),
\]
and apply the mean value theorem along the segment from $\Omega_m^\sharp$ to $\Omega_m$ inside $B_m$ to obtain
\[
|G_m(\Omega_m)-1|
\le
\sup_{\omega\in B_m}|G_m'(\omega)|\,|\Omega_m-\Omega_m^\sharp|
\le
C_n^{\mathrm{Lag}}\,\frac{\delta}{d_\sharp},
\]
which is \eqref{eq:interp-robust-m}.
Similarly, for $j\neq m$ we use $G_m(\Omega_j^\sharp)=0$ and obtain \eqref{eq:interp-robust-off}.

\emph{Step 4: polynomial growth on horizontal lines.}
The growth bound \eqref{eq:interp-growth} is immediate from \eqref{eq:lagrange-weight} since $G_m$ is a degree-$n$ polynomial:
for $\omega=\sigma+i\eta$,
\[
|G_m(\omega)|
\le
\prod_{j\neq m}\frac{|\omega|+|\Omega_j^\sharp|}{|\Omega_m^\sharp-\Omega_j^\sharp|}
\le
\Big(\frac{1}{d_\sharp}\Big)^n\prod_{j\neq m}\big(|\sigma|+|\eta|+|\Omega_j^\sharp|\big)
\le
C_{\eta,n}\,(1+|\sigma|)^n,
\]
absorbing $\max_j|\Omega_j^\sharp|$ into $C_{\eta,n}$ (for fixed $n$).
\end{proof}

\begin{remark}[Specialization to the weights $g_{\pm,\ell}$ and $\widetilde g_{\pm,\ell}$]\label{rem:interp-specialization}
In Section~\ref{sec:two-mode}, for each sign $\pm$ we fix the pseudopoles
$\Omega_j^\sharp=\omega_{j,\pm,\ell}^\sharp$ ($0\le j\le n$) from
Theorem~\ref{thm:input-pole-isolation}, and we take $m=n$ in \eqref{eq:lagrange-weight}.
The resulting Lagrange polynomial $G_n$ is exactly the interpolation weight $g_{\pm,\ell}$ in \eqref{eq:analytic-weight}.
The modified window $\widetilde g_{\pm,\ell}$ in \eqref{eq:analytic-weight-modified} is obtained by multiplying $g_{\pm,\ell}$ by the
additional factor $(\omega/\omega^\sharp_{\pm,\ell})^{m_0}$; this preserves the interpolation identities \eqref{eq:g-interp}
and does not affect the residue calculus since it is entire.

Since the true poles satisfy $\omega_{j,\pm,\ell}=\omega_{j,\pm,\ell}^\sharp+\mathcal O(\ell^{-\infty})$ and the layer gaps are uniform,
Proposition~\ref{prop:pseudopole-interp} yields
\[
g_{\pm,\ell}(\omega_{\pm,\ell})=1+\mathcal O(\ell^{-\infty}),
\qquad
g_{\pm,\ell}(\omega_{j,\pm,\ell})=\mathcal O(\ell^{-\infty})\quad (0\le j\le n-1).
\]
Because $(\omega/\omega^\sharp_{\pm,\ell})^{m_0}$ is uniformly bounded on $\bigcup_{0\le j\le n}D_{j,\pm,\ell}$,
the same estimates hold with $g_{\pm,\ell}$ replaced by $\widetilde g_{\pm,\ell}$; this is Lemma~\ref{lem:pole-suppression}.
\end{remark}

\subsection{Polynomial growth and contour shifts for the forced problem}

The windowed evolution in Section~\ref{subsec:analytic-localization} involves integrals of the form
\[
\int_{\Im\omega=C} e^{-i\omega t}\,g(\omega)\,R(\omega)\widehat F_\vartheta(\omega)\,d\omega,
\]
with $g$ entire and of polynomial growth on horizontal lines (for $\widetilde g_{\pm,\ell}$ one has degree $n+m_0$).

Assume that $g$ is entire and satisfies, for every fixed $\eta\in\mathbb R$, a polynomial growth bound
\begin{equation}\label{eq:entire-poly-growth}
|g(\sigma+i\eta)|\ \le\ C_{\eta}\,(1+|\sigma|)^K,\qquad \sigma\in\mathbb R,
\end{equation}
for some $K\ge 0$.
Because $F_\vartheta(t_*,\cdot)$ is smooth and compactly supported in $t_*$, its Fourier--Laplace transform satisfies
(for every $N\in\mathbb N$ and fixed $\eta\in\mathbb R$)
\begin{equation}\label{eq:Fhat-decay-app}
\|\widehat F_\vartheta(\sigma+i\eta)\|_{H^{s-1}}
\le
C_{s,N,\eta}\,(1+|\sigma|)^{-N},\qquad \sigma\in\mathbb R,
\end{equation}
obtained by repeated integration by parts in $t_*$.
Moreover, for $|\omega|$ large the stationary family $P(\omega)$ is elliptic and $R(\omega)=P(\omega)^{-1}$ is polynomially bounded
in operator norm on each horizontal line.

Consequently, if $g$ satisfies \eqref{eq:entire-poly-growth} (for $\widetilde g_{\pm,\ell}$ this holds with $K=n+m_0$),
then the integrand decays as $|\Re\omega|\to\infty$ along horizontal segments, and the vertical sides in the standard rectangular
contour argument vanish as the rectangle expands.  This justifies contour shifts for the windowed integral.

\subsection{Band isolation by contour subtraction}

Besides interpolation weights, the paper uses another analytic tool which is often more robust: \emph{band isolation}
by subtracting two contour representations.  This method isolates the poles in a horizontal strip without requiring any
smallness of weights at other poles.

\begin{proposition}[Contour subtraction isolates a pole band]\label{prop:contour-subtraction-abstract}
Let $R(\omega)$ be a meromorphic family of bounded operators between fixed Sobolev spaces, and let $g$ be entire satisfying
a polynomial growth bound \eqref{eq:entire-poly-growth} on horizontal lines.
Assume $\widehat F(\omega)$ satisfies the rapid decay \eqref{eq:Fhat-decay-app} on horizontal lines.
Fix $\nu_1<\nu_2$ such that $R(\omega)$ has no poles on the lines $\Im\omega=-\nu_1$ and $\Im\omega=-\nu_2$.

Define
\[
I_{\nu}(t):=\frac{1}{2\pi}\int_{\Im\omega=-\nu} e^{-i\omega t}\,g(\omega)\,R(\omega)\widehat F(\omega)\,d\omega.
\]
Then for every $t\ge0$ one has the identity
\begin{equation}\label{eq:band-isolation-identity-app}
I_{\nu_1}(t)-I_{\nu_2}(t)
=
i\sum_{\substack{\omega_j\in\mathrm{Poles}(R)\\ -\nu_2<\Im\omega_j<-\nu_1}}
e^{-i\omega_j t}\,\Res_{\omega=\omega_j}\big(g(\omega)\,R(\omega)\widehat F(\omega)\big),
\end{equation}
where the sum ranges over poles in the open strip $-\nu_2<\Im\omega<-\nu_1$.
\end{proposition}

\begin{proof}
Let $\Gamma_R$ be the positively oriented rectangle with horizontal sides on
$\Im\omega=-\nu_1$ and $\Im\omega=-\nu_2$ and vertical sides at $\Re\omega=\pm R$.
By the residue theorem,
\begin{multline*}
\int_{\Gamma_R} e^{-i\omega t}\,g(\omega)\,R(\omega)\widehat F(\omega)\,d\omega\\
=
2\pi i\sum_{\substack{\omega_j\in\mathrm{Poles}(R)\\ \omega_j\in \mathrm{Int}(\Gamma_R)}}
e^{-i\omega_j t}\,\Res_{\omega=\omega_j}\big(g(\omega)\,R(\omega)\widehat F(\omega)\big).
\end{multline*}
The contribution of the horizontal sides is exactly $2\pi\,(I_{\nu_1}(t)-I_{\nu_2}(t))$.
The contribution of the vertical sides tends to $0$ as $R\to\infty$ by the decay of $\widehat F$ in \eqref{eq:Fhat-decay-app},
the polynomial growth of $g$, and the polynomial bound of $R(\omega)$ on horizontal lines.
Letting $R\to\infty$ yields \eqref{eq:band-isolation-identity-app}.
\end{proof}

\begin{remark}[Relation to Proposition~\ref{prop:overtone-band-subtraction}]
Proposition~\ref{prop:overtone-band-subtraction} in Section~\ref{sec:two-mode} is a specialization of
Proposition~\ref{prop:contour-subtraction-abstract} to the forced Kerr--de~Sitter evolution, with the entire window
$g(\omega)=\widetilde g_{\pm,\ell}(\omega)$ and with $\nu_1,\nu_2$ chosen between consecutive overtone layers.
In particular, one may also take $g\equiv 1$ to isolate the overtone strip purely by contour subtraction; see
Remark~\ref{rem:bandpass-alternative}.
\end{remark}

\subsection{Microlocal sectorization (input from the companion paper)}

The analytic weights and contour subtraction isolate \emph{which poles can contribute}.
To show that only the equatorial poles contribute to the microlocalized signals, we additionally use
the sectorization estimate of Theorem~\ref{thm:input-sectorization} (see Appendix~\ref{app:companion-inputs} and Remark~\ref{rem:companion-sectorization-outline}):
after applying the equatorial cutoff $A_{\pm,h_\ell}$ and the azimuthal projector $\Pi_{k_\pm}$,
all other poles above the chosen contour are suppressed by $\mathcal O(\ell^{-\infty})$.

\section{Companion-paper inputs used in Sections~\ref{sec:two-mode}--\ref{sec:application}}
\label{app:companion-inputs}

Sections~\ref{sec:two-mode} and~\ref{sec:application} rely on three families of high--frequency inputs proved in the companion
paper~\cite{LiPartI}:
(i)~a stable labeling of the equatorial branches together with super--polynomial pseudopole approximation and eventual pole simplicity;
(ii)~a quantitative microlocal \emph{sectorization} statement controlling the generalized Laurent coefficients of the resolvent
after equatorial and azimuthal localization; and
(iii)~a quantitative local inversion estimate for recovering $(M,a)$ from the equatorial pair of QNMs.
The goal of this appendix is not to reproduce the companion proofs in full, but to make the dependence of the present paper explicit:
we indicate precisely where each input appears in~\cite{LiPartI} and we provide a short proof outline for sectorization, since this is
the most technical black box used in the two-mode dominance argument.

Throughout this appendix we work on the compact slow--rotation set $\mathcal K$ fixed in Section~\ref{sec:two-mode} and suppress the
dependence of constants on $\mathcal K$.

\subsection{Equatorial pole labeling, simplicity, and overtone gaps}

Dyatlov's semiclassical quantization of Kerr--de~Sitter QNMs in a fixed-width strip (slow rotation) produces an explicit pseudopole
set and shows that true QNMs are super--polynomially close to it; see~\cite{DyatlovAHP,DyatlovQNM}.
In the form used in~\cite{LiPartI}, this appears as \cite[Theorem~15]{LiPartI} together with a stable labeling reformulation
\cite[Proposition~20]{LiPartI}.
Specializing to the equatorial pair $k=\pm\ell$ and to the first $n+1$ overtones yields exactly the statement recorded in
Theorem~\ref{thm:input-pole-isolation} of the main text, including the disjointness of the equatorial disks and the
super--polynomial proximity estimate \eqref{eq:pole-pseudopole-closeness}.

Simplicity of the labeled poles for $\ell\gg1$ follows from the same barrier--top Grushin reduction that produces the quantization
function: the effective scalar quantization function has a simple zero at each labeled pole, uniformly on compact parameter sets.
This is recorded in~\cite[Remark~21]{LiPartI} and is compatible with the analytic Fredholm framework and the analytic dependence of
isolated simple poles; compare \cite{PetersenVasyAnalyticity,PetersenVasyKerrDeSitterJEMS}.

Finally, the choice of a contour height separating the $n$th and $(n+1)$st equatorial layers (Lemma~\ref{lem:layer-gap} in the main
text) is a consequence of the barrier--top normal form: for fixed sign $\pm$ one has an expansion of the form
\[
-\Im\omega^\sharp_{j,\pm,\ell}(M,a)=(j+\tfrac12)\,\lambda_\pm(M,a)+\mathcal O(\ell^{-1}),
\]
where $\lambda_\pm(M,a)>0$ is the coordinate-time Lyapunov exponent of the corresponding equatorial trapped null orbit; see
\cite[\S5]{LiPartI} and the geometric computation in \cite[Appendix~C]{LiPartI}.
Since $\lambda_\pm$ is continuous and strictly positive on the subextremal set, it has a positive minimum on $\mathcal K$, which gives a
uniform vertical gap between consecutive overtones for $\ell\gg1$.

\subsection{Microlocal sectorization and suppression of non-equatorial poles}

We next indicate why the sectorization estimate in Theorem~\ref{thm:input-sectorization} holds.
The key point is that the residue operators $\Pi_{\omega_j}^{[q]}$ arising from the Laurent expansion of the resolvent are operators on the
\emph{spatial} Sobolev scale $H^s(X)$; the relevant mapping property is therefore $H^{s-1}(X)\to H^{s+1}(X)$, matching the way the forced
transform $\widehat F_\vartheta(\omega)$ enters the residue calculus.

\begin{remark}[Proof outline for Theorem~\ref{thm:input-sectorization}]\label{rem:companion-sectorization-outline}
We sketch the argument underlying \eqref{eq:sectorization}; full details (including parameter-uniform bounds and the Grushin setup) are in
\cite[\S5 and Appendix~B]{LiPartI}.

\smallskip
\noindent\textup{Step 1: $k$--mode reduction and angular semiclassical scaling.}
After applying the azimuthal projector $\Pi_{k_\pm}$ and setting $h=h_\ell=\ell^{-1}$, the stationary equation becomes a semiclassical
problem for a non-selfadjoint family $P_\pm(h,\omega)$ with real principal symbol on $T^*X$.
The microlocal cutoff $A_{\pm,h}$ is chosen so that its semiclassical wavefront set lies in a small neighborhood of the equatorial trapped
set corresponding to the sign $\pm$; the supports for $+$ and $-$ are disjoint by construction.

\smallskip
\noindent\textup{Step 2: microlocal invertibility outside the equatorial disks.}
For $\omega$ in a fixed strip $\{\Im\omega>-\nu\}$, the Hamilton flow has no trapped trajectories in the complement of the equatorial
neighborhood selected by $A_{\pm,h}$.
Using elliptic regularity where $P_\pm(h,\omega)$ is elliptic and propagation of semiclassical singularities where it is of real principal
type \cite{VasyKerrDeSitter}, one constructs a microlocal inverse $E_{\pm}(h,\omega)$ satisfying
\[
A_{\pm,h}\,P_\pm(h,\omega)\,E_{\pm}(h,\omega)=A_{\pm,h}+\mathcal O(h^\infty)
\quad\text{on }H^{s-1}\to H^{s-1},
\]
uniformly for $\omega$ away from the equatorial disks $D_{m,\pm,\ell}$.

\smallskip
\noindent\textup{Step 3: Grushin reduction in the equatorial phase-space channel.}
Near the equatorial trapped set one sets up a Grushin problem for $P_\pm(h,\omega)$ with auxiliary operators $(R_+,R_-)$ so that the
extended operator is microlocally invertible with inverse $\mathcal E_\pm(h,\omega)$.
The effective Hamiltonian $E_{-+}(h,\omega)$ is a finite-dimensional matrix whose determinant is (after normalization) the quantization
function; its zeros are exactly the equatorial poles inside the disks $D_{m,\pm,\ell}$, and it is uniformly invertible outside these disks
with inverse $\mathcal O(h^\infty)$.
Consequently, modulo $\mathcal O(h^\infty)$ remainders, the meromorphic part of the microlocalized resolvent is entirely captured by the
equatorial poles.

\smallskip
\noindent\textup{Step 4: Cauchy formula for Laurent coefficients.}
Let $\omega_j\notin\bigcup_{m\le n}D_{m,\pm,\ell}$ be a pole with $\Im\omega_j>-\nu$.
Choose a small circle $\gamma$ around $\omega_j$ contained in $\{\Im\omega>-\nu\}$ and disjoint from the equatorial disks.
Since $A_{\pm,h}\Pi_{k_\pm}R(\omega)\Pi_{k_\pm}A_{\pm,h}$ is holomorphic on and inside $\gamma$ modulo $\mathcal O(h^\infty)$,
Cauchy's formula gives
\[
\Pi_{\omega_j}^{[q]}
=\frac{1}{2\pi i}\oint_\gamma (\omega-\omega_j)^{q-1} R(\omega)\,d\omega,
\]
and hence $\chi A_{\pm,h}\Pi_{k_\pm}\Pi_{\omega_j}^{[q]}\chi=\mathcal O(h^\infty)$ on $H^{s-1}\to H^{s+1}$ after composing with
$\chi$ and using the radial point/propagation estimates needed to move between microlocal regions.
Converting $h^\infty$ to $\ell^{-N}$ yields \eqref{eq:sectorization}.
\end{remark}

\subsection{Local inversion from a pair of equatorial QNMs}

For completeness, we record where the stability estimate used in Section~\ref{sec:application} is proved in the companion paper.
Theorem~\ref{thm:inv-companion} in the main text is exactly the two--parameter true--QNM inverse theorem
\cite[Theorem~48]{LiPartI}, whose proof combines the pseudopole inverse theorem \cite[Theorem~35]{LiPartI} with the $C^1$ pseudopole--QNM
proximity estimate \cite[Lemma~45]{LiPartI}.

\section{Uniformity with respect to the cosmological constant on compact sets}\label{app:lambda-uniformity}

This appendix records a parameter-uniformity statement used when passing from the fixed-$\Lambda$ analysis in
Sections~\ref{sec:setup2}--\ref{sec:freq-extraction} to the three-parameter bias bound in Section~\ref{subsec:application-3}.
The point is that the constants in the energy estimates, resolvent bounds, microlocal cutoffs, and Grushin reductions can be chosen
uniformly on compact three-parameter sets.

\begin{proposition}[Uniform forward constants on compact $(M,a,\Lambda)$ sets]\label{prop:uniform-in-lambda}
Let $\mathcal K^{(3)}$ be a compact subset of the slow-rotation subextremal Kerr--de~Sitter parameter region.
Then one may choose the auxiliary geometric parameters (buffer size $\delta$, the regular time function $t_*$, and the spatial slice $X_{M,a}$)
uniformly on $\mathcal K^{(3)}$, and all constants in the forward analysis of Sections~\ref{sec:setup2}--\ref{sec:freq-extraction} can be taken uniform
for $(M,a,\Lambda)\in\mathcal K^{(3)}$.
In particular, once a pole-free contour is fixed uniformly on $\mathcal K^{(3)}$, the constants in Theorem~\ref{thm:two-mode-dominance} and Proposition~\ref{prop:freq-error-app}
may be chosen uniformly for $(M,a,\Lambda)\in\mathcal K^{(3)}$.
\end{proposition}

\begin{proof}
We indicate the three inputs where the background parameters enter.

\smallskip
\noindent\emph{(i) Geometric constructions and Sobolev conventions.}
On the subextremal region, the horizon radii $r_e(M,a,\Lambda)$ and $r_c(M,a,\Lambda)$ depend smoothly on the parameters, and the same holds for the metric coefficients in star coordinates.
Since $\mathcal K^{(3)}$ is compact and contained in the subextremal set, the separation $r_c-r_e$ has a positive lower bound on $\mathcal K^{(3)}$.
Thus one may choose a single buffer size $\delta>0$ so that the extended radial interval $(r_e-\delta,r_c+\delta)$ is well-defined for all $(M,a,\Lambda)\in\mathcal K^{(3)}$.
With this choice, the spatial slices and the associated Sobolev norms vary smoothly with the parameters and remain uniformly equivalent on $\mathcal K^{(3)}$.

\smallskip
\noindent\emph{(ii) Fredholm setup and redshift estimates.}
The meromorphic/Fredholm resolvent framework for Kerr--de~Sitter \cite{VasyKerrDeSitter} is stable under smooth perturbations of the metric.
On compact parameter families one can choose the variable-order radial point spaces and the redshift weights with uniform constants; in particular, the energy estimate of Theorem~\ref{thm:wellposed} holds with constants uniform on $\mathcal K^{(3)}$.

\smallskip
\noindent\emph{(iii) High-energy resolvent estimates and the equatorial package.}
Normally hyperbolic trapping persists throughout the subextremal Kerr--de~Sitter family, and the dynamical quantities entering the semiclassical estimate depend continuously on the parameters; see \cite{WunschZworskiNH,PetersenVasyKerrDeSitterJEMS}.
Therefore the semiclassical cutoff resolvent bound in Theorem~\ref{thm:semiclassical-resolvent} can be arranged uniformly for $(M,a,\Lambda)\in\mathcal K^{(3)}$ after fixing a compact energy set $I\Subset\mathbb R\setminus\{0\}$.
The equatorial sectorization and Grushin reductions used later are proved in the companion paper with parameter-uniform estimates on compact sets \cite[\S5 and Appendix~B]{LiPartI}.
Combining these observations gives the claimed uniformity on $\mathcal K^{(3)}$.
\end{proof}

\section{Two-exponential models near coalescence and exceptional points}\label{app:two-exp-ep}

This appendix complements Lemma~\ref{lem:prony-stability} by making explicit the algebraic condition numbers
governing two-frequency extraction.  The goal is not to introduce a new algorithm, but to record in a self-contained
way why any ``two-mode-at-once'' reconstruction becomes ill-conditioned as the nodes approach coalescence (the
exceptional point limit).

\subsection{Hankel determinants and Vandermonde conditioning}

Consider the noiseless two-exponential data
\begin{equation}\label{eq:app-two-exp-noiseless}
y_j = a_1 z_1^j + a_2 z_2^j,\qquad j\in\mathbb N_0,
\end{equation}
with $a_1a_2\neq 0$ and $z_1\neq z_2$.
The basic linear-algebraic object is the $2\times 2$ Hankel matrix
\[
H_0:=\begin{pmatrix}y_0 & y_1\\ y_1 & y_2\end{pmatrix}.
\]
A direct computation gives the factorization
\begin{equation}\label{eq:hankel-factorization}
H_0
=
\begin{pmatrix} 1 & 1\\ z_1 & z_2\end{pmatrix}
\begin{pmatrix} a_1 & 0\\ 0 & a_2\end{pmatrix}
\begin{pmatrix} 1 & z_1\\ 1 & z_2\end{pmatrix},
\end{equation}
so in particular
\begin{equation}\label{eq:hankel-det}
\det H_0 = y_0y_2-y_1^2 = a_1a_2\,(z_1-z_2)^2.
\end{equation}
Thus $\|H_0^{-1}\|$ necessarily blows up at least like $|a_1a_2|^{-1}|z_1-z_2|^{-2}$ as $z_1\to z_2$.
This is the first mechanism behind the loss in Lemma~\ref{lem:prony-stability}.

The second mechanism is root sensitivity: even if one knows the symmetric functions
$s_1=z_1+z_2$ and $s_2=z_1z_2$ accurately, converting them into the individual nodes requires taking roots of
$\lambda^2-s_1\lambda+s_2=0$, and the map $(s_1,s_2)\mapsto (z_1,z_2)$ has derivative blowing up like $|z_1-z_2|^{-1}$.

\subsection{A deterministic conditioning estimate}

\begin{proposition}[Conditioning of the two-node Prony map]\label{prop:prony-conditioning}
Let $y_j$ be the noiseless samples \eqref{eq:app-two-exp-noiseless} and let $\widetilde y_j=y_j+e_j$ for $j=0,1,2,3$, where
$|e_j|\le \eta$.
Assume $a_1a_2\neq 0$ and $z_1\neq z_2$, and set
\[
\Delta_0 := y_0y_2-y_1^2 = a_1a_2\,(z_1-z_2)^2,\qquad R:=\max\{|z_1|,|z_2|\}.
\]
There exists an absolute constant $c_0>0$ (depending only on crude bounds on $R,|a_1|,|a_2|$) such that
if
\[
\eta \le c_0\,|a_1a_2|\,|z_1-z_2|^4,
\]
then the Prony reconstruction described in Lemma~\ref{lem:prony-stability} produces nodes
$\widetilde z_1,\widetilde z_2$ which can be labeled so that
\begin{equation}\label{eq:app-prony-conditioning}
\max_{k=1,2}|\widetilde z_k-z_k|
\ \le\
C(R,a_1,a_2)\,
\frac{\eta}{|a_1a_2|\,|z_1-z_2|^3}.
\end{equation}
In particular, the local Lipschitz constant of the inverse map from four samples to two nodes
blows up at least like $|a_1a_2|^{-1}|z_1-z_2|^{-3}$ as $z_1\to z_2$.
\end{proposition}

\begin{proof}
The statement is a reorganization of the proof of Lemma~\ref{lem:prony-stability}:
\eqref{eq:hankel-det} shows the $|z_1-z_2|^{-2}$ loss when solving the $2\times 2$ Hankel system for $(s_1,s_2)$,
and Rouch\'e's theorem together with the identity
$p(\widetilde z_k)=(\widetilde z_k-z_k)(\widetilde z_k-z_{3-k})$ yields an additional factor $|z_1-z_2|^{-1}$
in passing from $(s_1,s_2)$ to the individual roots.
All constants can be tracked explicitly and are collected into $C(R,a_1,a_2)$.
\end{proof}

\subsection{The exceptional point limit and confluent models}

If $z_1=z_2=:z$ in \eqref{eq:app-two-exp-noiseless}, the model collapses to a \emph{confluent} exponential sum:
one has
\[
y_j = (b_0 + b_1 j)\,z^j,\qquad b_0=a_1+a_2,\quad b_1=a_2-a_1,
\]
which corresponds, in the original time variable, to a mode of the form $(\alpha+\beta t)\,e^{-i\omega t}$.
In this regime the ``two distinct nodes'' inverse problem is ill-posed, and stable recovery requires
confluent Prony-type methods and a different parametrization.
Quantitative stability bounds for confluent Prony systems, and sharp dependence on the separation scale in the near-confluent regime,
are discussed in \cite{BatenkovYomdinSJAM2013}.
For the Kerr--de~Sitter application, this is the analytic reason why near exceptional points one expects strong
parameter sensitivity when attempting to separate nearby modes without prior microlocal isolation.

\end{document}